\documentclass[aps, pra, reprint, superscriptaddress]{revtex4-2}
\usepackage{amsmath}
\usepackage{graphicx}
\usepackage{dcolumn}
\usepackage{bm}
\usepackage{amsfonts, amssymb, amsmath, amsthm, mathrsfs, braket, xcolor, float,physics}
\usepackage[bottom]{footmisc}
\usepackage[section]{placeins}
\usepackage{algcompatible}
\usepackage{natbib}
\usepackage{tikz}
\usepackage[linesnumbered, ruled, vlined]{algorithm2e}
\usepackage{quantikz}
\usepackage{adjustbox}

\usepackage{subcaption}
\usepackage{multirow}
\usepackage{booktabs}
\usepackage[normalem]{ulem}

\usepackage[colorlinks=true, linkcolor=blue, citecolor=blue, urlcolor=blue]{hyperref}
\usepackage[capitalise]{cleveref}

\SetCommentSty{mycommfont}
\SetKwInput{KwInput}{Input}
\SetKwInput{KwOutput}{Output}

\newtheoremstyle{mytheoremstyle}{2pt}{2pt}{\itshape}{}{\bfseries}{.}{2em}{} 
\theoremstyle{mytheoremstyle}

\newcommand{\rrangle}{\rangle\hspace{-0.95mm}\rangle}
\newcommand{\llangle}{\langle\hspace{-0.95mm}\langle}

\newcommand{\kett}[1]{\ket{#1}\hspace{-0.95mm}\rangle}

\theoremstyle{definition}

\theoremstyle{plain}
\newtheorem{theorem}{Theorem}
\newtheorem{lemma}{Lemma}
\newtheorem{proposition}{Proposition}
\newtheorem{corollary}{Corollary}

\newtheorem{claim}{Claim}

\allowdisplaybreaks[4]

\begin{document}

\title{End-to-End Complexity Analysis for Quantum Simulation of the Extended Jaynes-Cummings Models}

\author{Nam Nguyen}
\affiliation{Applied Mathematics, Boeing Research \& Technology, Huntington Beach, CA, 92647, USA}

\author{Michael Yu}
\affiliation{Disruptive Computing and Networks, Boeing Research \& Technology, Huntington Beach, CA, 92647, USA}

\author{Alan Robertson}
\affiliation{Centre for Quantum Software and Information, University of Technology Sydney, Sydney, New South Wales 2007, Australia}

\author{Hiromichi Nishimura}
\affiliation{Applied Mathematics, Boeing Research \& Technology, Huntington Beach, CA, 92647, USA}

\author{Samuel J. Elman}
\affiliation{Centre for Quantum Software and Information, University of Technology Sydney, Sydney, New South Wales 2007, Australia}

\author{Benjamin Koltenbah}
\affiliation{Applied Physics, Boeing Research \& Technology, Seattle, WA, 98108, USA}  

\date{\today}

\begin{abstract}
The extended Jaynes-Cummings model (eJCM) is a foundational framework for describing multi-mode light-matter interactions, with direct applications in quantum technologies such as photon addition and quasi-noiseless amplification. However, the model's complexity makes classical simulation intractable for large systems that could be of practical interest. In this work, we present a comprehensive, end-to-end framework for the quantum simulation of the eJCM. We develop explicit quantum algorithms and circuits for simulating the system's time evolution using first and second-order product formulas, analyzing the dynamics in both the Schrodinger and interaction pictures. Our analysis includes rigorous, closed-form error bounds that guide the choice of simulation parameters, and we extend the methodology to efficiently handle both pure and mixed quantum states. Furthermore, we validate our theoretical cost models with numerical simulations and provide a detailed fault-tolerant resource analysis, compiling the simulation circuits for a surface-code architecture to yield concrete estimates for physical qubit counts and execution times. This work establishes a complete roadmap for simulating the eJCM on future quantum computers.
\end{abstract}

\pacs{Valid PACS appear here} 
\maketitle

\tableofcontents
\setcounter{tocdepth}{1}

\onecolumngrid
\newpage

\begin{figure}[t]
    \includegraphics[width= 1\textwidth]{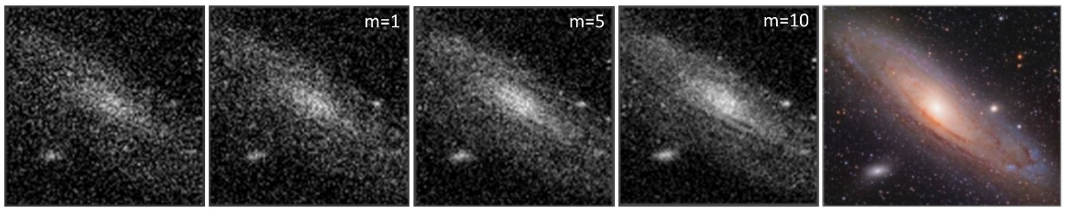}
    \caption{Simulated enhanced imaging with Quasi-Noiseless Amplification utilizes photon addition to improve signal-to-noise ratios of weak signals beyond classical limits. Here $m$ denotes the number of added photons to a grainy image of the Andromeda galaxy.}
  \label{fig: QNA application}
\end{figure}

\section{Introduction}

Many quantum technologies, including quantum sensing, communication, and networking, utilise the light-matter interactions at the mesoscopic scale. As such, understanding the interactions between quantum emitters (such as trapped atoms or ions, solid-state quantum dots, nitrogen vacancies or other quantum materials) and photonic fields is important for designing such technologies.
A foundational model for these interactions is the Jaynes-Cummings model (JCM): a framework that describes the dynamics between a two-level atom and a single mode of the electromagnetic field~\cite{Jaynes1962ComparisonOQ, Loudon1974TheQT, Shore1993TheJM}. While the JCM has been widely applied in the study of quantum systems, many practical applications require more complex models, where multi-level atomic systems interact with many photon modes. These extended systems, described by the extended JCM (eJCM)~\cite{Seke1985ExtendedJM, lahiri2016basic}, are crucial for advanced applications such as single-photon generation (PG) and photon addition (PA), particularly in quasi-noiseless amplification (QNA)~\cite{Akbari2024PhotonSA} and in generating entangled photon states. Photon–emitter interactions also underpin applications that exploit the enhanced photon statistics predicted by models such as the eJCM, including LiDAR sensing, imaging, and on-demand single-photon detection.

In PG, an emitter produces a pulse into field modes that are initially empty or partially occupied. Tracking the temporal and spectral occupation of these modes throughout and after the emission process is essential for characterizing the emitter's performance and assessing its noise resilience \cite{Akbari2024PhotonSA, Vajner2023OnDemandGO}. In such scenarios, additional photon modes may be introduced to capture a wider range of emission directions, frequencies, polarizations, and other degrees of freedom beyond those defining the initial emission process (EP). As these systems grow in complexity, conventional classical approaches face significant challenges because execution times and memory requirements scale exponentially with the number of photon modes and emitter levels \cite{Gard2013InefficiencyOC, Dowling2013ClassicalCV}.

Quantum computing offers a promising alternative, with the potential to simulate these complex quantum systems more efficiently by directly exploiting their quantum mechanical properties that are otherwise difficult to capture classically \cite{Feynman1999SimulatingPW, Lloyd1996UniversalQS, Georgescu2013QuantumS, Fauseweh2024QuantumMS}. In the remainder of this paper, we specialize to the simplest non-trivial setting: a single two-level atom (our emitter) interacting with $N_F$ truncated photon modes. 
We present a detailed, end-to-end framework for quantum simulation of the eJCM. 
We develop an end-to-end quantum simulation framework for the eJCM, including explicit product-formula circuits that realizes the time-evolution operator in both the Schrodinger and interaction pictures, for pure and mixed states alike. The resulting simulations yield key observables, most notably photon-number statistics, crucial for applications such as QNA, where enaces signal-to-noise ratios enable detection of extremely weak optical signals. 
\cref{fig: QNA application} illustrates this capability, showing how QNA sharpens a noisy image of the Andromeda galaxy. Taken together, our results demonstrate that quantum computers can provide scalable, high-precision tools for modeling light–matter interactions central to emerging quantum technologies.

In this paper, we perform a comparative analysis of simulation costs between two different theoretical frameworks: the Schrodinger picture and the Interaction picture. While the Schrodinger picture provides a direct simulation of the full system Hamiltonian, its cost complexity scales with the maximum frequency of the system, $\omega_{\max}$. In contrast, the Interaction picture separates the Hamiltonian into a `free' and an `interaction' part, leading to a simulation cost that scales not with the absolute frequencies but with the maximum detuning between the field modes and the emitter, $\delta_{\max}$. This suggests that for systems where frequencies are high but near-resonant, a possible scenario in physical applications, the Interaction picture can offer a significantly more efficient simulation technique. Throughout this paper, we will rigorously quantify this trade-off, providing detailed cost analyses for both pictures under first and second-order Trotter-Suzuki decompositions to establish a clear guide for selecting the optimal simulation strategy based on the system's physical parameters.

\section{Quantum Simulation of the eJCM in the Schrodinger Picture}
\label{sec:Quantum_Simulation_in_Schrodinger}

In this section, we will present a detailed analysis of the quantum simulation of the eJCM on the quantum computer. The extended Jaynes–Cummings model describes the joint dynamics of \(N_F\) photonic modes and a multi-level emitter. The Hilbert space of each photonic field \(\mathcal{H}^{(j)}_{\rm photon}\) is truncated to \(n+1\) Fock levels and the emitter is described by a Hilbert space of \(N_L\) levels \(\mathcal{H}_{\rm emitter}\).
Thus, the total Hilbert space is the $(n+1)^{N_F}N_L$-dimensional space
\[
  \mathcal{H}
    = \Bigl(\!\bigotimes_{j=1}^{N_F}\mathcal{H}^{(j)}_{\rm photon}\Bigr)
      \otimes \mathcal{H}_{\rm emitter}.
\]
The Hamiltonian is divided into three distinct parts: a free-field term for the non-interacting harmonic oscillators, an emitter Hamiltonian for internal energy levels of the atom, and an interaction Hamiltonian that couples the field modes to emitter transitions.

For applications in photon addition and quasi-noiseless amplification, we are primarily interested in systems describing near resonant, multi-mode photonic fields interacting with a single two-level emitter ($N_L = 2$) with transition frequency $\omega = \omega_e - \omega_g$. Truncating the photonic field modes to only the first $n$ excitations and setting the zero-point energy between the emitter levels, the eJCM Hamiltonian in the Schrodinger picture under the rotating-wave approximation (RWA) has the form:
\begin{equation}
    \label{eq:eJCM_2Level_Schrodinger_Picture}
    \begin{aligned}
        H = &
        \underbrace{\sum_{m=1}^{N_F}  \omega _m \left( \bigotimes_{k=1}^{N_F} D_k^{(m)} \right) \otimes \mathbb{I}_2}_{H_\text{Photon}} 
        + \underbrace{\frac{1}{2} \omega \left( \bigotimes_{m=1}^{N_F} \mathbb{I}_{n+1} \right) \otimes \sigma^z}_{H_\text{Atom}} 
        + \underbrace{ \sum_{m=1}^{N_F} \gamma_{m} \left[
       \left( \bigotimes_{k=1}^{N_F} A_k^{(m)\dagger} \right) \otimes \sigma^- +
        \left( \bigotimes_{k=1}^{N_F} A_k^{(m)} \right) \otimes \sigma^+
        \right]}_{H_\text{Photon-Atom}}
    \end{aligned}
\end{equation}
where \(\sigma^z\) is the standard Pauli $Z$ operator, \(\sigma^{\pm}= (\sigma^{x}\pm i\sigma^{y})/2 \) are the standard raising (lowering) operators and \(\mathbb{I}_d\) is the \(d\)-dimensional identity operator. $\omega_m$ and $\gamma_m$ are the frequencies and coupling strengths (with the emitter) of field mode $m$ respectively. 
Here, and for the remainder of the paper, we adopt natural units with $\hbar = 1$ for simplicity. The operators \(D_k^{(m)}\) and \( A_k^{(m)} \) are defined as
\begin{equation}
    \label{eq:D_k_and_A_k}
    \begin{aligned}
        D_k^{(m)} &=
        \begin{cases}
            \hat{a}^\dagger \hat{a} & \text{if } k = m \\
            \mathbb{I} & \text{otherwise}
        \end{cases}, \quad
        A_k^{(m)} =
        \begin{cases}
            \hat{a} & \text{if } k = m \\
            \mathbb{I} & \text{otherwise},
        \end{cases}
    \end{aligned}
\end{equation}
with truncated creation and annihilation operators of the form:
\begin{equation}
    \label{eq: truncated photon creation and annihilation operator}
    \hat{a} \approx (\hat{a})_n =
    \begin{bmatrix}
        0 & \sqrt{1} & 0 & \cdots & 0 \\
        0 & 0 & \sqrt{2} & \cdots & 0 \\
        0 & 0 & 0 & \ddots & \vdots \\
        \vdots & \vdots & \vdots & \ddots & \sqrt{n} \\
        0 & 0 & 0 & \cdots & 0
    \end{bmatrix}
    \hspace{1cm}
    \hat{a}^\dagger \approx (\hat{a}^\dagger)_n = 
    \begin{bmatrix}
        0 & 0 & 0 & \cdots & 0 \\
        \sqrt{1} & 0 & 0 & \cdots & 0 \\
        0 & \sqrt{2} & 0 & \cdots & 0 \\
        \vdots & \ddots & \ddots & \ddots & \vdots \\
        0 & \cdots & 0 & \sqrt{n}  & 0
    \end{bmatrix}.
\end{equation}
We use \((\hat{a})_n\) and \((\hat{a}^\dagger)_n\) to emphasize that the operators are truncated to a maximum occupation number \(n\). Without loss of generality, we assume \(n = 2^k - 1\) for some integer \(k \geq 1\) throughout this work.

\subsection{Mapping to qubit representation}
\label{sec:Mapping_to_qubit}
To simulate the Hamiltonian described in \cref{eq:eJCM_2Level_Schrodinger_Picture} on a quantum computer, the terms must first be re-expressed in a representation compatible with qubits. We achieve this by mapping the Hamiltonian to a sum of Pauli strings:
\begin{equation}
    H = \sum_{j=1}^M H_j = \sum_{j=1}^M \alpha_j P_j , \quad P_j \in \{I,\sigma^x, \sigma^y, \sigma^z\}^{\otimes N}
\end{equation}
where each term \( H_j \) corresponds to a Pauli string \( P_j \) with real coefficient \( \alpha_j \), that can be efficiently implemented as a quantum circuit. Note that the eJCM Hamiltonian in \cref{eq:eJCM_2Level_Schrodinger_Picture} consists of \( O(N_F) \) terms in its original basis of representation. However, when mapped to the Pauli basis, the number of terms generally increases due to the decomposition of each operator into sums of Pauli strings. The atom term is already in the Pauli representation consisting of a single Pauli string, and hence no additional work is necessary. The truncated photonic terms, however, can be expressed entirely in terms of tensor products of identity and Pauli-\(Z\) operators:
\begin{equation}
    \label{eq: number operator formula}
    (\hat{a}^\dagger)_n (\hat{a})_n = \sum_{i=0}^n i \, |i\rangle \langle i| = \sum_{i=0}^{2^k - 1} \frac{i}{2^k} \bigotimes_{j=0}^{k-1} \left( I + (-1)^{b_j^{(i)}} Z_j \right),
\end{equation}
where \( b_j^{(i)} \in \{0, 1\} \) denotes the \(j\)-th bit in the binary representation of \(i\), with \( j = 0 \) corresponding to the least significant bit. Notice that this is a sparse representation because \cref{eq: number operator formula} contains only single-qubit Pauli-$Z$ operators. All higher-order terms cancel due to symmetry. Therefore, the final operator consists of exactly \(k+1\) single-Pauli operators. Moreover, we can explicitly write the Pauli expansion for the number operator as shown in the following Claim (see the proof in \cref{sec_appendix:qubit_encoding_hamiltonian_formulation} of the Supplementary Materials).

\begin{claim}
\label{lemma: number operator pauli basis}
Let \( (\hat{a}^\dagger)_n (\hat{a})_n \) be the number operator acting on the Fock space with truncation \( n = 2^k - 1 \), and encoded using \( k \) qubits in the binary basis. Then its Pauli decomposition is given by
\begin{equation}
    \label{eq: pauli decomp for number operator}
    (\hat{a}^\dagger)_n (\hat{a})_n = \frac{n}{2} \cdot I^{\otimes k} - \sum_{j=0}^{k-1} \frac{2^{k-j-1}}{2} Z_j.
\end{equation}
\end{claim}

Extending this result to all photonic modes gives 
 \begin{equation}
\label{eq: H_Photon_pauli}
    H_{\text{Photon}} = \left[
    \left( \sum_{m=1}^{N_F} \omega_m \right) \cdot \frac{2^k - 1}{2} \cdot I^{\otimes N_F \cdot k}
    - \sum_{m=1}^{N_F} \sum_{j=0}^{k-1} \omega_m \cdot \frac{2^{k-j-1}}{2} \cdot Z_j^{(m)}
    \right] \otimes \mathbb{I}_2,
\end{equation}
which is the sum of exactly \(N_F\cdot k\) Pauli strings, excluding the identity which can be considered a global shift.

With sparse Pauli representations for the atomic and photonic Hamiltonians, we now turn our attention to the interaction term. To do so, we map the individual creation annihilation operators, $\hat{a}^\dagger$ and $\hat{a}$, into their Pauli representations. As explicated by the following Lemma, there is an efficient Pauli representation for these truncated creation (annihilation) operators.

\begin{lemma}
\label{lemma: pauli count for a and a^dagger}
Let $( \hat{a})_{n}$ or $( \hat{a}^\dagger)_{n}$ act on a truncated $2^k$-dimensional Fock space encoded into $k$ qubits in the computational basis. Then the Pauli decomposition of the operator contains exactly $2^k \cdot k$ distinct Pauli strings.
\end{lemma}

We now analyze the total number of Pauli strings and their structure for $H_{\textrm{Photon-Atom}}$. The following Lemma summarises the count of distinct Pauli strings as a function of the number of photon modes $N_F$ and the truncation parameter $k = \log_2(n+1)$.

\begin{lemma}
\label{lemma:pauli_count_photon_atom_Schrodinger}
Let \(H_{\text{Photon-Atom}}\) be defined as in \cref{eq:eJCM_2Level_Schrodinger_Picture}. Then it can be decomposed into exactly 
\[
N_P \;=\; N_F \, 2^{k}\, k 
\]
number of distinct, non-zero Pauli strings.
\end{lemma}

A detailed proof of~\cref{lemma:pauli_count_photon_atom_Schrodinger} is provided in \cref{sec_appendix:qubit_encoding} of the Supplementary Materials. Note that the classical runtime to generate the Pauli string decomposition of $H_{\mathrm{Photon-Atom}}$ is
\[
\mathcal{O}(2^k \cdot k^2 + N_F \cdot 2^k \cdot k) = \mathcal{O}(2^k \cdot k \cdot (k + N_F)).
\]
The first term accounts for the decomposition of the ladder operators $a$ and $a^\dagger$, as discussed in \Cref{lemma: pauli count for a and a^dagger}. The second term corresponds to constructing and embedding those Pauli strings across all $N_F$ photon modes and attaching the atomic operators $\sigma^\pm$. Note that we never perform a full matrix projection and hence the runtime is minimal.

In many practical scenarios where a small photon cutoff is sufficiently accurate, the factor \(2^k\) will be a modest constant, and hence the runtime of the Pauli decomposition of $H_{\mathrm{Photon-Atom}}$ will be negligible as it scales linearly with $N_F$. For the PA/QNA applications central to this work, we consider a more relevant scaling regime where the photon capacity per mode grows proportionally to the number of modes. This leads to a  runtime that scales roughly quadratic in \(N_F\). 
Overall, the eJCM Hamiltonian can be decomposed into exactly
\begin{equation}
    \label{eq: number of Pauli strings for eJCM}
    M = \underbrace{ \left( N_F \cdot k+1 \right) + 1 }_{\textrm{Z-type}}+ \underbrace{  N_F \, 2^{k}\, k }_{\textrm{non Z-type}}  
\end{equation}
number of Pauli strings.
See \ref{sec_appendix:qubit_encoding_hamiltonian_formulation} of the Supplementary Materials for the pseudo-code.

\subsection{Quantum simulation strategy}
\label{sec:QuantumSimulationSchrodinger}

In this section, we outline the methodology for simulating the eJCM as described in~\cref{eq:eJCM_2Level_Schrodinger_Picture} given an initial pure state. After expressing the Hamiltonian in a sparse Pauli decomposition, \( H = \sum_{j=1}^M \alpha_j P_j, \)
we approximate the time-evolution operator
\[
U(T) = e^{-i H T}
\]
using first and second-order Trotter-Suzuki product formula approximations. We then analyze the associated approximation errors and derive explicit bounds as functions of the simulation parameters.

\subsubsection{Analysis of error bounds}
A first-order approximation (see~\cite{Lloyd1996UniversalQS, Childs2021TheoryOT}) is given by

\begin{equation}
    \label{eq: first order trotter}
    e^{-i T \sum_{j=1}^M \alpha P_j} \approx S_1 = \left( \prod_{j=1}^M e^{-i \alpha_j P_j  T / N_T} \right)^{N_T},
\end{equation}
where the \textit{Trotter number}  \(N_T\) controls the trade-off between accuracy and circuit depth. The corresponding first-order error bound is given by
\begin{equation}
    \label{eq: first order trotter error}
    \left\lVert e^{-i T \sum_{j=1}^M H_j} - S_1 \right\rVert 
    \leq \frac{T^2}{2N_T} \sum_{j=1}^M \left\lVert \sum_{k=j+1}^M [H_k, H_j] \right\rVert.
\end{equation}

Higher-order Trotter formulas can be constructed recursively. Specifically:
\begin{align}
    \label{eq: higher trotter recursive formula}
    S_2(T) &= \left[ \prod_{j=1}^M e^{-i H_j T / (2N_T)} \prod_{j=M}^1 e^{-i H_j T / (2N_T)} \right]^{N_T}, \\
    S_{2r}(T) &= S_{2r-2}(u_r  T)^2 \, S_{2r-2}((1 - 4u_r) T ) \, S_{2r-2}(u_r T)^2,
\end{align}
where \(u_r = (4 - 4^{1/(2r - 1)})^{-1}\). Error bounds for each Trotter order are derived in Ref.~\cite{Childs2021TheoryOT}. In particular, the error from the second-order approximation is bounded by:
\begin{equation}
    \label{eq: second order trotter error}
    \left\lVert e^{-i T\sum_{j=1}^M H_j} - S_2 \right\rVert \leq \frac{T^3}{12N_T^2} \sum_{j=1}^M \left( \left\lVert \sum_{k,l=j+1}^M [H_l, [H_k, H_j]] \right\rVert + \frac{1}{2} \left\lVert \sum_{k=j+1}^M [H_j, [H_j, H_k]] \right\rVert \right).
\end{equation}

We can derive explicit, analytical bounds on the first and second-order Trotter errors as a function of system parameters, including the number of photon modes \( N_F \), the truncation level \(n = 2^k -1 \), and the maximum coupling strength \( \gamma_{\max}  \).
Note that these theoretical error bounds tend to \textit{overestimate} the true simulation error. As a result, they can lead to overly conservative estimates of circuit depth and gate count~\cite{RaeisWiebeQuantumCircuitDesign2012, ChildsToward2018, ZhaoEntanglement2024}. To further investigate these bounds and illustrate their discrepancies, we additionally provide and compare against empirical results for the number of Trotter steps required to accurately simulate the eJCM Hamiltonian.

The Trotter error critically depends on the commutation of the individual Hamiltonian terms. By exploiting the sparse structure of the Pauli decomposition of the eJCM, as well as the fact that the Pauli strings can be grouped into a logarithmic number of mutually commuting subsets (see~\cref{thm:comm_family_HI}), we can greatly reduce the Trotter error.

\begin{theorem}[First-order Trotter error]
\label{thm: commutator norm first order}
Given the eJCM Hamiltonian \( H \) as in \cref{eq:eJCM_2Level_Schrodinger_Picture} with \( N_F \) photon modes of truncation $n = 2^k-1$ and characteristic frequencies bounded by \( \omega_{\max} = \max_m(\abs{\omega_m}, \abs{{\omega}}) \),  and an interaction term partitioned into \(G\) mutually commuting families,  the first-order Trotter approximation \(S_1(t)\) over \(N_T\) steps with a total evolution time \(T\) has an error bounded by
\begin{equation}
\label{eq:first_order_trotter_error_bound}
    \norm{ e^{-i T H} - S_1(T) } \leq \frac{T^2}{N_T} \left[ 
    \frac{\gamma_{\max}^2 \Lambda_k^2 N_F^2 k^2 (G - 1)}{72 G} 
    + \frac{1}{2}N_F \omega_{\max}\gamma_{\max}n^{1/2} 
    + N_F \omega_{\max}\gamma_{\max}n^{3/2} \right],
\end{equation}
where \(\Lambda_k = (2^k + 1)^{3/2} - 1\).
\end{theorem}
As an immediate corollary of \cref{thm: commutator norm first order} we see that in order to achieve a total Trotter error \( \varepsilon^{(1)} \leq \varepsilon \), it suffices to choose
\begin{equation}
    \label{eq:first_order_trotter_bound}
    N_T \geq \frac{T^2}{\varepsilon} \left[ 
    \frac{\gamma_{\max}^2 \Lambda_k^2 N_F^2 k^2 (G - 1)}{72 G} 
    + \frac{1}{2}N_F \omega_{\max}\gamma_{\max}n^{1/2} 
    + N_F \omega_{\max}\gamma_{\max}n^{3/2} \right].
\end{equation}

Extending the same conditions to the second-order Trotter approximation error yields the following Corollary:

\begin{corollary}[Second-order Trotter error]
\label{cor: commutator norm second order}
Under the same assumptions as in \cref{thm: commutator norm first order}, let \( S_2(T) \) denotes the second-order approximation. Let 
\begin{align*}
    C = \gamma_{\max}^3 \, \frac{2 \Lambda_k^3 N_F^3 k^3 (G - 1)(G - 2)}{12 \cdot 324 G^2} 
    + \frac{(1 + 2n) \sqrt{n} \, N_F \, \omega_{\max} \, \gamma_{\max} \left(2 \sqrt{n} \, N_F \, \gamma_{\max} 
    + \frac{1}{2} (1 + 2n N_F) \, \omega_{\max}\right)}{12}.
\end{align*}
Then the second-order Trotter error satisfies 
\begin{align}
    \label{eq:second_order_trotter_bound}
    \epsilon^{(2)} \leq \frac{T^3}{N_T^2} \cdot C.
\end{align}
\end{corollary}

Thus, to guarantee that \( \epsilon^{(2)} \leq \varepsilon \), it suffices to choose 
\[
N_T \geq 
    T^{3/2} \sqrt{\frac{C}{\varepsilon}} 
\]

\subsubsection{Quantum circuit and simulation cost}
\label{sec:quantum_circuit_simulation_cost}

The propagator $e^{-iHT}$ can be simulated by first decomposing $H$ into 1-sparse Hermitian operators, followed by applying a Trotter approximation, as discussed previously. From \cref{sec:Mapping_to_qubit}, we know that the eJCM Hamiltonian in~\cref{eq:eJCM_2Level_Schrodinger_Picture} can be mapped to exactly $N_F \cdot k \cdot (1 + 2^{k}) + 2$ Pauli strings, each representing a 1-sparse operator. The quantum circuit for the first order method takes the form:
\begin{equation}
    S_1(T) = \left( 
        e^{-i \frac{T}{N_T} H_{\text{Photon}}} 
        \cdot e^{-i \frac{T}{N_T} H_{\text{Atom}}} 
        \cdot e^{-i \frac{T}{N_T} H_{\text{Photon-Atom}}} 
    \right)^{N_T}.
\end{equation}

 The photonic and atomic components of the eJCM are sums of single-qubit Pauli-$Z$ operators. Consequently, the evolution operator $\exp(-i T (H_\text{Photon} + H_\text{Atom}))$ reduces to a sequence of single-qubit $R_Z$ rotations on each of the qubits. 

The interaction term, however, consists of non-diagonal Pauli strings.
Its unitary, \( \exp(-i T H_{\textrm{Photon-Atom}} ) \), can then be decomposed into a product of sub-circuits consisting of $R_Z$ rotations conjugated by Clifford circuits. See \cref{sec_appendix:explicit_circuit_construction} in the Supplementary Materials for additional details on the circuit construction. The structure of \( H_{\text{Photon-Atom}} \) allows its Pauli strings to be nicely partitioned into mutually commuting sub-groups  $G_m$. The following Corollary formalizes the number of such commuting groups into which the Pauli strings in \( H_{\text{Photon-Atom}} \) can be partitioned.

\begin{corollary}[Mutually commuting subsets]
    \label{corollary:number_of_commuting_groups_H_photon_atom_Schro}
    Let \( H_{\text{Photon-Atom}} \) denote the photon-atom interaction term in Equation~\eqref{eq:eJCM_2Level_Schrodinger_Picture} with \( N_F \geq 1 \) photon modes, each truncated to level \( n = 2^k - 1 \). Then, its Pauli string decomposition can be partitioned into exactly $2k$ mutually commuting groups.
\end{corollary}

\cref{corollary:number_of_commuting_groups_H_photon_atom_Schro} follows directly from  \cref{thm:comm_family_HI} with additional details and its corresponding proof in \cref{sec_appendix:Hamiltonian_Structure} of the Supplementary Materials. This commutation structure provides flexibility in the ordering of operations within each subset (See \cref{thm: exponential of sum of commuting terms} in the Supplementary Materials) and enables simultaneous diagonalization of the grouped terms, which can be leveraged to further optimize the circuit. In particular, the unitary \( \exp (-i \sum_{P_j \in G} \theta_j P_j) =\exp(-i G)\) can be efficiently implemented using a sequence of commuting diagonal unitaries conjugated with Clifford circuits \( \mathcal{U} \). This commutation property also helps reduce the Trotterisation error that naturally arises from non-commuting terms~\cite{Mukhopadhyay2022SynthesizingEC, Berg2020CircuitOO, Tomesh2021OptimizedQP, Kawase2022FastCS}. \Cref{fig: Time-Evolution First Order Trotter Half Circuit} depicts an example of the overall circuit following the first-order Trotter method.

At each Trotter step, the simulation applies 
\[
N_{Z} = N_F\,k + 1
\] 
single-qubit $R_Z$ rotations followed by 
\[
N_{P} = N_F\,2^k\,k
\] 
multi-qubit Pauli‐string exponentials (\cref{lemma:pauli_count_photon_atom_Schrodinger}).  Define the error‐prefactor
\begin{equation}
    \label{eq:B_prefactor}
    B = \frac{\gamma_{\max}^2\,\Lambda_k^2\,N_F^2\,k^2\,(G - 1)}{72\,G} + \frac{1}{2}\,N_F\,\omega_{\max}\,\gamma_{\max}\,n^{1/2} + N_F\,\omega_{\max}\,\gamma_{\max}\,n^{3/2} 
\end{equation}
as in \cref{thm: commutator norm first order}, then in order to achieve a target error \(\varepsilon\), one needs  
\[
N_T \geq \tfrac{T^2}{\varepsilon}\,B
\]
Trotter steps. Since each step requires 
\((N_{Z}+N_{P})\) non-Clifford $RZ$ rotations, the total complexity in terms of non-Clifford, $RZ$ gates, is:
\[
C^{(1)} 
= (N_{Z}+N_{P})\times N_T
= \tfrac{T^2}{\varepsilon}\bigl[N_F\,k+1 + N_F\,2^k\,k\bigr]\,B.
\]
Explicitly expanding the leading asymptotics gives
\begin{equation}
    \label{eq:first_order_schrodinger_cost}
    C^{(1)} 
    = \mathcal O\!\Bigl(\tfrac{T^2}{\varepsilon}\bigl[
       \gamma_{\max}^2\,N_F^3\,k^3\,2^{4k}
      + \omega_{\max}\,\gamma_{\max}\,N_F^2\,k\, (2^{3k/2} + 2^{5k/2})
    \bigr]\Bigr).
\end{equation}

Let $H_{0} = H_{\text{Photon}} +  H_{\text{Atom}}$ and $H_{1} = H_{\text{Photon-Atom}}$, then the quantum circuit for the second order method takes the form:
\begin{equation}
    S_2(T) = \left( e^{-i \frac{T}{2 N_T} H_0} \cdot e^{-i \frac{T}{N_T} H_1} \cdot e^{-i \frac{T}{2 N_T} H_0} \right)^{N_T}.
\end{equation}
Since $e^{-i(T/N_T)H_{1}}$ is itself realized by a second‐order Trotter over its 
$N_{P}=N_F2^k k$ non‐diagonal Pauli strings, it requires
$2\,N_{P}$ Pauli‐exponentials per Trotter step.  
In total, each Trotter step uses $2(N_{Z}+N_{P})$ non‐Clifford, $RZ$, rotations, 
and the full cost follows by multiplying by $N_T$. By \cref{cor: commutator norm second order}, to reach error \(\varepsilon\) one needs
\(\,N_T \geq   T^{3/2}\sqrt{C/\varepsilon} \) steps, with $C$ defined in the Corollary. 
Hence the total second-order non-Clifford gate complexity is
\[
C^{(2)}
= \frac{T^{3/2}}{\sqrt{\varepsilon}} \, 2\,(N_F\,k+1 + N_F\,2^k k)\; \sqrt{C}.
\]
Expanding the leading-order scalings gives
\begin{equation}
    C^{(2)} = \mathcal{O}\left( 
    \frac{T^{3/2}}{\sqrt{\varepsilon}} 
        \left[ 
        \gamma_{\max}^{3/2} \, 2^{13k/4} \, N_F^{5/2} \, k^{5/2} 
        + \gamma_{\max} \, \omega_{\max}^{1/2} \, 2^{2k} \, N_F^2 \, k 
        + \gamma_{\max}^{1/2} \, \omega_{\max} \, 2^{9k/4} \, N_F^2 \, k 
        \right] 
    \right).
\end{equation}

\begin{figure}[t]
\begin{center}
\begin{adjustbox}{scale=1}
\begin{quantikz}
\lstick{$q_0$}      & \gate[1, style={fill=yellow!10, rounded corners, minimum width=1.6cm, minimum height=0.6cm} ]{RZ^0}\gategroup[4,steps=6,style={dashed,rounded corners,fill=blue!10, inner xsep=4pt},background,label style={label position=below,anchor=north,yshift=-0.3cm}]{{\sc First Trotter Step}}  & \qw & \gate[4, style={fill=purple!10, rounded corners, minimum width=2cm, minimum height=0.6cm} ]{e^{-iG_1 }} & \push{\cdots} & \gate[4, style={fill=purple!10, rounded corners, minimum width=2cm, minimum height=0.6cm} ]{ e^{-iG_W } }  & \qw  & \push{\cdots} \push{\cdots}  & \gate[1, style={fill=yellow!25, rounded corners, minimum width=1.6cm, minimum height=0.6cm} ]{RZ^0}\gategroup[4,steps=6,style={dashed,rounded corners,fill=blue!15, inner xsep=4pt},background,label style={label position=below,anchor=north,yshift=-0.3cm}]{{\sc Last Trotter Step}} & \qw   & \gate[4, style={fill=purple!30, rounded corners, minimum width=2cm, minimum height=0.6cm} ]{e^{-iG_1 } } & \push{\cdots} & \gate[4, style={fill=purple!30, rounded corners, minimum width=2cm, minimum height=0.6cm} ]{ e^{-iG_W }  }  & \qw & \qw\\
\lstick{$q_1$}      & \gate[1, style={fill=yellow!10, rounded corners, minimum width=1.6cm, minimum height=0.6cm} ]{RZ^1}  & \qw  & \ghost[1 ]{P}     & \push{\cdots}  &  \ghost[1 ]{P} & \qw  & \push{\cdots} \push{\cdots} & \gate[1, style={fill=yellow!25, rounded corners, minimum width=1.6cm, minimum height=0.6cm} ]{RZ^1} & \qw  &   \ghost[1 ]{P} & \push{\cdots} &   \ghost[1 ]{P}  & \qw & \qw \\
\lstick{$\vdots$}  \\
\lstick{$q_{N-1}$}  & \gate[1, style={fill=yellow!10, rounded corners, minimum width=1.6cm, minimum height=0.6cm} ]{RZ^{N-1} } & \qw  &   \ghost[1 ]{P}   & \push{\cdots}   & \qw  & \qw & \push{\cdots} \push{\cdots} & \gate[1, style={fill=yellow!25, rounded corners, minimum width=1.6cm, minimum height=0.6cm} ]{RZ^{N-1} } & \qw &   \ghost[1 ]{P} & \push{\cdots} &   \ghost[1 ]{P} & \qw & \qw
\end{quantikz}
\end{adjustbox}
\caption{Quantum circuit to prepare \(S_1(T) = \left( 
        e^{-i \frac{T}{N_T} H_{\text{Photon}}} 
        \cdot e^{-i \frac{T}{N_T} H_{\text{Atom}}} 
        \cdot e^{-i \frac{T}{N_T} H_{\text{Photon-Atom}}} 
    \right)^{N_T}\) using first-order Trotterization. The layer $RZ^i$ represents the $RZ$ rotations resulting from exponentiating single $Z$ operators acting on qubit $i$ in the $H_{\text{Photon}} + H_{\text{Atom}}$ part of the Hamiltonian. The sequence from $e^{-iG_1}$ to $e^{-iG_W}$ approximates $e^{-i \Delta t H_{\textrm{Photon-Atom}}}$ via Trotterisation at time interval $t_j$. Each group $G_i$ contains all Pauli strings from the interaction Hamiltonian ($H_{\textrm{Photon-Atom}}$) that mutually commute, allowing flexibility in the ordering of terms within the implementation of $e^{-iG_i}$. Each dashed box represents a single Trotter step.}
\label{fig: Time-Evolution First Order Trotter Half Circuit}
\end{center}
\end{figure}
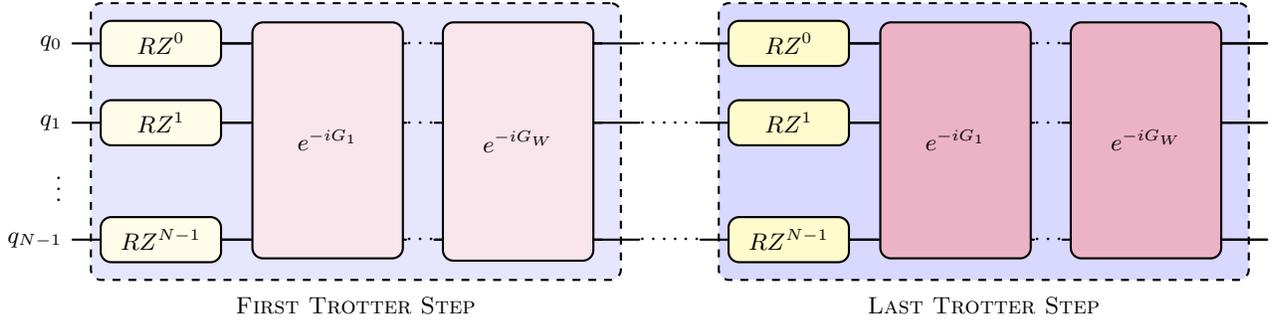

\subsubsection{Numerical validation and performance analysis}

To validate our theoretical analysis and assess the simulation's practical performance, we conducted several numerical experiments. We begin by confirming that our Trotterized circuit converges to well established analytical results for the base JCM.  Specifically, we consider the Jaynes-Cummings (JC) Hamiltonian, which is a special case of the eJCM when \(N_F = 1\). Let \(\Delta = \omega_1 - \omega\) be the detuning between the single photon mode and the emitter, then for an initial state \(\ket{0}_{\text{photon}}\ket{e}_{\text{emitter}}\) the analytical survival probability is  
\begin{equation}
\label{eq:JC_exact_prob}
P_{|0\rangle \otimes |e\rangle}(t) =
1 -
\frac{4g^{2}}{\Delta^{2} + 4g^{2}}
\sin^{2}\left(\frac{\Omega t}{2}\right),
\quad
\Omega = \sqrt{\Delta^{2} + 4g^{2}}
\end{equation}
where \(g\) is the light-matter coupling strength. Note that for an on-resonant system (\(\Delta = 0\)) this reduces to the well-known Rabi oscillations
\(P(t) = \cos^{2}(g t)\).
\cref{fig:JC_demo}(left) confirms the perfect agreement between our simulation and
Eq.~\eqref{eq:JC_exact_prob}, while
\cref{fig:JC_demo}(right) demonstrates first-order convergence of the \(L_2\) error with respect to the Trotter step count \(N_T\) -- in line with the theoretical bound. However, the observed convergence rate is significantly faster than the theoretical bound derived in our analysis, suggesting that the actual simulation cost may be significantly less in practice.

\begin{figure}[ht]
    \centering
    \includegraphics[width=\textwidth]{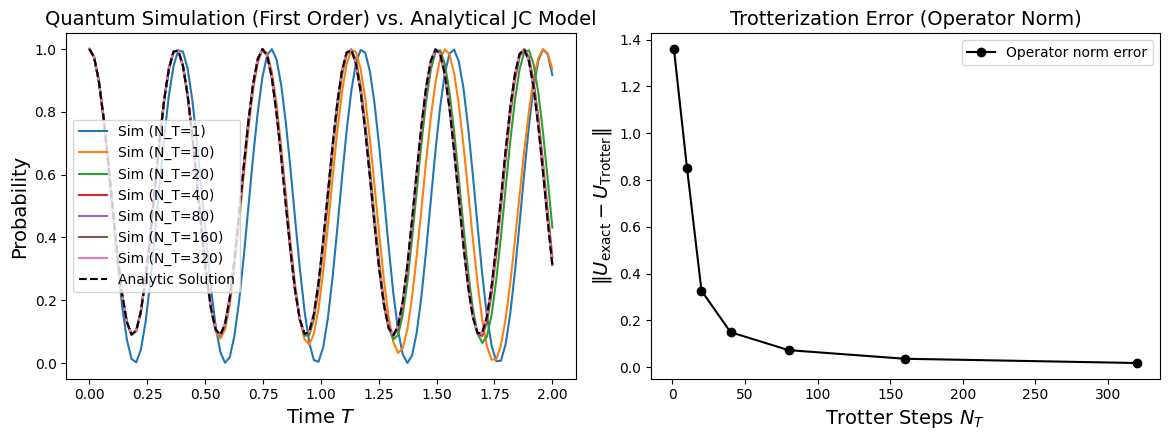}
    \caption{
    (left) Quantum Simulation (First Order) vs. Analytical JC Model. Comparison between first-order quantum simulation and the analytical solution for the Jaynes–Cummings model with initial state \( |0\rangle \otimes |e\rangle \). The y-axis shows the probability of finding the system in the initial state as a function of time. The observed is the expected Rabi oscillations under exact JC dynamics. As the number of time discretization gets larger, we converge to the exact solution.
    (right) Trotterization error in operator norm. Error convergence with respect to number of Trotter step \( N_T \), confirming first-order scaling. 
    }
    \label{fig:JC_demo}
\end{figure}

Equations~\eqref{eq:first_order_trotter_bound} and \eqref{eq:second_order_trotter_bound} provide the worst-case upper bound required to simulate the eJCM Hamiltonian using first and second-order product formulas, respectively. To assess how conservative this bound is in practice, we performed numerical simulations across various parameters. We found that the actual cost is significantly smaller than the theoretically predicted values. 
This overestimation arises in part because the theoretical bound is derived by bounding the commutator norm estimates to bound the Trotter error. However, this approach does not account for potential cancellations between commutator terms, which can substantially reduce the actual error. In particular, as indicated by Equations~\eqref{eq: first order trotter error} and \eqref{eq: second order trotter error}, the Trotter error depends heavily on the specific ordering of terms in the decomposition. 
For this reason, one can often reduce the error further in practice through randomization. In particular,  we apply randomization techniques to the Trotter ordering to suppress error accumulation. That is, at each Trotter step, we randomly permute the order of commuting groups as well as the Pauli strings within them. Since the leading-order Trotter error arises from commutators that depend on this ordering, random permutations help average out contributions, introducing cancellations are otherwise not present in fixed sequences. 
Our numerical simulations demonstrate the properties described above for various photon modes \(N_F\), truncation levels \(k\), and Trotter step counts \(N_T\). See \cref{tab:OpErrSchro_NF=3_k=2} for more details.

\begin{table}[ht]
\centering
\begin{tabular}{c@{\hspace{15pt}}ccccccc}
\toprule
\multirow{2}{*}{$N_T$}
& \multicolumn{3}{c}{First-Order Method}
& \multicolumn{1}{c}{\phantom{abc}} 
& \multicolumn{3}{c}{Second-Order Method} \\
\cmidrule(r){2-4} \cmidrule(l){6-8}
& Theoretical & Fixed Ordering & Randomization
& & Theoretical & Operator Error & Randomization \\
\midrule
1   & $6.353{\times}10^{1}$ & $1.997{\times}10^{0}$ & $1.997{\times}10^{0}$ & & $1.372{\times}10^{2}$ & $1.570{\times}10^{0}$ & $1.570{\times}10^{0}$ \\
2   & $3.176{\times}10^{1}$ & $1.387{\times}10^{0}$ & $1.499{\times}10^{0}$ & & $3.431{\times}10^{1}$ & $5.625{\times}10^{-1}$ & $6.558{\times}10^{-1}$ \\
4   & $1.588{\times}10^{1}$ & $7.553{\times}10^{-1}$ & $6.630{\times}10^{-1}$ & & $8.577{\times}10^{0}$ & $1.416{\times}10^{-1}$ & $1.365{\times}10^{-1}$ \\
8   & $7.941{\times}10^{0}$ & $3.841{\times}10^{-1}$ & $2.640{\times}10^{-1}$ & & $2.144{\times}10^{0}$ & $3.533{\times}10^{-2}$ & $2.869{\times}10^{-2}$ \\
16  & $3.971{\times}10^{0}$ & $1.924{\times}10^{-1}$ & $1.107{\times}10^{-1}$ & & $5.360{\times}10^{-1}$ & $8.824{\times}10^{-3}$ & $5.675{\times}10^{-3}$ \\
32  & $1.985{\times}10^{0}$ & $9.615{\times}10^{-2}$ & $3.182{\times}10^{-2}$ & & $1.340{\times}10^{-1}$ & $2.206{\times}10^{-3}$ & $1.187{\times}10^{-3}$ \\
64  & $9.926{\times}10^{-1}$ & $4.804{\times}10^{-2}$ & $9.094{\times}10^{-3}$ & & $3.350{\times}10^{-2}$ & $5.513{\times}10^{-4}$ & $2.632{\times}10^{-4}$ \\
128 & $4.963{\times}10^{-1}$ & $2.401{\times}10^{-2}$ & $3.256{\times}10^{-3}$ & & $8.376{\times}10^{-3}$ & $1.378{\times}10^{-4}$ & $6.118{\times}10^{-5}$ \\
\bottomrule
\end{tabular}
\caption{
Comparison of theoretical and numerical error metrics for first- and second-order product formula for the quantum simulation of the eJCM Hamiltonian with $N_F=3$, $k=2$ over the time interval $[0,1]$, with coupling parameters are set uniformly as $\gamma_m = \tilde{\omega}_m = 1$, thus $\gamma_{\max} = 1$ and $\omega_{\max} = 1$. The first-order method have an overestimating factor of $\sim 20$ between the theoretical bound and the actual observed error. Similarly, the second-order method have an overestimating factor of $\sim 60$. Note that the way the parameters are configured in this simulation helps tighten the gap between the theoretical error bounds and the observed errors.
}
\label{tab:OpErrSchro_NF=3_k=2}
\end{table}

Moreover, when evaluating the error for time-evolving physically motivated initial states, such as coherent states, the actual state error is often even smaller. This is consistent with the fact that the operator norm captures the worst-case limit over all possible states:
\begin{equation}
    \left\lVert \left( U_{\text{exact}}  - U_{\text{Trotter} } \right) |\psi_0 \rangle   \right\rVert \leq  \left\lVert \left( U_{\text{exact}}  - U_{\text{Trotter} } \right) \right\rVert.
\end{equation}
In particular, this effect is especially pronounced when the intermediate states become highly entangled~\cite{ZhaoEntanglement2024}. This is quantitatively illustrated in \cref{tab:Combined_Op_St_Err_Schro_NF3_k2}, which compares the operator norm and state vector errors for both fixed and randomized Trotter orderings, using both first- and second-order product formulas.

\begin{table}[ht]
\centering
\begin{tabular}{c@{\hspace{8pt}}cccc@{\hspace{8pt}}cccc}
\toprule
\multirow{2}{*}{$N_T$}
& \multicolumn{4}{c}{First-Order Method}
& \multicolumn{4}{c}{Second-Order Method} \\
\cmidrule(r){2-5} \cmidrule(l){6-9}
& Op. Error & State Error & Op. Err (rand) & St. Err (rand)
& Op. Error & State Error & Op. Err (rand) & St. Err (rand) \\
\midrule
1   & $1.794{\times}10^{0}$ & $1.083{\times}10^{0}$ & $1.725{\times}10^{0}$ & $1.083{\times}10^{0}$
    & $1.288{\times}10^{0}$ & $9.181{\times}10^{-1}$ & $1.096{\times}10^{0}$ & $6.057{\times}10^{-1}$ \\
2   & $1.065{\times}10^{0}$ & $5.144{\times}10^{-1}$ & $1.052{\times}10^{0}$ & $6.316{\times}10^{-1}$
    & $3.323{\times}10^{-1}$ & $2.087{\times}10^{-1}$ & $3.177{\times}10^{-1}$ & $1.979{\times}10^{-1}$ \\
4   & $5.484{\times}10^{-1}$ & $2.557{\times}10^{-1}$ & $4.187{\times}10^{-1}$ & $2.367{\times}10^{-1}$
    & $8.345{\times}10^{-2}$ & $4.881{\times}10^{-2}$ & $7.207{\times}10^{-2}$ & $4.455{\times}10^{-2}$ \\
8   & $2.750{\times}10^{-1}$ & $1.278{\times}10^{-1}$ & $1.570{\times}10^{-1}$ & $8.588{\times}10^{-2}$
    & $2.088{\times}10^{-2}$ & $1.198{\times}10^{-2}$ & $1.284{\times}10^{-2}$ & $8.520{\times}10^{-3}$ \\
16  & $1.374{\times}10^{-1}$ & $6.388{\times}10^{-2}$ & $5.039{\times}10^{-2}$ & $3.227{\times}10^{-2}$
    & $5.220{\times}10^{-3}$ & $2.982{\times}10^{-3}$ & $2.610{\times}10^{-3}$ & $1.853{\times}10^{-3}$ \\
32  & $6.861{\times}10^{-2}$ & $3.194{\times}10^{-2}$ & $1.898{\times}10^{-2}$ & $1.046{\times}10^{-2}$
    & $1.305{\times}10^{-3}$ & $7.446{\times}10^{-4}$ & $5.819{\times}10^{-4}$ & $4.334{\times}10^{-4}$ \\
64  & $3.428{\times}10^{-2}$ & $1.597{\times}10^{-2}$ & $6.652{\times}10^{-3}$ & $3.951{\times}10^{-3}$
    & $3.262{\times}10^{-4}$ & $1.861{\times}10^{-4}$ & $1.281{\times}10^{-4}$ & $1.002{\times}10^{-4}$ \\
128 & $1.713{\times}10^{-2}$ & $7.984{\times}10^{-3}$ & $2.484{\times}10^{-3}$ & $1.278{\times}10^{-3}$
    & $8.156{\times}10^{-5}$ & $4.652{\times}10^{-5}$ & $3.046{\times}10^{-5}$ & $2.390{\times}10^{-5}$ \\
\bottomrule
\end{tabular}
\caption{
Comparison of state fidelity errors for fixed and randomized Trotter orderings using first- and second-order product formulas, simulating the eJCM Hamiltonian with $N_F=3$, $k=2$. For each $N_T$, the state error is computed for a physically motivated initial state. The state error is consistently smaller than the corresponding operator norm error, highlighting that the latter represents a worst-case upper bound.
}
\label{tab:Combined_Op_St_Err_Schro_NF3_k2}
\end{table}

\subsection{Initial State Preparation \& Measurements of Observables}
\label{subsec: Extracting Observables}

\textbf{Initial State Preparation:} In simulating light-matter interactions, the choice of initial state is critical for modeling physically relevant scenarios. For the photonic part of the system, coherent states are a natural and common choice. They are central to applications across quantum communication, sensing, and photon generation because they represent the quantum analogue of the classical laser beam with well-defined frequencies and polarizations, and are readily prepared in experimental setups. 
A coherent state, truncated to a finite occupation number \(n\), is expressed as:
\begin{equation}
\label{eq: coherent_state}
    \ket{\alpha_j} = e^{-|\alpha_j|^2/2} \sum_{b=0}^n \frac{\alpha_j^b}{\sqrt{b!}} \ket{b},
\end{equation}
where \(\ket{b}\) denotes the Fock state encoded in binary for compatibility with qubit representations.  The initial quantum state of a multi-mode photonic system is then naturally their tensor product:
\begin{equation}
    \ket{\psi_{\text{photon}}} = \bigotimes_{j=1}^{N_F}\ket{\alpha_j}\,,
\end{equation}
where each \(\alpha_j\in\mathbb{C}\).  Such states, characterized by smoothly varying amplitude profiles, can be efficiently prepared using methods like the Grover–Rudolph algorithm\cite{Grover2002CreatingST} or other advanced state preparation routines, such as the modified Grover-Rudolph algorithm for arbitrary differentiable functions~\cite{Zhang2022QuantumSP}, Quantum Eigenvalue Transformation~\cite{McArdle2022QuantumSP}, black-box state preparation~\cite{Sanders2018BlackBoxQS}, and Matrix Product State encoding~\cite{Holmes2020EfficientQC}.

After initialization, the initial state of the photon–emitter system is simply
\[
\ket{\psi(0)} = \ket{\psi_{\text{photon}}}\otimes\ket{\psi_{\text{emitter}}},
\]
where 
\[
\ket{\psi_{\text{emitter}}} = \sum_{i=0}^{N_L-1} c_i \ket{i}
\]
describes the initial state of the emitter.  This state subsequently evolves according to the system Hamiltonian, H, yielding
\[
\ket{\psi(t)} = e^{-iHt}\ket{\psi(0)}.
\]

Observing the dynamics of interest requires measuring expectation values of relevant Hermitian observables, computed as
\begin{equation}
    \langle \hat{O} \rangle = \bra{\psi(t)}\hat{O}\ket{\psi(t)}.
\end{equation}

\vspace{0.5 cm}
\textbf{Measurements of Observables:} After evolving the system for some time \( t \), we are interested in measuring an observable \( \hat{O} \). The method for computing its expectation value depends on whether the quantum state is pure or mixed. For pure states in particular, measuring $\hat{O}$, in general, involves first decomposing the observable into a weighted sum of Pauli strings, $\hat{O} = \sum_j c_j \hat{P}_j$, then
applying appropriate basis-change gates to rotate each \( \hat{P}_j \) into a diagonal form in the computational basis, and finally estimating \( \langle \hat{P}_j \rangle \) from repeated measurements. This procedure is the standard approach used in the Variational Quantum Eigensolver (VQE) and related variational algorithms.

In many quantum optics simulations, the key outputs are the photon number statistics, as they characterize the quantum stat of the light fields involved. The most fundamental of these are the mean and variance of the photon number, which corresponds to the first and second moments of the photon distribution. To measure these quantities, we require the Pauli decomposition of the photon number operator, $\hat{O}_N = (\hat{a}_j^\dagger)_n(\hat{a}_j)_n$, and the number operator squared, $\hat{O}_N = (\hat{a}_j^\dagger)_n(\hat{a}_j)_n(\hat{a}_j^\dagger)_n(\hat{a}_j)_n$, on the photon modes of the system. Following the analytical decomposition of photon number operator $\hat{O}_N$ in \cref{lemma: number operator pauli basis}, the mean photon number operator $\hat{\bar{N}}$ can then be decomposed as:

\begin{equation}
    \hat{\bar{N}} = \sum_{m=1}^{N_F} \hat{O}_{N}^{(m)} = N_F\cdot\frac{2^k-1}{2}\cdot I^{\otimes N_F\cdot k} - \sum_{m=1}^{N_F}\sum_{j=0}^{k-1}\frac{2^{k-j-1}}{2}\cdot Z_j^{(m)},
\end{equation}
where superscript $m$ denotes operation in the $m$th photon mode subspace and $n = 2^k -1$ for photon cutoff $n$. Notably, the \( (N_F \cdot k + 1) \) Pauli strings are all diagonal in the computational basis ( Z-type Pauli strings). Since these terms commute, they can all be measured simultaneously with a single measurement using projective measurements in the Z-basis. 

Decomposition of the number operator squared $\hat{O}_{N^2}$ follows naturally as it is simply the repeated product of the number operator and has the following explicit form
\begin{equation}
\label{eq:Nsquared_PauliOp}
    \hat{O}_{N^2} = \left( \frac{n^{2}}{4}
    + \sum_{j=0}^{k-1}2^{2(k-j-2)} \right) I^{\otimes k}
    - n\sum_{j=0}^{k-1} \,2^{k-j-2} \ Z_j
    + \sum_{j=1}^{k-1} \sum_{\ell=0}^{j-1} 2^{2k-j-\ell-3} \ Z_j Z_\ell,
\end{equation}

\noindent
where the term $\sum_{j=0}^{k-1}2^{2(k-j-2)}I^{\otimes k}$ stems from $Z_jZ_j = I$. Since the Pauli representation of both $\hat{O}_N$ and $\hat{O}_{N^2}$ are composed entirely of Z-type Pauli strings, it is easy to see that the variance operator

\begin{equation}
    var[\hat{N}] = \sum_{\ell=1}^{N_F} \hat{O}_{N^2}^{(\ell)} - 2\sum_{m=2}^{N_F}\sum_{n=1}^{m-1}\hat{O}_N^{(m)}\otimes \hat{O}_N^{(n)}
\end{equation}

\noindent 
is likewise strictly composed of mutually commuting Z-type Pauli strings and thus can be simultaneously measured.

Due to the fundamental shot-noise limit associated with the sampling process, estimating, say, \( \langle \hat{N} \rangle \) with a target additive error \(\varepsilon\) requires \(N_{\text{meas}}\) measurements that scales as \( N_{\text{meas}} = \mathcal{O}\left(\text{Var}(\hat{N})/\varepsilon^2\right). \)
While the precise variance, \( \text{Var}(\hat{N}) = \langle \hat{N}^2 \rangle - \langle \hat{N} \rangle^2 \), depends on the specific quantum state at the time of measurement, an upper bound for the variance is \( \mathcal{O}((N_F 2^k)^2\). Therefore, a sufficient number of measurements to guarantee precision \(\varepsilon\) scales as 
\begin{equation}
    N_{\text{meas}} = \mathcal{O}\left(\frac{(N_F 2^k)^2}{\varepsilon^2}\right).
\end{equation}

\section{Quantum Simulation of the eJCM in the Interaction Picture}
\label{sec:Quantum_Simulation_in_Interaction}

We now simulate the eJCM of Equation~\eqref{eq:eJCM_2Level_Schrodinger_Picture} in the interaction picture. First we split the total Hamiltonian into ``free'' and ``interaction'' parts. That is, 
\begin{equation}
    H = H_0 + H_1, \quad \textrm{where} \quad H_0 = H_{\textrm{Photon}} + H_{\textrm{Atom}}, \quad H_1 = H_{\textrm{Photon-Atom}}.
\end{equation}
Then interaction picture states and Hamiltonian obey 
\begin{align}
  \frac{\mathrm d}{\mathrm dt}\,\ket{\psi_I(t)} &= -\,i\,H_I(t)\ket{\psi_I(t)},\label{eq:Interaction_Pic_SE}\\[6pt]
  H_I(t) &= e^{iH_0 t}\,H_1\,e^{-iH_0 t}.
\end{align}
Due to the structure of the \cref{eq:eJCM_2Level_Schrodinger_Picture}, a straightforward calculation  (See \cref{lemma:H_interaction} in the Supplementary Information) gives
\begin{equation}
\label{eq:HI_explicit}
     H_I(t) = \sum_{m=1}^{N_F} \gamma_{m} \left[
        e^{i(\tilde{\omega}_m - \omega)t} \left( \bigotimes_{k=1}^{N_F} A_k^{(m)\dagger} \right) \otimes \sigma^- +
        e^{-i(\tilde{\omega}_m - \omega)t} \left( \bigotimes_{k=1}^{N_F} A_k^{(m)} \right) \otimes \sigma^+
        \right],
\end{equation}
with \(D_k^{(m)}\) and \( A_k^{(m)} \) as defined in \cref{eq:D_k_and_A_k}.

Unlike the Schrodinger picture Hamiltonian, \(H_I(t)\) is explicitly time-dependent, so its propagator is the time-ordered exponential
\begin{equation}
    |\psi_I(T) \rangle = \underbrace{\mathcal{T}\exp\left(-i \int_0^T H_I (s)  ds\right) }_{U_{T,0}}|\psi(0)\rangle 
\end{equation}
We now have a time-dependent Hamiltonian simulation problem as opposed to the time-independent case previously in the Schrodinger picture. The time-evolution operator $U_{T,0}$, is a time-ordered exponential, where \( \mathcal{T} \) denotes the time-ordering operator. In general, this time-ordered exponential cannot be written as a simple exponential of an integral unless \(H(t)\) commutes with itself at all times.  
Moreover, note that the time propagator under the Schrodinger picture, $U_S$, is equivalent to a time-ordered exponential followed by a diagonal unitary. In particular, 
\begin{equation*}
    U_S(t) = e^{-iH_0 t} \mathcal{T}\exp\left(-i \int_0^t H_I (s)  ds\right).
\end{equation*}
Hence the simulation now consists of two parts: 1) the simulation of $ e^{-iH_0 t}$, which is a simple diagonal unitary that consists of $N_F \, k + 1$ non-Clifford $RZ$ rotations and 2) the simulation of (now time-dependent) $H_I(t)$, which is structurally identical to \(H_{\text{Photon-Atom}}\) but carries an extra oscillatory phase. 
Note this additional phase does add to the complexity of the Pauli string representation of $H_I(t)$ compared to that of \(H_{\text{Photon-Atom}}\).

\begin{lemma}
\label{lemma:pauli_count_H_I_interaction}
The interaction picture Hamiltonian $H_I(t)$ of the eJCM on a truncated, ${(n +1) = 2^k}$-dimensional Hilbert space, encoded in the computational basis on $k$ qubits, can be decomposed into $N_I$ distinct Pauli strings as 
\begin{equation}
    H_I(t) = \sum_{i=1}^{N_I} \alpha_i(t) P_i  
\end{equation}
with  
\[
N_I =
\begin{cases}
2N_F \cdot R & \text{if } t = 0, \\
(2N_F - M_0) \cdot R & \text{if } t \ne 0,
\end{cases}
\]
where $M_0 = \left| \{ m \in \{1, \dots, N_F\} \mid {\omega} = \omega_m \} \right|$ is the number of photon modes resonant with the atom, and $R = 2^k \, k$.
\end{lemma}

A detailed proof of~\cref{lemma:pauli_count_H_I_interaction} can be found in \cref{sec_appendix:qubit_encoding} of the Supplementary Materials.

\subsection{Quantum Simulation Strategy}

The time-ordered integral  \(U_{t,0}\) can be explicitly expressed as the Dyson series:
\[
U_{t,0} = I - i \int_0^t H(t_1)\,dt_1 
+ (-i)^2 \int_0^t dt_1 \int_0^{t_1} dt_2 \, H(t_1) H(t_2) + \cdots.
\]

Although scalable algorithms based on a truncated Dyson series exist~[CITE], we choose the simpler, ancilla-free Trotter (product-formula) approach.  In particular, we partition the total time $t$ into $L$ equal intervals of duration $\Delta t = t/L$. On each slice, we approximate the time-dependent Hamiltonian $H_I(s)$ as being constant. This approach introduces two distinct sources of error: 1) \textit{Time discretization error}, the error from approximating the true time-ordered evolution with a product of piecewise-constant propagators and 2) \textit{Trotter error}, the error from approximating the exponential of a sum of non-commuting operators ($H_I(t_k)$) as a product of exponentials. We now analyze each source of error in turn.

\subsubsection{Time Discretization Error}
\label{sec:time_discretization_error}
The simplest discretization scheme uses the Hamiltonian's value at the beginning of each interval, leading to the approximation:
\begin{equation}
\label{eq:first-order-TD}
  \mathcal{T}\exp\left(-i \int_0^t H (s)  ds\right)
  \;\approx\;
  \prod_{k=0}^{L-1} e^{-iH(t_k)\,\Delta t},
  \qquad  t_k = k\,\Delta t .
\end{equation}
To analyze the error from this approximation, consider a single interval \([t_j, t_{j+1}]\) and define
\[
U_j = \mathcal{T} \exp\left(-i \int_{t_j}^{t_{j+1}} H(s) ds \right), \quad 
\widetilde{U}_j = e^{-i H(t_j) \Delta t}.
\]
The error of this single step (from Ref.~\cite[Lemma 1]{Berry2019TimedependentHS}) is  
\begin{align}
\left\lVert U_j - \widetilde{U}_j \right\rVert_\infty
&\leq \int_{t_j}^{t_{j+1}} \left\lVert H(s) - H(t_j) \right\rVert_\infty ds. \label{eq: single-interval-error}
\end{align}
where \( \left\lVert \cdot \right\rVert_{\infty} \) (or \( \norm{ \cdot } \) ) denotes the spectral norm, which has nice properties such as the scaling property \(  \left\lVert \alpha A \right\rVert = \| \alpha \| \cdot \norm{A} \), submultiplicative property \( \norm{A B} \leq \norm{A} \cdot \norm{B} \), and the triangle inequality  \(\norm{A + B} \leq \norm{A} + \norm{B} \). Assuming \( H(s) \) is differentiable, we apply a first-order Taylor bound to obtain
\[
\left\lVert H(s) - H(t_j) \right\rVert_\infty \leq (s - t_j) \cdot \left\lVert  H'(t) \right\rVert_{\infty,\infty},
\]
where \( \| \cdot \|_{\infty,\infty} \) denotes the supremum of the spectral norm over \( [0, t] \). Integrating over \( s \in [t_j, t_{j+1}] \), we obtain
\begin{align}
\left\lVert U_j - \widetilde{U}_j \right\rVert_\infty
&\leq \int_{t_j}^{t_{j+1}} (s - t_j) \cdot \left\lVert  H'(t) \right\rVert_{\infty,\infty} ds 
= \frac{1}{2} \Delta t^2 \cdot \left\lVert  H'(t) \right\rVert_{\infty,\infty}. \label{eq: single-interval-bound}
\end{align}
Since the error from each segment are independent and thus adds linearly, the total error of the full approximation is
\begin{align}
\left\lVert 
\mathcal{T} \exp\left(-i \int_0^t H(s) ds \right) 
- \prod_{k=0}^{L-1} e^{-i H(t_k) \Delta t} 
\right\rVert_\infty
\leq \sum_{j=0}^{L-1} \left\lVert U_j - \widetilde{U}_j \right\rVert_\infty 
\leq \frac{1}{2} L \cdot \Delta t^2 \cdot \left\lVert H'(t) \right\rVert_{\infty,\infty} 
= \frac{t^2}{2L} \cdot \left\lVert  H'(t) \right\rVert_{\infty,\infty}.
\end{align}
Therefore, in order to ensure the time discretization error is within some tolerance \(\varepsilon\), it suffices to choose \(L\) such that
\begin{equation}
    L \geq\frac{\lVert  H'(t) \rVert_{\infty, \infty} t^2}{\varepsilon}.
\end{equation}

A more accurate, second-order method uses the midpoint of each interval. This is stated in the following proposition and proved in \cref{sec_appendix:quantum_sim_interaction_picture} of the Supplementary Material.

\begin{proposition}
\label{prop: second order time discretisation error}
For the simulation of a time-dependent and differentiable Hamiltonian \(H(t)\) over a total time \(t\) using \(L\) steps of the second-order midpoint integrator \( \tilde{U}_j = e^{-i \Delta t H(t_j + \Delta t /2)} \), the global error is bounded by 
\begin{equation}
    \left\| \mathcal{T} \exp\left(-i \int_0^t H(s) \, ds \right) - \prod_{j=0}^{L-1} \tilde{U_j} \right\|_{\infty}
\leq \frac{t^3}{L^2} \left( \frac{1}{24} \left\| H''(t) \right\|_{\infty, \infty} + \frac{1}{12} \left\| [H'(t), H(t)] \right\|_{\infty, \infty} \right)
\end{equation}

\end{proposition}

\vspace{0.5 cm}
The error analyses for both first- and second-order time discretization methods depend critically on the norms of the Hamiltonian \(H(t)\), its time derivatives \(H'(t)\) and \(H''(t)\), and the commutator \([H'(t),H(t)]\), as shown above. Explicit bounds for these parameters with the interaction picture eJCM Hamiltonian are provided in the lemma below. Detailed derivations can be found in \cref{sec_appendix:quantum_sim_interaction_picture} of the Supplementary Materials. 

\begin{lemma}
\label{lem:components_for_time_discretisation_error_bounds}
Consider $H_I(t)$ as described in \cref{eq:HI_explicit}, and let $n = 2^k - 1$ be the photon truncation level of each mode. Let $\| \cdot \|_{\infty, \infty}$ denote the supremum over $t \in [0, T]$. The norms required for the time discretization error bounds of $H_I(t)$ are:
\begin{align}
    \| H_I'(t)\|_{\infty, \infty} &\leq 2\sqrt{n} \sum_{m=1}^{N_F} |\gamma_m (\omega_m - \omega)| \\
    \| H_I''(t) \|_{\infty, \infty} &\leq 2\sqrt{n} \sum_{m=1}^{N_F} |\gamma_m| (\omega_m - \omega)^2 \\
    \|[H_I'(t), H_I(t)]\|_{\infty, \infty} &\leq 8n \left( \sum_{m=1}^{N_F} |\gamma_m (\omega_m - \omega)| \right) \left( \sum_{m=1}^{N_F} |\gamma_m| \right)
\end{align}
\end{lemma}

Note that $\max_{m} \left( \omega_m - \omega \right) \leq 2 \pi$.

\subsubsection{Trotter Error}

We now analyze the errors associated with applying Trotter approximations to the time-evolution operator generated by the eJCM in the interaction picture. In particular, we derive explicit, closed-form bounds for first- and second-order Trotter errors as a function of system parameters, including the number of photon modes \( N_F \), the truncation level \(n = 2^k -1 \), and the maximum coupling strength \( \gamma_{\max}  \).

\begin{corollary}
\label{corollary:first_order_trotter_error_interaction}
Consider the eJCM Hamiltonian \( H_I(t) \) of \cref{eq:HI_explicit} with \( N_F \) photon modes truncated at level $n = 2^k -1$ and characteristic frequencies bounded by \( \omega_{\max} = \max_m(\abs{\omega_m}, \abs{\tilde{\omega}}) \)
. Let $G$ denote the number of mutually commuting families that the Pauli string representation of $H_I$ can be partitioned into. Then the first-order Trotter approximation \(S_1(\Delta t)\) over \(N_T\) steps of evolution time \(\Delta t\) has an error bound
\begin{equation}
    \norm{ e^{-i \Delta t H_I(t)} - S_1(\Delta t) } \leq \frac{(\Delta t)^2}{N_T} \left[ 
    \gamma_{\max}^2 \cdot \frac{ \Lambda_k^2 \cdot N_F^2 k^2 (G - 1)}{18 G}  \right],
\end{equation}
where \(\Lambda_k = (2^k + 1)^{3/2} - 1\).
\end{corollary}

Following the results of \cref{corollary:first_order_trotter_error_interaction}, we see that in order to achieve a total Trotter error \( \varepsilon_{\text{Trotter}}^{(1)} \leq \varepsilon \), it suffices to choose
\begin{equation}
    N_T \geq \frac{(\Delta t)^2}{\varepsilon} \left[ 
      \gamma_{\max}^2 \cdot \frac{ \Lambda_k^2 \cdot N_F^2 k^2 (G - 1)}{18 G}  \right].
\end{equation}


\subsection{Quantum Circuit and Simulation Cost}
\label{sec:interaction_circuit_simulation_cost}

The explicit quantum circuit for the interaction-picture propagator can be constructed much as in the Schrodinger picture, with two main differences: 1) we no longer interleave the free evolution \(e^{-iH_0 t}\) at each Trotter step, and 2) the structure of \(H_I(t)\) changes both the total Pauli-string count and the number of mutually commuting groups (see \cref{thm:comm_family_HI}).  \cref{fig:second_order_interaction_circuit} illustrates the second-order scheme for the interaction picture.

\begin{figure}[h!]
\begin{center}
\begin{adjustbox}{scale=0.72}
\begin{quantikz}
\lstick{$q_0$}     & \gate[4, style={fill=purple!10, rounded corners} ]{e^{-i\frac{\Delta t}{2} G_1(t_0)}}\gategroup[4,steps=5,style={dashed,rounded corners,fill=blue!10, inner xsep=4pt},background,label style={label position=below,anchor=north,yshift=-0.2cm}]{{\sc Slice $k=0$: $S_2$ for $H_I(t_0)$}} & \push{\cdots} & \gate[4, style={fill=purple!10, rounded corners} ]{e^{-i\frac{\Delta t}{2} G_W(t_0)}} & \push{\cdots} & \gate[4, style={fill=purple!10, rounded corners} ]{e^{-i\frac{\Delta t}{2} G_1(t_0)}} & \qw & \gate[4, style={fill=purple!20, rounded corners} ]{e^{-i\frac{\Delta t}{2} G_1(t_1)}}\gategroup[4,steps=5,style={dashed,rounded corners,fill=blue!15, inner xsep=4pt},background,label style={label position=below,anchor=north,yshift=-0.2cm}]{{\sc Slice $k=1$: $S_2$ for $H_I(t_1)$}} & \push{\cdots} & \gate[4, style={fill=purple!20, rounded corners} ]{e^{-i\frac{\Delta t}{2} G_W(t_1)}} & \push{\cdots} & \gate[4, style={fill=purple!20, rounded corners} ]{e^{-i\frac{\Delta t}{2} G_1(t_1)}} & \push{\cdots} & \gate[1, style={fill=yellow!25, rounded corners} ]{RZ^0}\gategroup[4,steps=1,style={dashed,rounded corners,fill=green!10, inner xsep=4pt},background,label style={label position=below,anchor=north,yshift=-0.2cm}]{{\sc $e^{-iH_0 T}$}} & \qw \\
\lstick{$q_1$}     & \ghost{e^{-i\frac{\Delta t}{2} G_1(t_0)}} & \push{\cdots} & \ghost{e^{-i\frac{\Delta t}{2} G_W(t_0)}} & \push{\cdots} & \ghost{e^{-i\frac{\Delta t}{2} G_1(t_0)}} & \qw & \ghost{e^{-i\frac{\Delta t}{2} G_1(t_1)}} & \push{\cdots} & \ghost{e^{-i\frac{\Delta t}{2} G_W(t_1)}} & \push{\cdots} & \ghost{e^{-i\frac{\Delta t}{2} G_1(t_1)}} & \push{\cdots} & \gate[1, style={fill=yellow!25, rounded corners} ]{RZ^1} & \qw \\
\lstick{$\vdots$}  \\
\lstick{$q_{N-1}$} & \ghost{e^{-i\frac{\Delta t}{2} G_1(t_0)}} & \push{\cdots} & \ghost{e^{-i\frac{\Delta t}{2} G_W(t_0)}} & \push{\cdots} & \ghost{e^{-i\frac{\Delta t}{2} G_1(t_0)}} & \qw & \ghost{e^{-i\frac{\Delta t}{2} G_1(t_1)}} & \push{\cdots} & \ghost{e^{-i\frac{\Delta t}{2} G_W(t_1)}} & \push{\cdots} & \ghost{e^{-i\frac{\Delta t}{2} G_1(t_1)}} & \push{\cdots} & \gate[1, style={fill=yellow!25, rounded corners} ]{RZ^{N-1}} & \qw
\end{quantikz}
\end{adjustbox}
\caption{Quantum circuit for the second-order Trotter–Suzuki scheme in the interaction picture.  Each blue dashed box implements  
  \( S_2(H_I(t_k),\Delta t)
     =\prod_{g=1}^{G}e^{-i\frac{\Delta t}{2}G_g(t_k)}
      \;\prod_{g=G}^{1}e^{-i\frac{\Delta t}{2}G_g(t_k)}\)
  on slice \(k\), and the green dashed box applies the diagonal free evolution \(e^{-iH_0T}\) with \(N_Fk+1\) single-qubit \(R_Z\) rotations.}
\label{fig:second_order_interaction_circuit}
\end{center}
\end{figure}
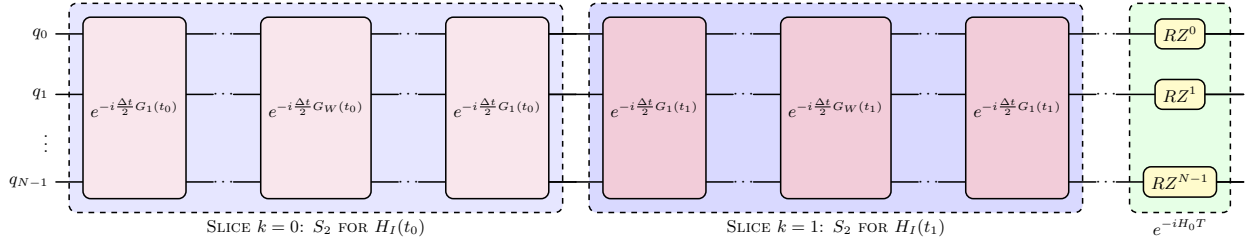

\begin{theorem}[Mutually commuting subsets for $H_I$]
\label{thm:comm_family_HI}
Let
\[
H_I(t)=\sum_{m=1}^{N_F}\gamma_m(t)
       \Bigl[
         e^{+i\delta_m t}\,
            \bigl(\textstyle\bigotimes_{k}A_k^{(m)\dagger}\bigr)\!\otimes\!\sigma^{-}
       + e^{-i\delta_m t}\,
            \bigl(\textstyle\bigotimes_{k}A_k^{(m)}      \bigr)\!\otimes\!\sigma^{+}
       \Bigr],\qquad
\delta_m=\tilde{\omega}_m-\omega ,
\]
with \(N_F\ge1\) photon modes, \(n=2^{k}-1\) and \(M_0=0\). Then the Pauli–string decomposition of \(H_I(t)\) can be partitioned into
\[
\begin{cases}
2k & \text{if } N_F=1,\\[2pt]
4k & \text{if } N_F\ge2 ,
\end{cases}
\]
mutually commuting groups, for all \(t>0\).
\end{theorem}

Note that when $N_F=1$, this reduces to $2k$ groups, and for \(t = 0\), the number of commuting groups is halved. \cref{fig: grouping_demo_NF1_k3} depicts a graphical representation for the commutation structure of the Pauli constituents of \( H_{\textrm{Photon-Atom}}(t_i) \). In particular, for \( N_F = 1 \) and truncation \( n = 7 \) ( \(k=3\) ), the Pauli decomposition can be partitioned into six mutually commuting groups. In \cref{fig: grouping_demo_NF1_k3}, the vertices represent Pauli terms while the edges connect anti-commuting terms -- this is known in literature as a \textit{frustration graph}~\cite{Chapman2020,*elman2021,*mann2024,*chapman2023}.

\begin{figure}[t]
    \centering
    \includegraphics[width= .98\textwidth]{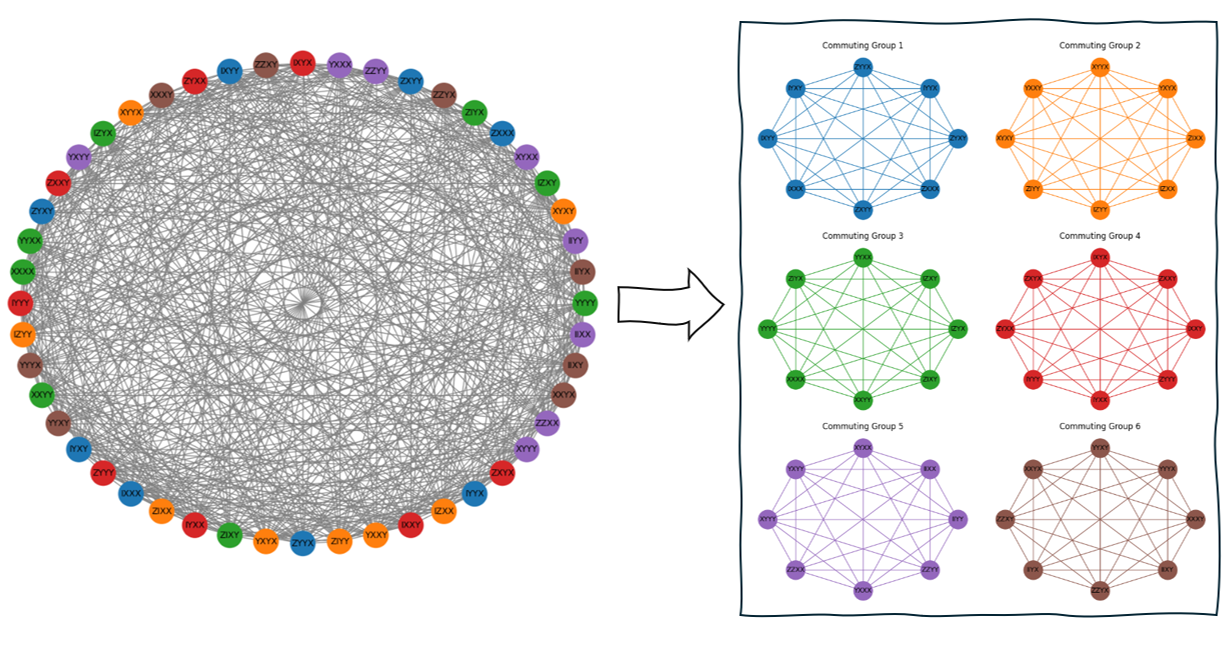}
    \caption{Visualization of the Pauli string partitioning in $H_{\textrm{Photon-Atom}}(t_i)$ for $N_F = 1$ and photon cutoff $n = 7$ ($k = 3$). 
    Left: anti-commutation graph where edges connect non-commuting Pauli strings. Colors indicate commuting group membership. 
    Right: commuting subsets are collected, corresponding to the six mutually commuting subsets, each forming an independent set in the frustration graph (drawn here as colored cliques).}
  \label{fig: grouping_demo_NF1_k3}
\end{figure}

\subsubsection{First-order simulation cost}
\label{subsec:1st_order_interaction}
We now estimate the simulation cost required to simulate the eJCM in the interaction picture up to total error \( \varepsilon \). Our approach combines a first-order time slicing to handle the time dependency of \(H_I(t)\) with first-order Trotter decomposition to handle the non-commuting terms within \(H_I(t_k)\) at each time slice. 
The total simulation error is thus
\[
\varepsilon = \varepsilon_{\text{Trotter}} + \varepsilon_{\text{TimeSlice}},
\]
where \( \varepsilon_{\text{Trotter}} \) captures the total error due to the first-order Trotterization of \( H_I \) for all time intervals \( \Delta t \), and \( \varepsilon_{\text{TimeSlice}} \) captures the error from approximating the time-ordered evolution \( \mathcal{T} \exp\left(-i \int_0^t H_I(s)\,ds\right) \) as a product of unitaries with constant Hamiltonians at each time interval \( \Delta t \). 
We introduce a parameter \(x \in (0,1)\) to flexibly allocate the error as 
\begin{equation}
   \varepsilon = \varepsilon_{\text{Trotter}} + \varepsilon_{\text{TimeSlice}} =  (1-x)\varepsilon +  x\varepsilon.
\end{equation}
Our end goal is to choose $x$ to minimize the total simulation cost. 

As shown in \cref{sec:time_discretization_error}, the first order time-ordering error is bounded by
\[
\left\| \mathcal{T} \exp\left(-i \int_0^t H_I(s)\,ds\right) - \prod_{j=0}^{L-1} e^{-i \Delta t H_I(t_j)} \right\|_{\infty} \leq \frac{t^2}{2L} \cdot \| H_I'(t) \|_{\infty,\infty} \leq \frac{t^2}{2L} \cdot A,
\]
where
\[
A \;=\; 2\sqrt{n} \sum_{m=1}^{N_F} |\gamma_m (\omega_m - \omega)| \leq 2 \sqrt{n} N_F \gamma_{\max} \delta_{\max},
\quad
\Delta t = \frac{t}{L}.
\]
We can ensure that the time discretization error is at most \( x\varepsilon \) by choosing \(L\) such that 
\begin{equation}
\frac{A t^2}{2L} \leq x\varepsilon \quad \Rightarrow \quad L(x) \geq \frac{A t^2}{2x\varepsilon}.
\end{equation}

Each of the \( L \) segments are then simulated using \( N_T \) steps of the first-order Trotter approximation. Let 
\[
B = \gamma_{\max}^2 \cdot \frac{ \Lambda_k^2 \cdot N_F^2 k^2 (G - 1)}{18 G},
\]
then by Corollary~\ref{corollary:first_order_trotter_error_interaction}, the error from each time slice \(\Delta t = t/L\) is bounded by 
\begin{equation}
    \varepsilon_{\text{Trotter-per-slice}} \leq \frac{ (\Delta t)^2}{N_T} \cdot B \leq \frac{(1-x) \varepsilon}{L},
\end{equation}
where the second inequality follows from the fact that the total Trotter error, \(\varepsilon_{\text{Trotter}} = (1-x)\varepsilon\), is equally distributed over \(L\) slices. Thus, the number of Trotter steps, \(N_T(L,x)\), satisfies:
\begin{equation}
    N_T(L,x) \geq \frac{ (\Delta t)^2 B K }{(1-x) \varepsilon} = \frac{ Bt^2}{L(1-x)\varepsilon}
\end{equation}

Combining the results from above, the total number of primitive exponential operations for first-order simulation of the eJCH under the interaction picture is then
\begin{align}
\label{eq: C_total^1_wrt_x}
C_{\text{int}}^{(1)} (x) &= L(x) \cdot N_T(L(x),x) \cdot N_I \nonumber \\
&= \left(\frac{A t^2}{2x\varepsilon}\right) \cdot \left(\frac{B t^2}{\left(\frac{A t^2}{2x\varepsilon}\right)(1-x)\varepsilon}\right) \cdot N_I \nonumber  \\
&= \left(\frac{A t^2}{2x\varepsilon}\right) \cdot \left(\frac{2xB}{A(1-x)}\right) \cdot N_I = \frac{B t^2 N_I}{(1-x)\varepsilon}.
\end{align}
where \( N_I \leq   2\,N_F\,2^k\,k \) is the number of Pauli strings in the decomposition of \(H_I\), and is bounded by \cref{lemma:pauli_count_H_I_interaction}. The cost function \(C_{\text{int}}^{(1)} (x)\) in \cref{eq: C_total^1_wrt_x} is optimized when \(x \rightarrow 0\). However, with the constraint that \(N_T \geq 1\), this implies \
\begin{equation}
    \frac{2xB}{A(1-x)} \geq  1 \quad \Rightarrow (A+2B)x \geq A 
\end{equation}
such that the smallest possible \(x\), denoted \(x_{\text{opt}}\), is
\begin{equation}
    x_{\text{opt}} = \frac{A}{A+2B},
\end{equation}

Substituting $x_{\text{opt}}$ into \cref{eq: C_total^1_wrt_x},
\begin{align*}
C_{\text{int}}^{(1)} (x_{\text{opt}}) &= L(x_{\text{opt}}) \cdot N_T(L(x_{\text{opt}}),x_{\text{opt}}) \cdot N_I \\
&= \left( \frac{(A+2B)t^2}{2\varepsilon} \right) \cdot  1 \cdot N_I =  \frac{t^2}{2\varepsilon} (A+2B) N_I,
\end{align*}
and simplifying, the optimal upper-bound for the total simulation cost using first order methods is given by
\begin{equation}
\label{eq:C_total_1st_optimal}
    C_{\text{int}}^{(1)} = \frac{t^2}{\varepsilon}\,\Bigl(\sqrt{n} N_F \gamma_{\max} \delta_{\max} \;+\; \gamma_{\max}^2 \cdot \frac{ \Lambda_k^2 \cdot N_F^2 k^2 (G - 1)}{18 G} \Bigr)\,N_I,
\end{equation}
with complexity
\begin{equation}
\label{eq: total queries cost for first order}
C_{\text{int}}^{(1)} = \mathcal{O} \left(
    \frac{t^2}{\varepsilon}
    \left[
        \gamma_{\max}^2 \, N_F^3 \, k^3 \, 2^{4k}
        + \gamma_{\max} \, \delta_{\max} \, N_F^2 \, k \, 2^{3k/2}
    \right]
\right).
\end{equation}

The key advantage of simulating in the interaction picture is revealed by this cost analysis. Notably, the complexity does not scale with the maximum absolute frequency $\omega_{\max}$, as it would in the Schrodinger picture, but rather with the maximum detuning $\delta_{\max}$. Therefore, for systems where the frequencies of the photon field modes are high but near-resonant to the transition frequency of the emitter, simulating under the interaction picture will be significantly more efficient.

\subsubsection{Second-order simulation cost}
\label{subsec:2nd_order_interaction}

 While the first-order Trotter approximation offers a conceptually simple and structurally compact approach to simulating time evolution under \( H_I(t) \), its total simulation cost scales as \( \mathcal{O}(t^2 / \varepsilon) \). This quadratic dependence on simulation time and inverse  dependence on error can become prohibitive for long-time or high-precision simulations. Higher-order methods can help reduce the number of required Trotter steps to achieve a fixed error \( \varepsilon \), at the cost of increased circuit complexity per step. They are particularly advantageous for long simulation times t or low error tolerances  \( \varepsilon \), as they require asymptotically fewer Trotter steps. 

A naive approach to improving performance might be to maintain first order time discretization of $H_I(t)$ while utilizing second order Trotter approximations at each step. The second-order Trotter approximation significantly reduces the errors for a given number of Trotter steps \( N_T \) per slice (local error scaling as \( \mathcal{O}((\Delta t)^3/N_T^2) \)). This allows \( N_T \) to be very small given sufficiently small \( \Delta t \). 
However, the overall simulation accuracy of this approach remains limited by the error from the first-order time-slicing, as \( \varepsilon_{\text{TimeSlice}} \sim \mathcal{O}(t^2/L) \). Since \( L \) must scale as \( \mathcal{O}(t^2/\varepsilon) \) to achieve tolerance \( \varepsilon_{\text{TimeSlice}} \), the total simulation cost \( C_{\text{int}}^2 = L \cdot N_T \cdot M_{S2} \), where \(M_{S2}\) is the number of primitive Pauli string exponentials in one Trotter step of \(S_2\), will still scale as \( \mathcal{O}(t^2/\varepsilon) \), thus offering no asymptotic improvement on the dependence of \( t \) or \( \varepsilon \).

To achieve true second-order scaling in \( \varepsilon \), we must employ second-order methods for \textit{both} Trotterization and time discretization. This can be done using a second-order symmetric integrator such as the midpoint method which approximates the evolution over \( [t_j, t_{j+1}] \) as \( e^{-i H(t_j+\Delta t/2)\Delta t} \)  as discussed in Proposition \ref{prop: second order time discretisation error}. As in the first order case, the total simulation error from both approximations is \( \varepsilon = \varepsilon_{\text{TimeSlice}} + \varepsilon_{\text{Trotter}} \). A parameter \(x \in (0,1)\) is introduced, as before, to flexibly allocate the error. 
\begin{equation}
   \varepsilon_{\text{TimeSlice}} = x\varepsilon \quad \text{and} \quad \varepsilon_{\text{Trotter}} = (1-x)\varepsilon.
\end{equation}

To optimize the total simulation error, we follow a similar procedure as in \cref{subsec:1st_order_interaction} and begin by first determining the number of time slices $L(x)$ followed by the number of Trotter steps per time slice $N_T(L,x)$. From the results of Proposition \ref{prop: second order time discretisation error} and defining $K$ such that

\begin{align}
    K &= \left( \frac{1}{24} \left\| H''(t) \right\|_{\infty, \infty} + \frac{1}{12} \left\| [H'(t), H(t)] \right\|_{\infty, \infty} \right)  \notag \\ 
    &\leq  \frac{1}{24} \cdot 2\sqrt{n} \sum_{m=1}^{N_F} |\gamma_m| (\omega_m - \omega)^2 
    + \frac{1}{12} \cdot 8n \left( \sum_{m=1}^{N_F} |\gamma_m (\omega_m - \omega)| \right) \left( \sum_{m=1}^{N_F} |\gamma_m| \right) \notag \\
    &\leq \frac{ \sqrt{n} \, N_F \, \gamma_{\max} \, \delta_{\max}^2 }{12}
    + \frac{ 2n \, N_F^2 \, \gamma_{\max}^2 \, \delta_{\max} }{3}
    = K_M, \notag
\end{align}

\noindent
where the first inequality follows directly from \cref{lem:components_for_time_discretisation_error_bounds}, the minimum number of time slices \(L(x)\) to ensure the time-discretization error is at most \(x\varepsilon\) is:
\begin{equation}
\label{eq: L_x_2nd}
L(x) = \sqrt{\frac{K_M}{x}} \frac{t^{3/2}}{\varepsilon^{1/2}}.
\end{equation}

Each of the \( L \) segments (of duration \(\Delta t = t/L\)) uses \(N_T\) repetitions of an elementary, symmetric second-order Trotter sequence \(S_2(\delta\tau)\), where \(\delta\tau = \Delta t / N_T\). From \cref{cor: commutator norm second order}, the error for one slice \(\Delta t\) using \(N_T\) such repetitions is bounded by \(\frac{(\Delta t)^3}{N_T^2} C\) where \(C\) is the constant defined in said corollary. The total Trotter error over \(L\) slices is then \(L \cdot \frac{(\Delta t)^3 C }{N_T^2} = \frac{C t^3}{L^2 N_T^2}\) which, recall, must be less than or equal to \( (1-x)\varepsilon\). Thus, the minimum number of Trotter repetitions per slice is:
\begin{equation}
\label{eq: NT_Lx_2nd}
N_T(L,x) = \frac{\sqrt{C } t^{3/2}}{L \sqrt{(1-x)\varepsilon}}.
\end{equation}

\noindent
Choosing $L$ from the result of \cref{eq: L_x_2nd}, this becomes
\[ 
N_T(x) = \frac{\sqrt{C } t^{3/2}}{\left(\sqrt{\frac{K_M}{x}} \frac{t^{3/2}}{\varepsilon^{1/2}}\right) \sqrt{(1-x)\varepsilon}} = \frac{\sqrt{x C}}{\sqrt{K_M(1-x)}}. 
\]

Combining the results from above, the total simulation cost of second-order time discretization and Trotter approximation approaches under the interaction picture eJCH as a function of the error allocation parameter $x$ is then

\[
\begin{aligned}
 C_{\text{int}}^{(2)}(x) &= L(x) \cdot N_T(L(x),x) \cdot M_{S2}\\
 &= \left(\sqrt{\frac{K_M}{x}} \frac{t^{3/2}}{\varepsilon^{1/2}}\right) \cdot \left(\frac{\sqrt{x C  }}{\sqrt{K_M(1-x)}}\right) \cdot M_{S2} = \frac{\sqrt{C } M_{S2}}{\sqrt{1-x}} \frac{t^{3/2}}{\varepsilon^{1/2}},
\end{aligned}
\]

\noindent
where \(M_{S2} =2N_I-1\) and \( N_I \) is the number of Pauli strings in the decomposition of \(H_I\) bounded by \cref{lemma:pauli_count_H_I_interaction}. Note that the cost function is minimized when \(x\) is minimized with the constraint that \(N_T(x) \geq 1\). Therefore,
\[ 
\frac{\sqrt{x C }}{\sqrt{K_M(1-x)}} \ge 1 \implies x C \ge K_M(1-x) \implies (K_M + C)x \ge K_M, 
\]
which implies the smallest possible \(x\) is \(x_{\text{opt}} = \frac{K_M}{K_M + C }\).
At this \(x_{\text{opt}}\), the number of Trotter steps per slice becomes 1, the optimal number of time slices becomes
\[ 
L(x_{\text{opt}}) = \sqrt{\frac{K_M}{K_M/(K_M+C )}} \frac{t^{3/2}}{\varepsilon^{1/2}} = \sqrt{K_M+C} \frac{t^{3/2}}{\varepsilon^{1/2}}, 
\]
and the optimal total simulation cost is

\begin{align}
\label{eq:C_total_2nd_optimal}
 C_{\text{int}}^{(2)} 
&= \frac{t^{3/2}}{\varepsilon^{1/2}} \sqrt{K_M+C}  \cdot M_{S2}. 
\end{align}
By bounding \(M_{S2}\) via \cref{lemma:pauli_count_H_I_interaction} and substituting the defined expressions for \(K_M\) and \(C\), the total cost complexity of the second-order method is given by:
\begin{equation}
\label{eq:C_total_2nd_complexity}
C_{\text{int}}^{(2)} = \mathcal{O}\left( 
    \frac{t^{3/2}}{\varepsilon^{1/2}} 
    \left[
        \gamma_{\max}^{3/2} \, N_F^{5/2} \, k^{5/2} \, 2^{13k/4} 
        + \gamma_{\max} \, \delta_{\max}^{1/2} \, N_F^2 \, k \, 2^{3k/2} 
        + \gamma_{\max} \, \delta_{\max} \, N_F^{3/2} \, k \, 2^{5k/4}
    \right]
\right)
\end{equation}

Compared to the optimized first-order cost of \cref{eq: total queries cost for first order}, we see a reduction in both time and error complexities as \(t^2/\varepsilon \to t^{3/2}/\varepsilon^{1/2}\).  Moreover, the dominant system-size scaling reduces from  \(\gamma_{\max}^{2}N_F^{3}k^{3}2^{4k}\)  to  \(\gamma_{\max}^{3/2}N_F^{5/2}k^{5/2}2^{13k/4}\). The two subleading terms in \eqref{eq:C_total_2nd_complexity}, scaling as  
\(\gamma_{\max}\,\delta_{\max}^{1/2}N_F^2k^{3/2}2^{3k/2}\)  
and  
\(\gamma_{\max}^{1/2}\,\delta_{\max}N_F^2k^{3/2}2^{9k/4}\),  
captures how detuning fluctuations \(\delta_{\max}=|\omega_m-\omega|\) feed into the cost.  
Overall, the second-order method achieves a better time and error scaling as well as a genuinely lower asymptotic dependence on \(N_F\) and \(k\) in the ideal, on-resonance (small \(\delta_{\max}\)) limit. To visualize the practical implications of these different scaling regimes, we compare the theoretical cost of the first and second-order methods in \cref{fig: resource_efficiency_1st_vs_2nd} for a hypothetical system consisting of 10 photon modes each with a photon cutoff of $15$ ($k=4$). For simulations requiring low precision (large \(\varepsilon\)) or simulations over short time scales, the first-order method is generally more efficient. Conversely, for high precision or long simulation times, the better asymptotic scalings of the second-order method becomes the dominant factor. 

\begin{figure}[h!]
    \centering
    \includegraphics[width= .6\textwidth]{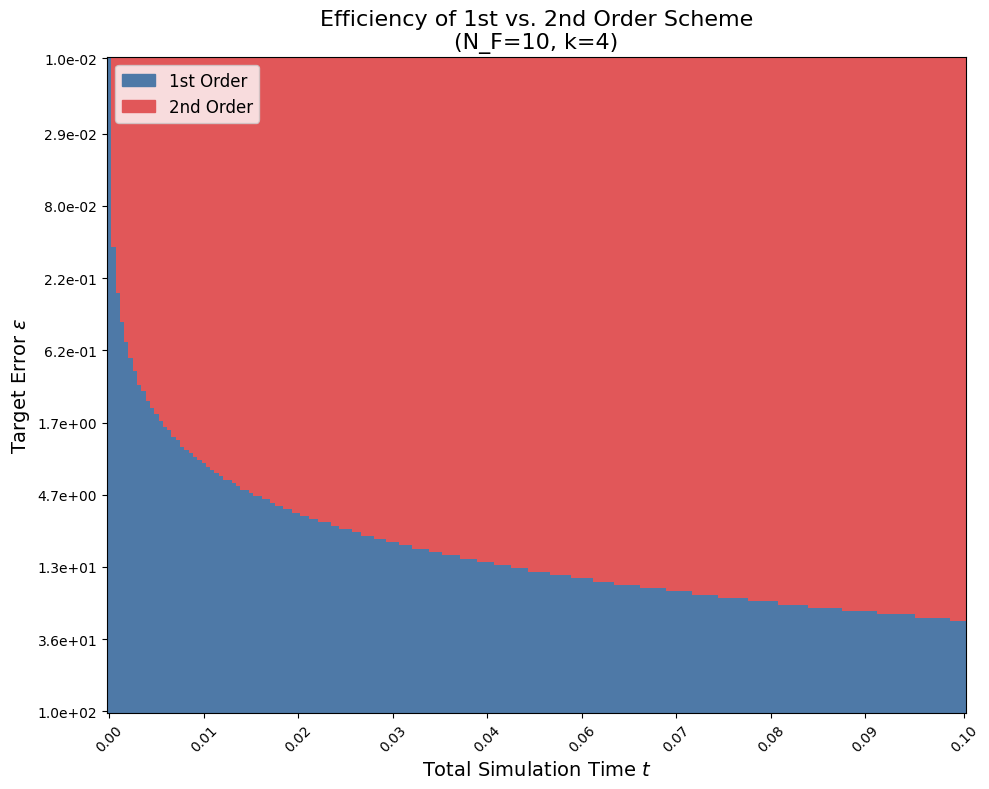}
    \caption{Theoretical cost comparison of first and second-order methods for a system with \(N_F =10\) and \(k=4\) with parameters $\gamma_{\max}, {\omega}_{\max}$ set to $1$. The colored grid indicates the more resource-efficient method for a given simulation time \(t\) and target error \(\epsilon\). In particular, blue cells indicate regions where the first-order method is cheaper, while red cells indicate regions where the second-order method is cheaper.   }
  \label{fig: resource_efficiency_1st_vs_2nd}
\end{figure}

\subsection{Numerical validation and performance analysis}

In this section, we perform numerical simulations to validate our theoretical derivations and briefly assess the practical performance of both the first and second-order approximation methods. First, we confirm that our combined time-slicing and Trotterization scheme achieves global second-order convergence. \cref{fig: second_order_analysis} illustrates the numerical convergence analysis for a sample system of 3 photon modes of truncation 3. The total simulation error is plotted against the number of time-slices, $L$, on a log-log scale, revealing a clear trend with a fitted slope of $-2$ -- thereby confirming a second-order convergence.

\begin{figure}[h!]
    \centering
    \includegraphics[width= .6\textwidth]{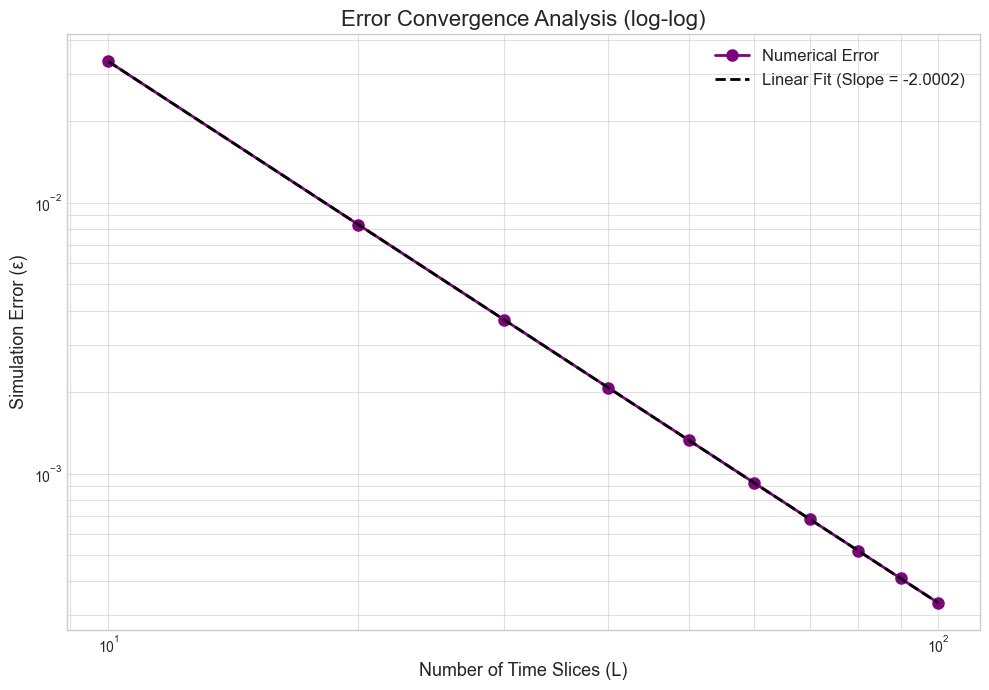}
    \caption{Log-log plot of the total simulation error \( \|U_{\text{approx}} - U_{\text{exact}}\| \) as a function of the number of time slices \( L \).  The slope of the fitted line is approximately \( -2 \), confirming global second-order convergence of the simulation method for a system of 3 photon modes ( \( N_F = 3 \)) and photon truncation \( k = 2 \).}
  \label{fig: second_order_analysis}
\end{figure}

To assess the tightness of the numerical error compared to the theoretical bound, we calculate the theoretical number of required time slices using \cref{eq:C_total_1st_optimal,eq:C_total_2nd_optimal} and compare it to the actual number of slices needed to reach that error in simulation for a range of target error tolerances $\varepsilon$. The results are provided in \cref{tab: numerical error versus theoretical error for first order method,tab: numerical error versus theoretical error for second order method} for first and second-order methods respectively.

\begin{table}[b]
\centering
\begin{tabular}{ccccccc}
\toprule
\multirow{2}{*}{$\varepsilon$} 
& \multicolumn{3}{c}{$N_F = 2$, $k = 2$} 
& \multicolumn{3}{c}{$N_F = 3$, $k = 2$} \\
\cmidrule(r){2-4} \cmidrule(l){5-7}
& Numerical $L$ & Theoretical $L$ & $\quad$ $\sim$Overest. Factor 
& Numerical $L$ & Theoretical $L$ & $\quad$ $\sim$Overest. Factor \\
\midrule
0.100 & 12   & 815   &                 & 22   & 1827    &       \\
0.050 & 25   & 1630  &                 & 44   & 3654   &       \\
0.010 & 125  & 8148  & $\sim$65×       & 220  & 18267  & $\sim$83×  \\
0.005 & 250  & 16295 &                 & 440  & 36534  &       \\
0.001 & 1255 & 81475 &                 & 2200 & 182668 &       \\
\bottomrule
\end{tabular}
\caption{
Comparison of numerical and theoretical values for the number of time slices \(L\) required to achieve a target \(\varepsilon\) for first-order simulation from $t=0$ to $t=1$. Coupling parameters are set uniformly as $\gamma_m = \tilde{\omega}_m = 1$, thus $\gamma_{\max} = 1$ and $\omega_{\max} = 1$. Theoretical values are based on the worst-case error bound, while numerical values are determined empirically through simulation. A representative overestimating factor (``Overest. Factor") is shown to illustrate that in practice, the actual simulation cost can be much lower than the theoretical prediction.
}
\label{tab: numerical error versus theoretical error for first order method}
\end{table}

\begin{table}
\centering
\begin{tabular}{c|ccc||ccc}
\toprule
\multirow{2}{*}{$\varepsilon$} 
& \multicolumn{3}{c||}{$N_F = 2$, $k = 2$} 
& \multicolumn{3}{c}{$N_F = 3$, $k = 2$} \\
\cmidrule(r){2-4} \cmidrule(l){5-7}
& Numerical $L$ & Theoretical $L$ & $\quad$ $\sim$Overest. Factor 
& Numerical $L$ & Theoretical $L$ & $\quad$ $\sim$Overest. Factor \\
\midrule
0.100 &  55  &  975  &           & 92  & 1780   &         \\
0.050 &  79 &  1380  &           & 129  & 2517   &          \\
0.010 &  171 &  3086 & $\sim$18×  & 288  & 5628  & $\sim$19.5×  \\
0.005 &  242 &  4364 &            & 408  & 1960  &            \\
0.001 &  542 &  9759 &           & 912 & 17800  &            \\
\bottomrule
\end{tabular}
\caption{Comparison between numerical and theoretical values for the number of time slices $L$ required to achieve a target error $\varepsilon$ using second-order simulation over the interval $t=0$ to $t=10$. Coupling parameters are set uniformly as $\gamma_m = \tilde{\omega}_m = 1$, thus $\gamma_{\max} = 1$ and $\omega_{\max} = 1$. The ``Overest. Factor" (overestimating factor) column quantifies the conservativeness of the theoretical bounds relative to numerical results. 
}
\label{tab: numerical error versus theoretical error for second order method}
\end{table}

\section{Extending to Mixed States}
\label{sec: Extending to Mixed-State}
Our discussion thus far has focused exclusively on pure states.  We now lift the methodology of \cref{sec:QuantumSimulationSchrodinger} to density operators, using the Choi–Jamiolkowski isomorphism.  For clarity we treat the two pictures separately.

\subsection{Schrodinger Picture}
\label{subsec: Mixed-State Schrodinger}

To simulate the time evolution of a mixed state on a quantum computer, we adopt a vectorized representation of the density matrix using the Choi–Jamiolkowski isomorphism. This formalism has been widely used in both classical~\cite{Banwell1963OnTA, Lee1971AGeneralisation, Bain2011LiouvilliansIN} and quantum contexts~\cite{Havel2002RobustPF,Kunold2023VectorizationOT}. Specifically, the density matrix \( \rho(t) \) is mapped to a pure state \( |\rho(t)\rrangle \) in an enlarged Hilbert space. This transformation enables us to recast the Liouville–von~Neumann equation,
\begin{equation}
    \frac{d}{dt}\rho(t) = -i [H, \rho(t)],
\end{equation}
as a Schrodinger-like equation in an enlarged space. In particular, defining the (un-normalized) Choi-Jamiolkowshi isomorphism as
\begin{equation}
    \label{eq: vectorised density matrix}
    |\tilde\rho(t)\rrangle = \sum_{i,j} \rho_{i,j}(t) \, |i\rangle \otimes |j\rangle,
\end{equation}
where $\rho_{i,j}=\bra{i}\rho\ket{j}$, the Liouville–von~Neumann equation can be re-expressed as
\begin{equation}
    \label{eq: vectorised schrodinger}
    \frac{d}{dt} |\rho(t) \rrangle = -i \mathcal{L}  |\rho(t)\rrangle,
\end{equation}
where \( \mathcal{L} = I \otimes H - H^T \otimes I \) is the Liouvillian super-operator acting on the vectorized state \( |\rho(t)\rrangle \) and \( H^T \) denotes the transpose of \( H \).

To simulate this equation on a quantum computer, the vectorized state must first be normalized. The normalized initial state is defined as
\begin{equation}
    |\rho(0)\rrangle = \frac{|\tilde\rho(0)\rrangle}{\lVert \rho(0) \rVert_F} = \frac{|\tilde\rho(0)\rrangle}{\sqrt{\text{Tr}(\rho(0)^\dagger \rho(0))}},
\end{equation}
where \( \lVert \cdot \rVert_F \) denotes the Frobenius norm, which is consistent with the Hilbert–Schmidt inner product: \({\llangle \rho(0) | \rho(0) \rrangle = \text{Tr}(\rho(0)^\dagger \rho(0))}\). 

 The time evolution governed by
\begin{equation}
    |\rho(t) \rrangle  = e^{-i\mathcal{L}t} |\rho(0)\rrangle
\end{equation}
can be implemented similar to the pure-state case (under the Schrodinger picture) with effectively an extra top register to simulate $e^{iH^T t}$. 
\cref{fig:time_evolve_constL}  illustrates the quantum circuit used to simulate the vectorized mixed-state evolution described in \cref{eq: vectorised schrodinger} using a second-order scheme.

\begin{figure}
\centering
\begin{adjustbox}{scale=1.05}
\begin{quantikz}
\lstick{$\ket{0}^{\otimes N}$} &
 \gate[2,style={fill=white!60,rounded corners,minimum width=1cm}]{U_{\rho_0}}
 & \gategroup[2,steps=2,style={dashed,rounded corners,fill=gray!8,inner xsep=2pt},
            background,label style={label position=below,anchor=north,yshift=-0.3cm}]{{\sc $U_\rho$}}
 & \gate[style={fill=white!20,rounded corners,minimum width=5.8cm}]
        { \left( e^{i \frac{t}{2 N_T} H_0^T} \cdot e^{i \frac{t}{N_T} H_1^T} \cdot e^{i \frac{t}{2 N_T} H_0^T} \right)^{N_T}   }
 & \qw \\
\lstick{$\ket{0}^{\otimes N}$} &
 \ghost{U_{\rho_0}} & 
 & \gate[style={fill=white!20,rounded corners,minimum width=5.8cm}]
        { \left( e^{-i \frac{t}{2 N_T} H_0} \cdot e^{-i \frac{t}{N_T} H_1} \cdot e^{-i \frac{t}{2 N_T} H_0} \right)^{N_T} }
 & \qw
\end{quantikz}
\end{adjustbox}
\caption{Second-order Trotterized circuit for the time-independent mixed-state evolution (Schrodinger picture).  Each of the $N_T$ slices applies the same pair of unitaries $e^{+iH^{T}\Delta t}$ (top register) and
$e^{-iH\Delta t}$ (bottom register), using the second order formula in this particular case. The unitary $U_{\rho_0}$ is to create the normalized initial state $|\rho(0)\rrangle$.} 
\label{fig:time_evolve_constL}
\end{figure}
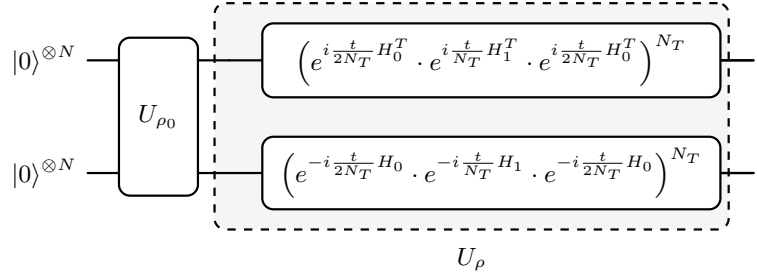

Note that the matrix \(\rho(t)\) obtained by unvectorizing or reshaping the evolved state vector  \( |\rho(t)\rrangle \) may not satisfy the unit-trace condition, \( \Tr(\rho(t)) =1 \). This is because evolution in the vectorized formalism only ensures that \( \llangle \rho(t) | \rho(t) \rrangle = 1 \), but does not maintain \( \Tr(\rho(t)) = 1 \) as required by the Liouville-Neumann equation. Thus, we need to normalize $\rho(t)$ as 

\begin{equation}
    \hat{\rho}(t) = \frac{\rho(t)}{\Tr(\rho(t))}, 
\end{equation}
where $\rho(t)$ is the matrix representation of $|\rho(t)\rrangle$. To compute \( \Tr(\rho(t)) = \sum_{i} \rho_{i,i}(t) \), we note that 
\begin{equation}
    \Tr(\rho) = \sqrt{2^N} \langle \Phi | \rho \rrangle, \hspace{0.5cm} \text{where} \hspace{0.5cm} |\Phi \rangle = \sum_{i=0}^{2^N-1} \frac{1}{\sqrt{2^N}} |i \rangle \otimes |i\rangle.
\end{equation}

\noindent
For the quantity \(\langle \Phi | \rho(t) \rrangle\), note that:
\begin{equation}
    \left| \langle \Phi | \rho(t) \rrangle \right|^2 = \left| \langle 0 | U_{\Phi}^\dagger U_{\rho(t)} | 0 \rangle \right|^2 = \left| \langle 0 | \phi(t) \rangle \right|^2,
\end{equation}
where we have defined \(|\phi(t)\rangle = U_{\Phi}^\dagger U_{\rho(t)} | 0 \rangle\). Thus, computing \(\left| \langle \Phi | \rho(t) \rrangle \right|^2\) can be done by first preparing the state \(|\phi(t)\rangle\) and then measuring the probability of observing \(|0\rangle\).
Since \(\langle \Phi | \rho(t) \rrangle\) is always positive due to the positivity of the density matrix, and since this positivity is preserved under unitary maps, we can use \( \left| \langle \Phi | \rho(t) \rrangle \right|^2 \) as a proxy to calculate \(\langle \Phi | \rho(t) \rrangle\).

\subsection{Interaction Picture}
\label{subsec: Mixed-State Evolution}

For time evolution of density matrices under the interaction picture, we begin by splitting Hamiltonian into time-independent and time-dependent parts as $H = H_0 + H_1$. The time-dependent interaction term, $H_1$, results in a \textbf{time-dependent} interaction Liouvillian, $\mathcal{L}_I(t)$. The simulation strategy is to first evolve the state under $\mathcal{L}_I(t)$ and then apply a single correction for the free evolution at the very end. The final state is given by:
\begin{align*}
    | \rho(t) \rrangle &= e^{-i\mathcal{L}_0 t} \left( \mathcal{T}\exp\left(-i \int_0^t \mathcal{L}_I(s) ds\right) \right) |\rho(0) \rrangle \\
    &= e^{-i\mathcal{L}_0 t} \left(
    \mathcal{T} \exp\left( -i \int_0^t \left( I \otimes H_I(s) - H_I^T(s) \otimes I \right) ds \right) \right) |\rho(0) \rrangle .
\end{align*}
where $ \mathcal L_0 = I\otimes H_0 - H_0^T\otimes I$, and $H_I(s)$ defined as in \cref{eq:HI_explicit}. This operator naturally decomposes into two independent time-dependent simulations of pure-state evolution. Therefore, the error bounds derived in \cref{sec:Quantum_Simulation_in_Interaction} naturally carry over.

\subsection{Initial State Preparation \& Measurements of Observables}
\label{sec: Initial States Preparation - Mixed State}

\noindent
\textbf{Initial state preparation:} Similar to the pure state case, we model the initial density matrix as a product state of the photon and emitter subsystems:
\begin{equation}
    \rho(0) = \rho_{\text{photon}} \otimes \rho_{\text{emitter}}.
\end{equation}
A natural choice to extend the pure state methodology would be to use an ensemble of coherent states
\begin{equation}
    \rho_{\text{photon}}(0)=\bigotimes_{j=1}^{N_F} \int  p_{\alpha_j}\ketbra{\alpha_j} d\alpha_j.
\end{equation} 
However, to distinguish between the mixed and pure state cases, we instead initialize the photonic field as a diagonal mixture over Fock states
\begin{equation}
    \rho_{\text{photon}} = \sum_{\boldsymbol{b}} p_{\boldsymbol{b}} \ketbra{\boldsymbol{b}}{\boldsymbol{b}},
\end{equation}
where \( \boldsymbol{b} \) denotes the photon number configuration encoded in binary. The vectorization of the diagonal photon density matrix, $\rho_{\text{photon}}$, can be expressed as the pure state
\begin{equation}
\kett{\rho_{\mathrm{photon}}} = \sum_{\boldsymbol{b}} \sqrt{p_{\boldsymbol{b}}} \ket{\boldsymbol{b}} \otimes \ket{\boldsymbol{b}}.
\end{equation}
In the special case where \( p_{\boldsymbol{b}} = \frac{1}{K} \) uniformly over all basis states \( \boldsymbol{b} \), the state becomes
\begin{equation}
\kett{\psi_{\mathrm{vec}}} = \frac{1}{\sqrt{K}} \sum_{\boldsymbol{b}} \ket{\boldsymbol{b}} \otimes \ket{\boldsymbol{b}},
\end{equation}
which can be efficiently prepared with a simple Hadamard and CNOT construction. The resulting state is a purification of the initial mixed state and can be used directly in simulations involving vectorized density matrices.

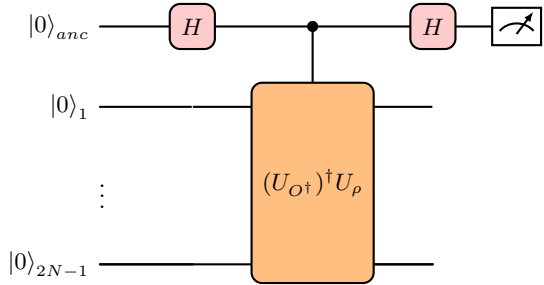
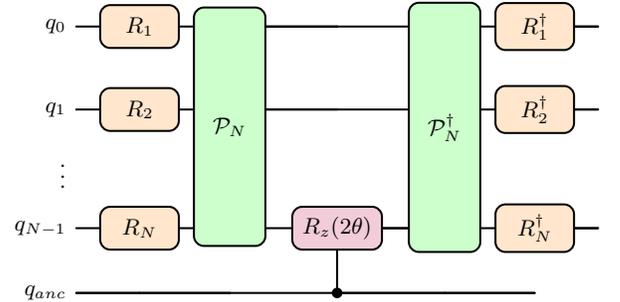
\begin{figure}[b]
\centering
\begin{subfigure}[b]{0.48\textwidth}
\centering
\begin{adjustbox}{scale=1.0}
\begin{quantikz}
\lstick{$\ket{0}_{anc}$} & \gate[1, style={fill=red!20, rounded corners, minimum width=0.6cm, minimum height=0.6cm} ]{H} & \ctrl{1} & \gate[1, style={fill=red!20, rounded corners, minimum width=0.6cm, minimum height=0.6cm} ]{H} & \meter{} \\
\lstick{$\ket{0}_1$} & \ghost[1]{H} & \gate[3, style={fill=orange!50, rounded corners, minimum width=1.2 cm}]{ (U_{O^\dagger})^\dagger U_{\rho}  } & \ghost[1]{H} \\
\phantom{Ca}\vdots\phantom{Ci} \\
\lstick{$\ket{0}_{2N-1}$} & \qw & \ghost[1]{U} & \qw  \\
\end{quantikz}
\end{adjustbox}
\caption{Hadamard test circuit to compute \(\langle O^\dagger | \rho \rangle \). This gives the real part of the overlap. To obtain the imaginary part, add an $S$ gate after the first Hadamard.}
\label{fig:hadamard_test}
\end{subfigure}
\hfill
\begin{subfigure}[b]{0.48\textwidth}
\centering
\begin{adjustbox}{scale=0.95}
\begin{quantikz}
\lstick{$q_0$}      & \gate[1, style={fill=orange!20, rounded corners, minimum width=1.1cm, minimum height=0.6cm} ]{R_1}   & \gate[4, style={fill=green!20, rounded corners, minimum width= 1 cm, minimum height=0.6cm} ]{\mathcal{P}_N} & \qw   & \gate[4, style={fill=green!20, rounded corners, minimum width= 1 cm, minimum height=0.6cm} ]{\mathcal{P}_N^\dagger} &  \gate[1, style={fill=orange!20, rounded corners, minimum width=1.1cm, minimum height=0.6cm} ]{R_1^\dagger}  & \qw \\
\lstick{$q_1$}      & \gate[1, style={fill=orange!20, rounded corners, minimum width=1.1cm, minimum height=0.6cm} ]{R_2}    & \ghost[1 ]{P}     & \qw &  \ghost[1 ]{P}  & \gate[1, style={fill=orange!20, rounded corners, minimum width=1.1cm, minimum height=0.6cm} ]{R_2^\dagger}  & \qw \\
\lstick{$\vdots$}  \\
\lstick{$q_{N-1}$}  & \gate[1, style={fill=orange!20, rounded corners, minimum width=1.1cm, minimum height=0.6cm} ]{R_{N}} &   \ghost[1 ]{P}   & \gate[1, style={fill=purple!20, rounded corners, minimum width=1.1cm, minimum height=0.6cm} ]{R_z(2 \theta)}  & \qw  & \gate[1, style={fill=orange!20, rounded corners, minimum width=1.1cm, minimum height=0.6cm} ]{R_{N}^\dagger}  & \qw \\
\lstick{$q_{anc}$} & \qw & \qw & \ctrl{-1}  & \qw & \qw
\end{quantikz}
\end{adjustbox}
\caption{Quantum circuit to implement \(\textrm{controlled-}e^{-iP\theta}\), where \(P \in \{I,X,Y,Z\}^{\otimes N}\). This is used to construct the controlled-\( U_\rho \) in the Hadamard test.}
\label{fig: controlled exp_pauli}
\end{subfigure}
\caption{(a) Hadamard test and (b) subroutine to implement the controlled exponentials of Pauli strings required for \(U_\rho\).}
\label{fig:hadamard_combined}
\end{figure}

\vspace{0.5 cm}
\noindent
\textbf{Measurements of Observables:} For a system described by density matrix \( \rho(t) \), the expectation value of some observable $O$ is given by
\begin{equation}
    \langle O \rangle = \mathrm{Tr}(\rho(t) O).
\end{equation}
In the vectorized (Liouville) formalism, where \( |\rho(t)\rrangle \) is the vectorization of \( \rho(t) \) and \( |O\rrangle \) is the vectorization of \( O \), this trace is equivalent to the inner product
\begin{equation}
    \llangle O^\dagger | \rho (t)\rrangle.
\end{equation}

Unlike the trace \( \mathrm{Tr}(\rho(t)) \), which is guaranteed to be real and positive, the quantity \( \llangle O^\dagger | \rho (t)\rrangle \) can, in general, be complex. Thus, we  compute the complex inner product directly, using the Hadamard test, as shown in~\Cref{fig:hadamard_test}.  The circuit requires a single ancilla register and controlled implementations of both \( U_{\rho} \) and \( U_{O^\dagger} \). To implement controlled-\( U_{\rho} \), note that \( U_{\rho} \) is the product of Pauli string exponentials. The sparsity of unitaries in terms of Paulis enables an efficient implementation of the controlled unitary operations by applying a controlled \( Z \)-rotation before conjugating by a Clifford circuit, as illustrated in \cref{fig: controlled exp_pauli} and formally proven in \cref{theorem: controlled exp Paulis}. Thus, significantly reducing the overhead of performing the controlled unitary.

Implementing $U_{O^\dagger}$, on the other hand, is highly dependent on the observable $O$ of interest as it generally cannot be easily decomposed into a product of Pauli string exponentials. For completeness, we provide an explicit circuit construction for an observable that is central in many quantum optics simulations: the photon number operator $O_N$. In the vectorized formalism, the observable state \( |O_N\rrangle \) can be written as
\begin{equation}
    \label{eq: vectorise O correct}
    |O_N \rrangle = \frac{1}{\mathcal{N}} \sum_{x=0}^{2^{N_F \cdot k} - 1} p(x) \, |x\rangle \otimes |x\rangle,
\end{equation}
where \( p(x) = \sum_{m=0}^{N_F - 1} \left( \left\lfloor x 2^{-k \cdot m} \right\rfloor \bmod 2^k \right) \) gives the total photon number associated with the computational basis state \( |x\rangle \), and \( \mathcal{N} = \sqrt{ \sum_x p(x)^2 } \) is a normalization constant. The state \( |O_N \rrangle \) can be efficiently prepared using the Grover–Rudolph algorithm~\cite{Grover2002CreatingST,Waite2025}, which prepares quantum states of the form
\begin{equation}
    |\psi\rangle = \sum_{x=0}^{2^n - 1} \sqrt{p_x} \, |x\rangle
\end{equation}
efficiently, provided the probability distribution \( \{p_x\} \) is non-negative and its cumulative distribution function \({ P(x) = \sum_{i=0}^{x} p_i} \) is efficiently computable. In our case, the distribution is given by
\begin{equation}
    p_x = \frac{p(x)^2}{\sum_x p(x)^2}, \quad \text{with} \quad p(x) = \sum_{m=0}^{N_F - 1} \left( \left\lfloor \frac{x}{2^{k \cdot m}} \right\rfloor \bmod 2^k \right),
\end{equation}
and the cumulative distribution function \( P(x) = \sum_{i=0}^x p_i \) is efficiently computable classically. 
The rest of the preparation circuit follows by entangling the prepared amplitude state with an auxiliary register by a CNOT fan-out, resulting in the final state \( |O_N \rrangle = \frac{1}{\mathcal{N}}\sum_x p(x)|x\rangle \otimes |x\rangle \). This construction is similar to the generation of the generalized Bell state and a schematic representation of the circuit is provided in~\cref{fig: circuit_to_create_vectorised_O}. The full preparation of \( |O_N\rrangle \) requires only \( O(\mathrm{poly}(N)) \) gates.

\begin{figure}
\begin{center}
\begin{adjustbox}{scale=1}
\begin{quantikz}[row sep=0.35cm]
\lstick{$\ket{0}_{0}$}      & \gate[4, style={fill=green!20, rounded corners, minimum width= 1 cm, minimum height=0.6cm} ]{\mathcal{X}_N} & \ctrl{4} & \qw  & \qw  & \qw\\
\lstick{$\ket{0}_{1}$}      & \ghost[1,  ]{H}  & \qw   & \ctrl{4} & \qw & \qw\\
\lstick{$\vdots$}            \\
\lstick{$\ket{0}_{N{-}1}$} & \ghost[1 ]{H}  &  \qw & \qw &  \ctrl{4} & \qw \\
\lstick{$\ket{0}_{N}$} & \qw      & \targ[style={scale=1.3}]  & \qw & \qw & \qw  & \qw \\
\lstick{$\ket{0}_{N{+}1}$} & \qw  & \qw  & \targ[style={scale=1.3}] & \qw & \qw  & \qw \\
\lstick{$\vdots$}      \\
\lstick{$\ket{0}_{2N{-}1}$} & \qw  & \qw & \qw & \targ[style={scale=1.3}]  & \qw & \qw
\end{quantikz}
\end{adjustbox}
\end{center}
\caption{Quantum circuit \( U_{O_N} \) to generate the \( 2n \)-qubit state \( |O_N \rangle = \frac{1}{\mathcal{N}} \sum_x p(x) |x\rangle \otimes |x\rangle \), where \( p(x) \) denotes the total photon number encoded in basis state \( |x\rangle \). The component \( \mathcal{X}_N \) prepares the arithmetic state \( \sum_x \frac{p(x)}{\mathcal{N}} |x\rangle \), which can be implemented efficiently using the Grover–Rudolph algorithm, given that the cumulative distribution function of \( p(x)^2 \) is classically computable.}
\label{fig: circuit_to_create_vectorised_O}
\end{figure}
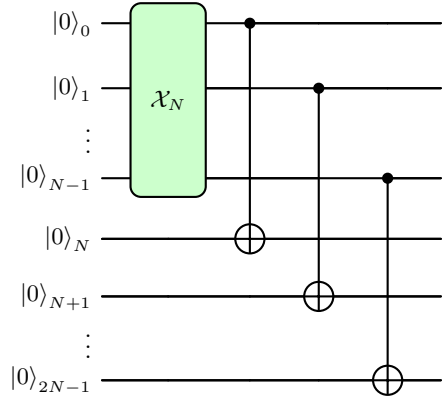

\section{Fault-Tolerant Resource Analysis with the Surface Code}
\label{sec:FaultTolerantResourceAnalysis}

To provide a concrete assessment of the resources required to simulate the extended Jaynes-Cummings model, we benchmark the quantum circuits using a fault-tolerant compilation framework. Our analysis targets a lattice surgery-based approach using homogeneous $n \times n$ rotated surface code patches and provides detailed estimates for physical qubits, execution time (Tocks), and the total space-time volume of the computation.

\subsection{Compilation Methodology}

Our resource estimation pipeline begins with a Cirq's circuit representation of a single time step of the second-order Trotter expansion, as developed in Section \ref{sec:Quantum_Simulation_in_Schrodinger} and \ref{sec:Quantum_Simulation_in_Interaction}. The target architecture consists of linked devices, each a homogeneous $20 \times 20$ grid of logical qubits encoded in rotated surface code patches with strictly local connectivity. These devices are connected using logical Bell state buffers. Logical operations are performed via lattice surgery, and magic states for implementing non-Clifford $T$ gates are produced using full-distance magic state factories. The entire compilation is staged with an assumed physical error rate of $p_{\text{phys}} = 10^{-3}$ and a target total logical error for the program of $\varepsilon_{\text{total}} = 10^{-2}$. The total error budget is divided among three primary sources: imprecision in $R_z$ gate synthesis, magic state distillation infidelity, and the space-time cost of the logical circuit.

The primary challenge in this compilation is the synthesis of the continuous $R_z(\theta)$ rotations from a discrete, fault-tolerant gate set. Each rotation must be decomposed into a sequence of Clifford gates and the non-Clifford $T$ gate, i.e., $R_z \rightarrow \{H, S, T\}^k$. To manage the finite precision of this process, we adopt an error budgeting strategy. Assuming the imprecision errors from each of the $n_{R_z}$ rotation gates in the circuit compose additively, we require that their total contribution does not exceed one-third of the error budget: $n_{R_z} \varepsilon_{R_z} \le \varepsilon_{\text{total}} / 3$, where $\varepsilon_{R_z}$ is the precision for a single rotation.  Specifically, to achieve this, each angle $-1 \le \theta \le 1$ is approximated as a rational $p/q$, where we select integers $p$ and $q$ such that $q = 2^{-\log_2(\varepsilon_{R_z}) + 3}$ and $p = \lfloor\theta q \rfloor$. The additional three bits of precision bound the imprecision such that $|\frac{p}{q} - \theta| \le \varepsilon_{R_z} / 8$ for all angles. This approximation is then compiled into a Clifford+$T$ sequence using GridSynth~\cite{RossSelingerOpt2016}.

Following the $R_z$ decomposition, the circuit is comprised of a countable number of Clifford and $T$ gates. The fidelity of the $T$ gates is ensured by magic state distillation. We model distilleries using nested versions of the $15$-to-$1$ protocol (where error scales as $p_{i+1} \approx 35 p_i^3$) and the $20$-to-$4$ protocol (scaling as $p_{i+1} \approx 5.5 p_i^2$)~\cite{Litinski2019MagicSD}. A specific distillation scheme is chosen to produce $T$ states with an output infidelity, $p_{\text{out}}$, that satisfies the second part of our error budget, $n_{T} p_{\text{out}} \le \varepsilon_{\text{total}} / 3$. These full-distance factories provide a homogeneous error model across the circuit, and are compiled to fixed logical footprints with their space-time volumes included in the total cost. For this analysis, we assume a deterministic injection procedure; incorporating probabilistic preparations would further increase qubit idling costs.

\begin{figure}[h!]
    \centering
    \includegraphics[width=\linewidth]{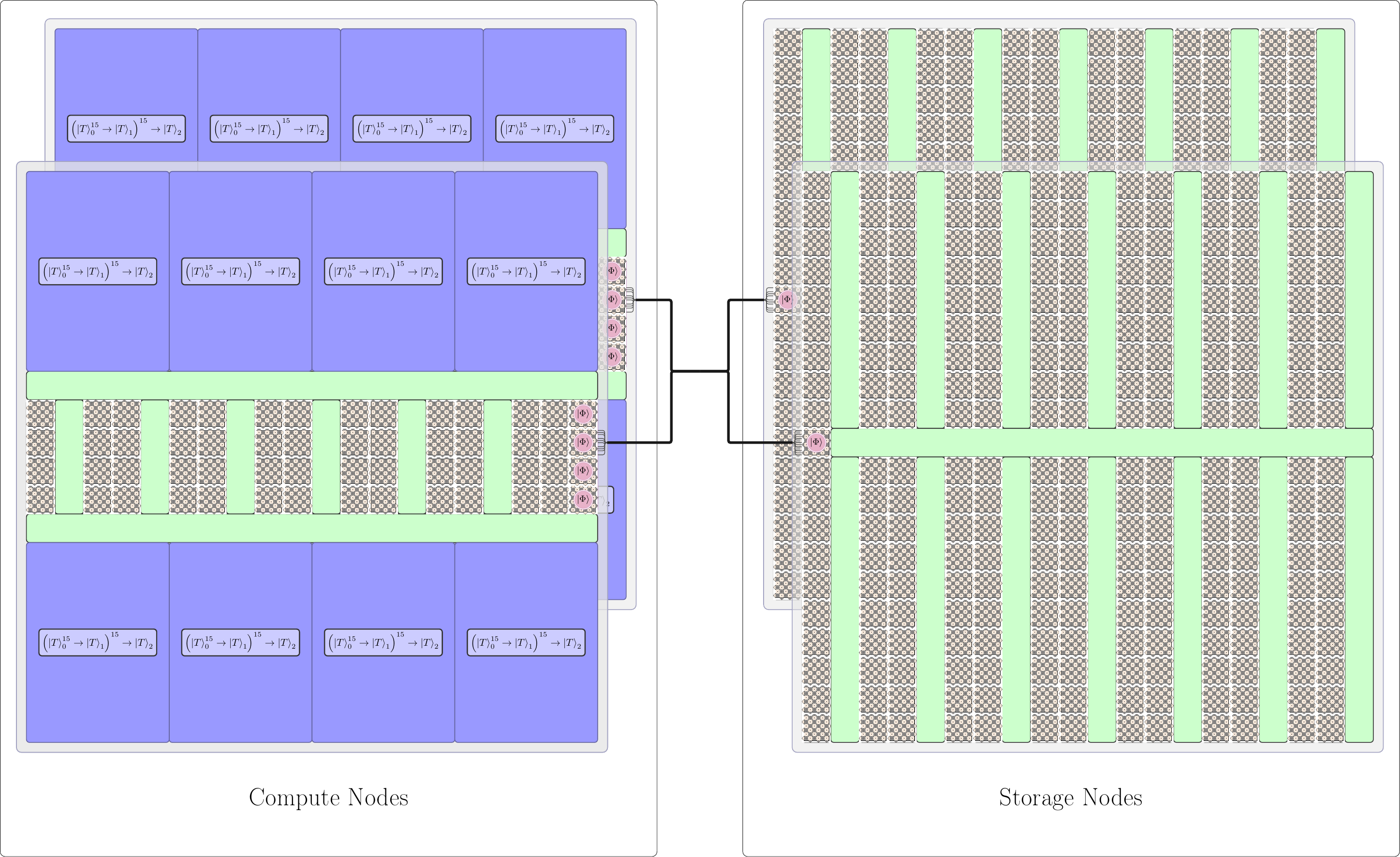}
    \caption{A schematic of the target compilation architecture used by the Rottnest compiler. The architecture consists of two Compute Nodes for pipelined graph state consumption and production, connected to a collection of Storage Nodes that idle non-participatory qubits. The nodes are linked via logical Bell state buffers.}
    \label{fig:device_architecture}
\end{figure}

Given a distillery and the full Clifford+$T$ circuit, the final stage is compilation to the physical layout using the Rottnest compiler~\cite{rottnestcompiler}. 
The compiler targets an architecture of two $20 \times 20$ compute devices for pipelined graph state consumption and production, supported by a collection of storage nodes that idle non-participatory qubits. The number of storage nodes depends on the quantity of idling qubits,  as depicted in Figure \ref{fig:device_architecture}, increasing the global physical qubit count accordingly. We assume the logical Bell state production rate between devices is on the order of $3d$ surface code cycles. The circuit itself is decomposed into a sequence of graph state preparations and $R_Z(\theta)$ measurements, with a maximum graph state size of 100 qubits. The T factories are compiled to fixed logical footprints~\cite{robertson2025resourceallocatingcompilerlattice}, and for the target instances, nested factories are constructed with only a single copy of any dependent factories to improve parallelism~\cite{Litinski2019MagicSD}.

From this compilation, we obtain a total space-time volume $v$. The final step is to determine the minimum surface code distance $d$ that bounds the logical error contribution from the Clifford gates, satisfying $v \cdot p_{\text{Logical}}(d) \le \varepsilon_{\text{total}} / 3$. The logical error rate $p_{\text{Logical}}(d)$ is calculated using coefficients for the rotated surface code under a depolarizing model~\cite{orourke2024comparepairrotatedvs}.

For the specific benchmarks presented, using the physical error rate of $10^{-3}$, we separate the accounting of the space-time volume cost of the factory from the required output fidelity of the $T$ state. From this, we determine that a nested $15$-to-$1$ to $20$-to-$4$ factory is sufficient for generating $T$ states with infidelities up to $6 \times 10^{-15}$, and a nested $15$-to-$1$ to $15$-to-$1$ factory suffices for infidelities up to $1.5 \times 10^{-21}$.

\subsection{Results and Analysis}

We applied the compilation pipeline to the circuits for simulating the eJCM in the Interaction picture. The results are presented for a range of system sizes ($N_F, k$) and target simulation accuracies ($\varepsilon$). Crucially, to demonstrate the advantage of the Interaction picture, we chose parameters corresponding to a high-frequency, near-resonant regime. Specifically, we set the photon and atomic frequencies ($\omega_m, \omega$) to be on the order of $10^3$ with a small detuning, and set all coupling strengths $\gamma_m = 1$. As noted in our numerical validation, the theoretical Trotter error bounds tend to be conservative. Based on the empirical results, we incorporate a constant `overestimation factor' of 50 to derive more practical estimates for the cost of the logical circuit.

Table~\ref{tab:interaction_small_omega_rottnest_results} details the fault-tolerant resource requirements for the complete second-order Trotter step under these conditions for $t=1$. The data shows a polynomial scaling in resources as the system size increases. For instance, simulating a moderately complex system with $N_F=50$ photon modes and a truncation of $k=6$ to a precision of $\varepsilon=0.01$ requires a space-time volume of $3.47 \times 10^{16}$ and approximately 5.6 million physical qubits. In this high-frequency regime, a classical simulation would be intractable, and a simulation in the Schrodinger picture would also be prohibitively expensive due to the large number of rotation angles required, validating the choice of the Interaction picture for this important physical scenario.

Given the highly linear nature of the circuit, and the constrained memory model in use, the runtime of the circuit is bounded by the parallelisation of the $T$ production and the overall circuit depth.
For the $N_F=50$, $k=6$ case discussed above, and assuming a fault tolerant cycle time (tock) of 1Mhz of $d$ measurement rounds we find a runtime of 1.96 years.
Were we instead to assume a 1Ghz tock time then the execution would complete within 18 hours.
The time required per tock is an architecturally dependent variable that includes transmission of measurement results, decoding and either corrections or frame tracking.

\begin{table}[H]
\centering
\begin{tabular}{ccc|ccc|cc|ccc|cccccc}
\toprule
$\varepsilon$ & N & k &
\shortstack{Precision \\ $R_Z$ (bits)} & \shortstack{Infidelity \\ $T $} 
& Vol & Tocks & Distance 
& \shortstack{Storage \\ Devices} & \shortstack{Physical \\ Qubits} \\ \midrule
0.25 & 10 & 4 & 28 & $10^{-11}$ & $2.24 \times 10^{11}$ & $9.14 \times 10^{8}$ & 30 & 1 & $2.16 \times 10^{6}$\\
0.25 & 20 & 5 & 34 & $10^{-13}$ & $1.78 \times 10^{13}$ & $5.86 \times 10^{10}$ & 35 & 1 & $2.94 \times 10^{6}$\\
0.25 & 30 & 5 & 36 & $10^{-13}$ & $5.96 \times 10^{13}$ & $1.65 \times 10^{11}$ & 36 & 1 & $3.11 \times 10^{6}$\\
0.25 & 40 & 6 & 41 & $10^{-15}$ & $3.31 \times 10^{15}$ & $6.65 \times 10^{12}$ & 40 & 1 & $3.84 \times 10^{6}$\\
0.25 & 50 & 6 & 42 & $10^{-15}$ & $6.61 \times 10^{15}$ & $1.18 \times 10^{13}$ & 41 & 2 & $5.38 \times 10^{6}$\\
0.25 & 60 & 6 & 42 & $10^{-15}$ & $1.15 \times 10^{16}$ & $1.86 \times 10^{13}$ & 41 & 2 & $5.38 \times 10^{6}$\\
0.25 & 60 & 7 & 46 & $10^{-16}$ & $1.12 \times 10^{17}$ & $1.65 \times 10^{14}$ & 44 & 2 & $6.19 \times 10^{6}$\\ \midrule
0.1 & 10 & 4 & 29 & $10^{-11}$ & $3.63 \times 10^{11}$ & $1.47 \times 10^{9}$ & 31 & 1 & $2.31 \times 10^{6}$\\
0.1 & 20 & 5 & 35 & $10^{-13}$ & $2.88 \times 10^{13}$ & $9.35 \times 10^{10}$ & 35 & 1 & $2.94 \times 10^{6}$\\
0.1 & 30 & 5 & 37 & $10^{-13}$ & $9.66 \times 10^{13}$ & $2.65 \times 10^{11}$ & 36 & 1 & $3.11 \times 10^{6}$\\
0.1 & 40 & 6 & 42 & $10^{-15}$ & $5.34 \times 10^{15}$ & $1.06 \times 10^{13}$ & 41 & 1 & $4.03 \times 10^{6}$\\
0.1 & 50 & 6 & 42 & $10^{-15}$ & $1.04 \times 10^{16}$ & $1.86 \times 10^{13}$ & 41 & 2 & $5.38 \times 10^{6}$\\
0.1 & 60 & 6 & 43 & $10^{-15}$ & $1.88 \times 10^{16}$ & $3.03 \times 10^{13}$ & 42 & 2 & $5.64 \times 10^{6}$\\
0.1 & 60 & 7 & 47 & $10^{-17}$ & $1.80 \times 10^{17}$ & $2.64 \times 10^{14}$ & 44 & 2 & $6.19 \times 10^{6}$\\ \midrule
0.01 & 10 & 4 & 30 & $10^{-11}$ & $1.18 \times 10^{12}$ & $4.67 \times 10^{9}$ & 32 & 1 & $2.46 \times 10^{6}$\\
0.01 & 20 & 5 & 37 & $10^{-13}$ & $9.56 \times 10^{13}$ & $3.05 \times 10^{11}$ & 36 & 1 & $3.11 \times 10^{6}$\\
0.01 & 30 & 5 & 38 & $10^{-14}$ & $3.12 \times 10^{14}$ & $8.48 \times 10^{11}$ & 38 & 1 & $3.46 \times 10^{6}$\\
0.01 & 40 & 6 & 43 & $10^{-15}$ & $1.74 \times 10^{16}$ & $3.48 \times 10^{13}$ & 42 & 1 & $4.23 \times 10^{6}$\\
0.01 & 50 & 6 & 44 & $10^{-16}$ & $3.47 \times 10^{16}$ & $6.20 \times 10^{13}$ & 42 & 2 & $5.64 \times 10^{6}$\\
0.01 & 60 & 6 & 45 & $10^{-16}$ & $6.16 \times 10^{16}$ & $9.89 \times 10^{13}$ & 43 & 2 & $5.92 \times 10^{6}$\\
0.01 & 60 & 7 & 48 & $10^{-17}$ & $5.79 \times 10^{17}$ & $8.50 \times 10^{14}$ & 45 & 2 & $6.48 \times 10^{6}$ \\
\bottomrule
 \end{tabular}
 \caption{Fault tolerant benchmark results for the Interaction picture for differing $N_F$, $k$ and choice of factory. Physical qubit counts assume a collection of devies with $20 \times 20$ logical surface code patches.}
 \label{tab:interaction_small_omega_rottnest_results}
 \end{table}

\section{Discussion}
\label{sec: Discussion}
\subsection{Classical Limitations and Quantum Advantage}\label{sec: Classical Limitations and Quantum Advantage}

While it is possible for small systems to perform full simulations using classical means, computational costs in both time and memory grow exponentially with the photon truncation level, $k$ and the number of photonic modes $N_F$. In particular, if one is using Krylov subspace method to solve the Schrodinger picture classically, then the cost complexity is:
\begin{equation}\label{eq:classical_cost_matrix_exp_schrodinger_pic}
O\bigl( \left[ t \, (\omega_{\max} \, N_F \, n \, + \, 2 \gamma_{\max} \, \sqrt{n} \, N_F ) + \log(1/
\varepsilon) \, \right] \, d \, 2^{N_F \, k + 1} \bigr).
\end{equation}
where $d$ is the sparsity of the Hamiltonian. Similarly, one can solve the interaction picture classically with 4th order Runge-Kutta method with the cost complexity of:
\begin{equation}\label{eq:classical_cost_RK_interaction_pic}
O\bigl( t \, (\sqrt{n} \, N_F \, \gamma_{\max}) \, d_I \,   2^{N_F \, k + 1}\bigr).
\end{equation}
where $d$ is the sparsity of $H_I$. As a result, classical approaches become intractable for larger system sizes even using highly parallelized techniques on HPC systems. For some practical applications, however, the systems of interest are those in which accurate simulation extend well beyond current classical capabilities. Approaching these regimes requires either more optimized classical protocols or adopting an alternative framework entirely that is more suited for describing the complex quantum mechanical nature of the system. This quantum framework, its circuit decomposition, resource estimates and error bounds are discussed in detail throughout the paper, and although the overhead bars near-term quantum simulation, its polynomial scaling promises a regime in which it will outperform even the most optimal classical algorithms. In \cref{tab:simulation_costs}, we outline the simulation cost complexity for various quantum approaches, all of which scale polynomially.

\begin{table}[ht]
\centering
\begin{tabular}{|c|c|c|}
\hline
\textbf{Order} & \textbf{Schrodinger Picture} & \textbf{Interaction Picture} \\
\hline
\textbf{1st Order} & 
$
\begin{aligned}
C^{(1)} = \mathcal{O}\Biggl(&\frac{T^2}{\varepsilon} \Bigl[
       \gamma_{\max}^2 N_F^3 k^3 2^{4k} \\
      &+ \omega_{\max} \gamma_{\max} N_F^2 k \bigl(2^{\frac{3k}{2}} + 2^{\frac{5k}{2}}\bigr)
    \Bigr]\Biggr)
\end{aligned}
$
& 
$
\begin{aligned}
C_{\text{int}}^{(1)} = \mathcal{O}\Biggl(&
    \frac{t^2}{\varepsilon}
    \Bigl[
        \gamma_{\max}^2 N_F^3 k^3 2^{4k} \\
        &+ \gamma_{\max} \delta_{\max} N_F^2 k 2^{\frac{3k}{2}}
    \Bigr]
\Biggr)
\end{aligned}
$ \\
\hline
\textbf{2nd Order} & 
$
\begin{aligned}
C^{(2)} = \mathcal{O}\Biggl(& 
    \frac{T^{3/2}}{\sqrt{\varepsilon}} 
        \Bigl[ 
        \gamma_{\max}^{3/2} 2^{\frac{13k}{4}} N_F^{5/2} k^{5/2} \\
        &+ \gamma_{\max} \omega_{\max}^{1/2} 2^{2k} N_F^2 k 
        + \gamma_{\max}^{1/2} \omega_{\max} 2^{\frac{9k}{4}} N_F^2 k 
        \Bigr] 
    \Biggr)
\end{aligned}
$
& 
$
\begin{aligned}
C_{\text{int}}^{(2)} = \mathcal{O}\Biggl(& 
    \frac{t^{3/2}}{\varepsilon^{1/2}} 
    \Bigl[
        \gamma_{\max}^{3/2} N_F^{5/2} k^{5/2} 2^{\frac{13k}{4}} \\
        &+ \gamma_{\max} \delta_{\max}^{1/2} N_F^2 k 2^{\frac{3k}{2}} 
        + \gamma_{\max} \delta_{\max} N_F^{3/2} k 2^{\frac{5k}{4}}
    \Bigr]
\Biggr)
\end{aligned}
$ \\
\hline
\end{tabular}
\caption{Comparison of simulation costs for Schrodinger and Interaction pictures at different Trotter orders}
\label{tab:simulation_costs}
\end{table}

Determining the regime in which quantum computers will offer computational advantage requires evaluating the quantum resources required to simulate the eJCM at various problem sizes, and compare them to classical simulation. This allow us to determine the regime in which a quantum computer will provide a definite computational advantage. A crucial first step in this comparison is to translate our abstract quantum circuits into a concrete measure of complexity. In the context of fault-tolerant quantum computation, the standard practice is to compile circuits into a universal gate set of Clifford gates and the non-Clifford T gate. Within architectures like the surface code, Clifford gates can be implemented relatively inexpensively through techniques like lattice surgery. The T gate, however, is resource-intensive, requiring costly "magic state distillation" to implement with high fidelity. Consequently, the total number of T gates serves as a primary, high-level metric for the logical cost of a quantum algorithm. 
In \cref{sec:FaultTolerantResourceAnalysis}, we performed exactly such an analysis. We constructed a rigorous, multi-stage compilation pipeline that maps our second-order Trotter-based simulation algorithm into a surface-code-protected layout, targeting a lattice-surgery architecture. The results, summarized in \cref{tab:interaction_small_omega_rottnest_results}, provide a tangible roadmap for what large-scale quantum hardware must deliver in order to simulate multi-mode light–matter interactions at scale. For example, we found that simulating a moderately complex system with $N_F=50$ photon modes and a truncation level of $k=6$ requires approximately 6.2 million physical qubits and a runtime of $4.2 \times 10^8$ surface code cycles per simulation block, assuming a nested $(15\text{-to-}1)^2$ T-factory architecture.

\begin{figure}[t]
    \centering
    \includegraphics[width= .8\textwidth]{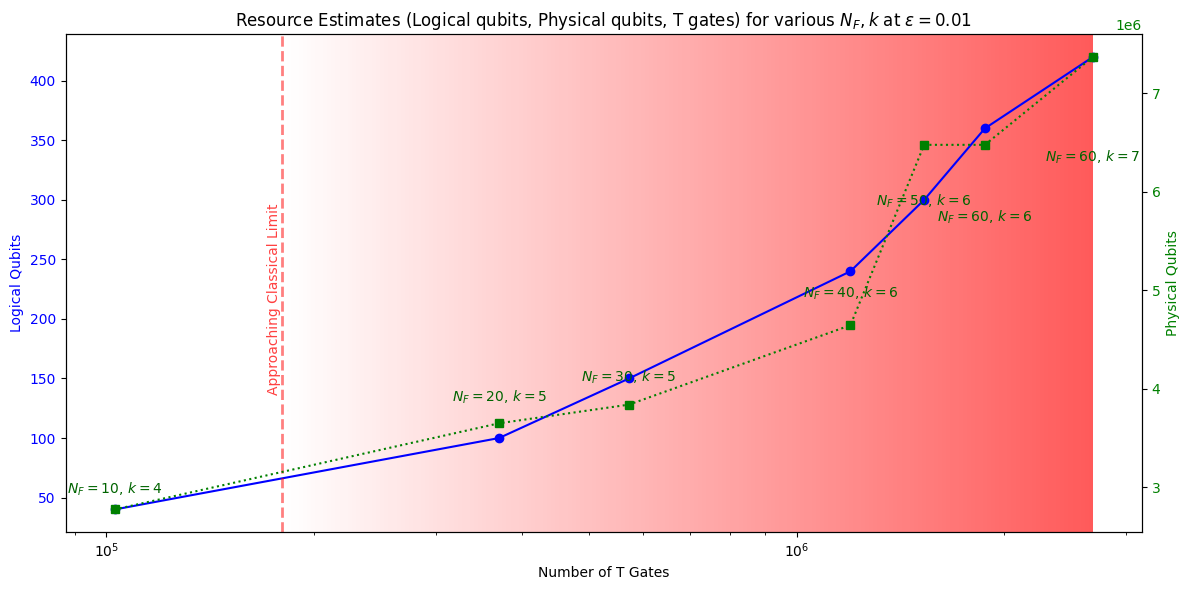}
    \caption{Scaling of qubit resources versus $T$-gate count for extended Jaynes--Cummings simulations. Blue circles (left axis) show logical qubits; green squares (right axis) show the corresponding physical qubits. Each marker is labeled by $(N_F,k)$. The red dashed line marks the ``approaching classical limit''; the red shading to the right darkens with increasing $T$, indicating progressively harder (classically intractable) regimes.}
  \label{fig: ResourceEstimate}
\end{figure}

\subsection{Alternative Quantum Approaches}
\label{sec: Alternative Quantum Approaches}

While product formulas provide a conceptually simple way to simulate time evolution with no additional ancilla qubit overhead, they suffer from suboptimal asymptotic scaling in circuit depth. In particular, their cost scales polynomially with both the simulation time \( t \) and the inverse precision \( 1/\varepsilon \), typically as \( \mathcal{O}(t^2/\varepsilon) \) under first-order decomposition, or slightly better with higher-order formulas.
To overcome this limitation, post-Trotter methods such as Quantum Signal Processing and its generalization, Quantum Singular Value Transformation (QSVT), provide optimal scaling for Hamiltonian simulation, achieving complexity \( \mathcal{O}(t + \log(1/\varepsilon)) \) in query cost and circuit depth under ideal block-encoding assumptions \cite{Gilyn2018QuantumSV}. 

The QSVT framework operates on a block-encoding of \( H \). That is, a unitary \( U_H \) is said to be a \((\alpha, a, \varepsilon)\)-block-encoding of \( H \) if
\begin{equation}
\left\| H - \alpha \cdot \left( \langle 0^a | \otimes I \right) U \left( |0^a\rangle \otimes I \right) \right\| \leq \varepsilon.
\end{equation}
Given such a block-encoding, QSVT constructs a quantum circuit that implements a polynomial transformation \( f(H/\alpha) \), where \( f \) is a real polynomial of bounded degree and definite parity (even or odd). This transformation is achieved by interleaving applications of \( U_H \), its adjoint, and single-qubit SU(2) phase rotations, with the total circuit determined by a phase sequence that encodes the desired polynomial. Recent work \cite{Simon2025LadderOB} has introduced a framework for directly block-encoding Hamiltonians built from ladder operators, offering a potentially more resource-efficient path for the general class of models to which the eJCM belongs.

For the time-independent Hamiltonian $H$ of the Schrodinger picture,  the goal is to implement the full evolution operator $e^{-iHT}$ in a single step. This is achieved by decomposing the exponential into the form \(e^{-i H T} = \cos(H T) - i \sin(H T) \), where each component is approximated independently. Both \( \cos(HT) \) and \( \sin(HT) \) can be expressed in the Chebyshev basis using the Jacobi–Anger expansion
\begin{align}
\cos(xT) &= J_0(T) + 2 \sum_{k=1}^\infty (-1)^k J_{2k}(T) T_{2k}(x), \\
\sin(xT) &= 2 \sum_{k=0}^\infty (-1)^k J_{2k+1}(T) T_{2k+1}(x),
\end{align}
where \( J_k(T) \) is the \( k \)-th Bessel function and \( T_k(x) \) is the \( k \)-th Chebyshev polynomial. These expansions converge uniformly on \( x \in [-1, 1] \), which matches the spectral range of the normalized operator \( H/\alpha \). Each expansion is truncated to a finite degree, where the required degree $r$ to achieve precision $\epsilon$ is
\begin{equation*}
    r = \Theta\left(T + \frac{\log(1/\epsilon)}{ \log(\epsilon + \log(1/\epsilon) )/T } \right).
\end{equation*}
Since QSVT circuits can only implement polynomials of definite parity, the cosine and sine approximations must be synthesized separately. The two resulting circuits are then combined through a linear combination of unitaries (LCU) technique to yield a final approximation to \( e^{-iHT} \) within error \( \varepsilon \). The total circuit depth and query cost for this time-independent simulation scale as \( \mathcal{O}(\alpha T + \log(1/\varepsilon)) \).

To extend this construction to the full time evolution under the time-dependent Hamiltonian \( H_I(t) \) of the interaction picture, we first discretize the evolution into \( L \) time slices. The total propagator is then approximated by a product of unitaries, where each unitary evolves the system under a ``frozen" Hamiltonian for a duration $\Delta t$: \( \mathcal{T}e^{-i\int_0^T H_I(s)ds} \approx \prod_{j=0}^{L-1} e^{-i H_I(t_j) \Delta t} \). The QSVT method is now applied to implement the propagator for each individual slice, \( \widetilde{U}_j \approx e^{-i H_I(t_j) \Delta t} \), as illustrated in Figure \ref{fig: QSVT at single time slice}. However, a direct composition of these unitaries would require measuring and re-preparing ancilla qubits at each step, leading to an exponentially small success probability. To avoid this, one can adopt the compression gadget framework~\cite{Fang2022TimemarchingBQ}, which allows the sequence \( \widetilde{U}_{L-1} \cdots \widetilde{U}_0 \) to be applied coherently without intermediate measurements. This is achieved by coherently tracking successful applications using an ancilla counter and applying amplitude amplification only once at the end.

A full implementation of the QSVT framework for both pictures will be considered in future work, where we plan to compare the trade-offs of these approaches with the resource estimates derived from product formulas in this paper.

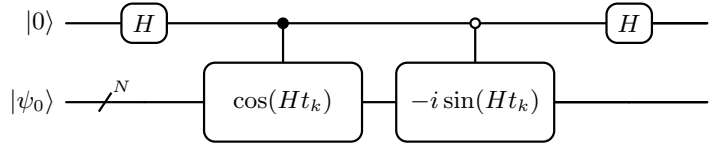
\begin{figure}[t]
\begin{center}
\begin{adjustbox}{scale=1.05}
\begin{quantikz}[row sep={1cm,between origins}, column sep=0.7cm]
\lstick{$\ket{0}$} & \gate[1, style={fill=white!20, rounded corners}]{H} & \ctrl{1} & \octrl{1} & \gate[1, style={fill=white!20, rounded corners}]{H} & \qw \\
\lstick{$\ket{\psi_0}$}  & \qwbundle{N}  & \gate[1, style={fill=white!20, rounded corners, minimum width=2 cm, minimum height=1.cm}]{\cos(H t_k)} & \gate[1, style={fill=white!20, rounded corners, minimum width=2 cm, minimum height=1.cm}]{-i \sin(H t_k)} & \qw & \qw
\end{quantikz}
\end{adjustbox}
\end{center}
\caption{Circuit diagram illustrating the block-encoding of the time-evolution operator \( e^{-i H t_k} \) at a single time slice by separate QSVT implementations of \( \cos(H t_k) \) and \( -i \sin(H t_k) \).The cosine and sine components are approximated separately using Chebyshev expansions and recombined using LCU to approximate the full exponential.}
\label{fig: QSVT at single time slice}
\end{figure}

\section{Conclusion}
\label{sec: Conclusion}

In this work, we have presented a comprehensive, end-to-end workflow analysis for the quantum simulation of the extended Jaynes-Cummings model, a fundamental model for describing multi-mode light–matter interactions with direct applications in quantum technologies such as photon addition and quasi-noiseless amplification. Our analysis provides a complete roadmap, translating the physical eJCM Hamiltonian into concrete resource estimates for a fault-tolerant implementation. 

We developed explicit quantum simulation algorithms for the eJCM in both the Schrodinger and Interaction pictures, using first and second-order Trotter-Suzuki product formulas, and extended this framework to accommodate the simulation of mixed quantum states. In particular, we derived rigorous, closed-form analytical bounds on the simulation error, expressing the gate complexity as a function of key physical parameters such as the number of photon modes ($N_F$), the Fock space truncation ($k$), and the target precision ($\varepsilon$). Our numerical simulations confirmed the theoretical convergence rates while also revealing that the worst-case bounds are conservative, suggesting that practical resource costs may be significantly lower. Moreover, we showed the potential advantage of the Interaction picture for near-resonant systems. We showed that its simulation cost scales with the maximum detuning ($\delta_{\max}$) rather than the absolute maximum frequency $(\omega_{\max})$, making it substantially more efficient for a wide range of physically relevant scenarios.

In order to bridge the gap to practical implementation, we performed a detailed physical resource analysis that goes beyond simple footprint estimation by compiling our simulation circuits to a surface-code architecture. Our analysis incorporates a sophisticated, error-budgeting model that accounts for errors from $R_z$ gate synthesis, magic state distillation, and logical qubit operations. By compiling circuits for specific problem instances, we provide concrete estimates for physical qubit counts and execution times. For instance, our analysis shows that simulating a high-frequency, near-resonant system with $N_F=50$ modes and a truncation of $k=6$ requires approximately 5.6 million physical qubits and a runtime on the order of $10^{13}$ Tocks for a simulation time of $t=1$.

While the product formulas analyzed here provide a clear and ancilla-free pathway to simulation, future work could explore the application of post-Trotter methods like QSVT. A full resource analysis of such an approach could reveal further optimizations, particularly in achieving optimal scaling with simulation time and precision. Overall, the detailed resource estimates provided in this work establish a tangible benchmark for hardware development and bring the simulation of complex light-matter interactions one step closer to practical reality on error-corrected quantum computers.

\section*{Acknowledgments}
The authors would like to thank the Boeing DC\&N organization, Jay Lowell, and Marna Kagele for creating an environment that made this research possible.  

NN would like to thank Yuval Sanders from UTS and Justin Elenewski from MIT Lincoln Lab for helpful discussions. 

\bibliography{refs}

\newpage

\begin{center}
    \textbf{Supplementary Materials: End-to-End Complexity Analysis for Quantum Simulations of Extended Jaynes-Cummings Models} 
\end{center}

\addcontentsline{toc}{section}{Supplementary Materials}
\label{sec: Suplementary Materials}

\setcounter{section}{0}
\setcounter{equation}{0}
\setcounter{figure}{0}
\setcounter{table}{0}
\renewcommand{\thesection}{S-\Roman{section}}
\renewcommand{\theequation}{S.\arabic{equation}}
\renewcommand{\thefigure}{S.\arabic{figure}}
\renewcommand{\thetable}{S.\Roman{table}}


\section{Qubit Encoding and Hamiltonian Formulation}
\label{sec_appendix:qubit_encoding_hamiltonian_formulation}

\subsection{Qubit encoding}
\label{sec_appendix:qubit_encoding}
\noindent\textbf{Claim~\ref{lemma: number operator pauli basis} (restated).} \textit{
Let \( (\hat{a}^\dagger)_n (\hat{a})_n \) be the number operator acting on Fock space truncated at \( n = 2^k - 1 \), and encoded using \( k \) qubits in the binary basis. Then its Pauli decomposition is given by
\begin{equation*}
    (\hat{a}^\dagger)_n (\hat{a})_n = \frac{n}{2} \cdot I^{\otimes k} - \sum_{j=0}^{k-1} \frac{2^{k-j-1}}{2} Z_j.
\end{equation*}
}
\begin{proof}
The number operator in the Fock basis is defined as
\begin{equation}
    (\hat{a}^\dagger)_n (\hat{a})_n = \sum_{i=0}^{2^k - 1} i \, |i\rangle \langle i|,
\end{equation}
where \( |i\rangle \) is the binary-encoded computational basis state over \( k \) qubits.
Each integer \( i \in [0, 2^k - 1] \) has a binary representation
\begin{equation}
    i = \sum_{j=0}^{k-1} b_j^{(i)} \cdot 2^j, \quad \text{with } b_j^{(i)} \in \{0,1\},
\end{equation}
so the number operator becomes
\begin{equation}
    (\hat{a}^\dagger)_n (\hat{a})_n  = \sum_{i=0}^{2^k - 1} \left( \sum_{j=0}^{k-1} 2^j b_j^{(i)} \right) |i\rangle \langle i| = \sum_{j=0}^{k-1} 2^j \left( \sum_{i=0}^{2^k - 1} b_j^{(i)} |i\rangle \langle i| \right).
\end{equation}
Now we analyze the operator \( \sum_{i=0}^{2^k - 1} b_j^{(i)} |i\rangle \langle i| \). It is a diagonal operator that acts as identity on all qubits except qubit \( j \), and can be written as
\begin{equation}
    \sum_{i=0}^{2^k - 1} b_j^{(i)} |i\rangle \langle i| = \frac{1}{2} \left( 2^{k} I^{\otimes k} - 2^{k} Z_j \right) \cdot \frac{1}{2^k} = \frac{1}{2} \left( I^{\otimes k} - Z_j \right),
\end{equation}
where we used the fact that, over all \( i \), half the time \( b_j^{(i)} = 1 \) and half the time \( b_j^{(i)} = 0 \).
Therefore,
\begin{equation}
    2^j \cdot \sum_{i=0}^{2^k - 1} b_j^{(i)} |i\rangle \langle i| = 2^j \cdot \frac{1}{2} (I^{\otimes k} - Z_j) = \frac{2^j}{2} I^{\otimes k} - \frac{2^j}{2} Z_j.
\end{equation}
Summing over all \( j = 0, \ldots, k-1 \), we obtain
\begin{equation}
    (\hat{a}^\dagger)_n (\hat{a})_n  = \sum_{j=0}^{k-1} \left( \frac{2^j}{2} I^{\otimes k} - \frac{2^j}{2} Z_j \right) = \left( \sum_{j=0}^{k-1} \frac{2^j}{2} \right) I^{\otimes k} - \sum_{j=0}^{k-1} \frac{2^j}{2} Z_j.
\end{equation}
We now simplify the coefficient of \( I^{\otimes k} \) to
\begin{equation}
    \sum_{j=0}^{k-1} \frac{2^j}{2} = \frac{1}{2} \sum_{j=0}^{k-1} 2^j = \frac{1}{2} (2^k - 1).
\end{equation}
So the full expression becomes
\begin{equation}
    (\hat{a}^\dagger)_n (\hat{a})_n  = \frac{2^k - 1}{2} \cdot I^{\otimes k} - \sum_{j=0}^{k-1} \frac{2^j}{2} Z_j.
\end{equation}
Finally, we relabel the coefficients as \( g_j = \frac{2^{k-j}}{2} \) by switching from \( j \to k - 1 - j \) to match the qubit indexing convention, where \( Z_0 \) is the least significant bit. This yields 
\begin{equation}
    (\hat{a}^\dagger)_n (\hat{a})_n = \frac{2^k - 1}{2} \cdot I^{\otimes k} - \sum_{j=0}^{k-1} \frac{2^{k-j-1}}{2} Z_j.
\end{equation}
which matches \cref{eq: pauli decomp for number operator} in the main text.
\end{proof}

\vspace{.5 cm}

\noindent\textbf{Lemma~\ref{lemma:pauli_count_photon_atom_Schrodinger} (restated).}
\textit{
Let
\[
H_{\text{Photon-Atom}}
= \sum_{m=1}^{N_F} \gamma_m \Bigl[
      \bigl( \bigotimes_{k=1}^{N_F} A_k^{(m)\dagger} \bigr) \otimes \sigma^- \;+\;
      \bigl( \bigotimes_{k=1}^{N_F} A_k^{(m)} \bigr)       \otimes \sigma^+
\Bigr],
\]
where the ladder operators $A_k^{(m)}$, $A_k^{(m)\dagger}$ are defined as in \cref{eq:D_k_and_A_k}. Then it can be decomposed into exactly
\[
|\mathcal{P}| \;=\; N_F \, 2^{k}\, k 
\]
number of distinct, non-zero Pauli strings. 
}

\begin{proof}
Fix a photon mode $m \in \{1,\dots,N_F\}$ and set
\[
A^{(m)} = \bigotimes_{k=1}^{N_F} A_k^{(m)}, 
\qquad
A^{(m)\dagger} = \bigotimes_{k=1}^{N_F} A_k^{(m)\dagger}.
\]
Because $A^{(m)}$ acts non-trivially only on the $k$ qubits of mode $m$, its Pauli
decomposition is
\[
A^{(m)} 
      = \sum_{i=1}^{2^{k}k} \alpha_i^{(m)}\, Q_i^{(m)}, 
\qquad
A^{(m)\dagger} 
      = \sum_{i=1}^{2^{k}k} \overline{\alpha_i^{(m)}}\, Q_i^{(m)},
\]
with $Q_i^{(m)}\in\mathcal{P}_{N_F}$ supported only on those $k$ qubits.
The interaction term for mode $m$ is therefore
\[
H^{(m)}_{\text{Photon-Atom}}
      = \gamma_m \sum_{i=1}^{2^{k}k}
        Q_i^{(m)} \otimes
        \bigl( \overline{\alpha_i^{(m)}}\,\sigma^- + \alpha_i^{(m)}\,\sigma^+ \bigr).
\]
Using $\sigma^-=\tfrac12(X-iY)$ and $\sigma^+=\tfrac12(X+iY)$ we obtain
\[
\overline{\alpha_i}\sigma^- + \alpha_i\sigma^+
      = \tfrac12\!\Bigl[ (\overline{\alpha_i}+\alpha_i)\,X 
                + i(\alpha_i-\overline{\alpha_i})\,Y \Bigr]
      = \Re(\alpha_i)\,X + \Im(\alpha_i)\,Y .
\]
At least one of $\Re(\alpha_i)$ or $\Im(\alpha_i)$ is non-zero, so each
index $i$ contributes exactly one Pauli string of the form
$Q_i^{(m)}\otimes X$ or $Q_i^{(m)}\otimes Y$.
Hence mode $m$ contributes precisely $2^{k}k$ distinct Pauli strings.

Because the sets of qubits on which different photon modes act are disjoint,
the Pauli strings generated for distinct $m$ are all different.
Summing over $m=1,\dots,N_F$ yields
\[
|\mathcal{P}| \;=\; N_F \times (2^{k}k),
\]
as claimed.
\end{proof}

\vspace{0.5 cm}

\noindent
\textbf{Lemma~\ref{lemma:pauli_count_H_I_interaction} (restated).} \textit{
Consider  
\[
H_{I}(t) = \frac{\hbar}{2} \sum_{m=1}^{N_F} \gamma_m \left[
e^{i({\omega_m} - \omega)t} \left( \bigotimes_{k=1}^{N_F} A_k^{(m)\dagger} \right) \otimes \sigma +
e^{-i({\omega_m} - \omega)t} \left( \bigotimes_{k=1}^{N_F} A_k^{(m)} \right) \otimes \sigma^\dagger
\right],
\]
where each $A_k^{(m)}$ is defined as
\[
A_k^{(m)} =
\begin{cases}
\hat{a} & \text{if } k = m, \\
\mathbb{I}_n & \text{otherwise},
\end{cases}
\qquad
A_k^{(m)\dagger} =
\begin{cases}
\hat{a}^\dagger & \text{if } k = m, \\
\mathbb{I}_n & \text{otherwise}.
\end{cases}
\]
Let $M_0 = \{ m \in \{1,\ldots,N_F\} \mid {\omega} = \omega_m \}$ denote the number of resonant photon modes. Then the number of distinct, nonzero Pauli strings in the decomposition of $H_{I}(t)$ when $t > 0$ is 
\[
|\mathcal{P}| = (2N_F - M_0) \cdot 2^k \cdot k.
\]
}

\begin{proof}
Fix a photon mode $m \in \{1, \dots, N_F\}$, and define
\begin{equation}
    A^{(m)} = \bigotimes_{k=1}^{N_F} A_k^{(m)}, \qquad A^{(m)\dagger} = \bigotimes_{k=1}^{N_F} A_k^{(m)\dagger}.
\end{equation}
Since $A^{(m)}$ acts nontrivially only on the $k$-qubit block corresponding to photon mode $m$, its Pauli decomposition takes the form
\begin{equation}
    \qquad
A^{(m)\dagger} = \sum_{i=1}^{2^k \cdot k} \overline{\alpha_i^{(m)}} Q_i^{(m)},
\end{equation}
where each $Q_i^{(m)} \in \mathcal{P}_{N_F}$ is a tensor product of $N_F$ Pauli matrices and identities, acting nontrivially only on the $k$ qubits of mode $m$. The interaction term for mode $m$ is then 
\begin{equation}
    H^m_{I}(t) = \frac{\hbar}{2} \gamma_m \left( e^{i\delta_m t} A^{(m)\dagger} \otimes \sigma + e^{-i\delta_m t} A^{(m)} \otimes \sigma^\dagger \right),
\qquad \textrm{where} \ \delta_m = {\omega_m} - \omega.
\end{equation}
Substituting the Pauli decompositions, we obtain
\begin{equation}
    H^m_{I}(t) = \frac{\hbar}{2} \gamma_m \sum_{i=1}^{2^k \cdot k} Q_i^{(m)} \otimes \left( \overline{\alpha_i^{(m)}} e^{i\delta_m t} \sigma + \alpha_i^{(m)} e^{-i\delta_m t} \sigma^\dagger \right).
\end{equation}
Using the identities
\begin{equation}
    \sigma = \frac{1}{2}(X + iY), \qquad \sigma^\dagger = \frac{1}{2}(X - iY),
\end{equation}
we write
\begin{align*}
\overline{\alpha_i} e^{i\delta_m t} \sigma + \alpha_i e^{-i\delta_m t} \sigma^\dagger
&= \frac{1}{2} \left[ \left( \overline{\alpha_i} e^{i\delta_m t} + \alpha_i e^{-i\delta_m t} \right) X +
i\left( \overline{\alpha_i} e^{i\delta_m t} - \alpha_i e^{-i\delta_m t} \right) Y \right].
\end{align*}
If $\delta_m \ne 0$, then the phase factors $e^{\pm i\delta_m t}$ are nontrivial, and the $X$ and $Y$ components have generically nonzero coefficients. This leads to $2\cdot (2^k \cdot k)$ distinct Pauli strings of the form $Q_i^{(m)} \otimes X$ and $Q_i^{(m)} \otimes Y$, each with distinct complex prefactors. However, if $\delta_m = 0$, that is ${\omega} = \omega_m$, then $e^{\pm i\delta_m t} = 1$, and the expression simplifies to
\begin{equation}
    \overline{\alpha_i} + \alpha_i = 2\Re(\alpha_i), \qquad \overline{\alpha_i} - \alpha_i = -2i \Im(\alpha_i),
\end{equation}
so that
\begin{equation}
    \overline{\alpha_i} \sigma + \alpha_i \sigma^\dagger = \Re(\alpha_i) X - \Im(\alpha_i) Y.
\end{equation}
Consequently, in the resonant case, $H_m(t)$ becomes
\begin{equation}
    H^m_{I}(t)  = \frac{\hbar}{2} \gamma_m \sum_{i=1}^{2^k \cdot k} Q_i^{(m)} \otimes \left( \Re(\alpha_i) X - \Im(\alpha_i) Y \right),
\end{equation}
which includes only one Pauli string per $i$, rather than two. That is, the original $2\cdot (2^k \cdot k) $ terms reduce to $2^k \cdot k$ terms due to cancellation of complex conjugate components. Thus, each resonant mode contributes exactly $2^k \cdot k$ fewer Pauli strings than a nonresonant one. Summing over all modes, we begin with $2N_F \cdot (2^k \cdot k)$ Pauli strings, and subtract $M_0 \cdot (2^k \cdot k)$ for the $M_0$ resonant modes. Therefore, the number of Pauli strings for $H^m_{I}(t) $ is 
\begin{equation}
    |\mathcal{P}| = 2N_F \cdot (2^k \cdot k) - M_0 \cdot (2^k \cdot k) = (2N_F - M_0) \cdot (2^k \cdot k). \qedhere
\end{equation}

\end{proof}

\subsection{Hamiltonian structure}
\label{sec_appendix:Hamiltonian_Structure}

\subsubsection{Commuting structure}

In this section, we analyze the structure of commuting Pauli families/groups arising from the decomposition of the photon-atom interaction term, \( H_{\text{Photon-Atom}} \), in the Hamiltonian. Understanding this structure is crucial for optimising the simulation circuit of the eJCM Hamiltonian and for reducing algorithmic errors introduced by the Trotter approximation, as discussed in Section~\ref{sec:quantum_circuit_simulation_cost}. 
In general, finding the minimum number of commuting groups is an NP-hard problem~\cite{Jena2019Pauli, Gokhale2019ON3MC, Verteletskyi2019Measurement}. Various algorithms have been proposed to address this, including approaches based on maximum clique detection and graph coloring. However, due to the structured nature of the eJCM Hamiltonian, we are able to efficiently partition the Pauli strings into commuting groups without resorting to such classical heuristics. 
We further demonstrate that the number of commuting groups obtained through our structured approach matches the optimal number of groups in instances where the exact solution can be verified.

Our approach to partition the Pauli strings in $H_{\textrm{Photon-Atom}}$ into commuting groups is by first perform the partitioning on the decomposed Pauli strings of the truncated bosonic creation and annihilation operator $(a)_n$ and $(a^\dagger)_n$ as described in \cref{eq: truncated photon creation and annihilation operator}. The following Lemma tells us the number commuting groups from this partitioning. 

\begin{lemma}
    \label{lemma: commuting groups for a and a^dagger}
    The Pauli strings decomposition of each of the truncated bosonic creation and annihilation operator $(a)_n$ and $(a^\dagger)_n$ can be partitioned into $2k$ commuting groups, where $n = 2^k-1$.
\end{lemma}

\begin{proof}
    We begin with the standard binary encoding of the truncated creation and annihilation operators on \( k \) qubits as \cref{eq: truncated photon creation and annihilation operator}. That is,
    \begin{equation}
        (a)_n = \sum_{i=1}^{n} \sqrt{i} \, |i\rangle \langle i+1|, \qquad 
        (a^\dagger)_n = \sum_{i=1}^{n} \sqrt{i} \, |i+1\rangle \langle i|,
    \end{equation}
     where \( n = 2^k - 1 \). Each basis state \( |i\rangle \) is represented by its \( k \)-bit binary string. Every term \( |i+1\rangle \langle i| \) in the sum has a fixed Hamming distance \( h \) between the bitstrings of \( i+1 \) and \( i \). This transition acts nontrivially on \( h \) qubits and can be implemented as a tensor product of \{  \( I_+ = \frac{I + Z}{2} \),  \( I_- = \frac{I - Z}{2} \), \( \sigma = \frac{X + iY}{2} \), \( \sigma^\dagger = \frac{X - iY}{2} \)\}. Since there are \( k \) possible Hamming weights \( h \in \{1, \dots, k\} \), the full decomposition consists of \( k \) such classes, with each class contributing \( 2^k \) Pauli strings.  This result immediately implies Lemma~\ref{lemma: pauli count for a and a^dagger}. Now,  within each class, the corresponding Pauli strings can be partitioned into two mutually commuting groups. This follows from the fact that all strings within the class act nontrivially on the same subset of qubits, and their (anti)commutation is determined solely by the number of \( Y \)'s in each string (or equivalently \(X\)'s). Specifically, two Pauli strings within a class commute if and only if they contain an even/odd number of \(Y\) (or \(X\) ). Therefore, each class can be partitioned into exactly two commuting groups, yielding a total of \( 2k \) commuting families for the full decomposition of \( (a^\dagger)_n \) or \( (a)_n \).
\end{proof}

\noindent\textbf{Theorem~\ref{thm:comm_family_HI} (restated).}
\textit{
Let
\[
H_I(t)=\sum_{m=1}^{N_F}\gamma_m(t)
       \Bigl[
         e^{+i\delta_m t}\,
            \bigl(\textstyle\bigotimes_{k}A_k^{(m)\dagger}\bigr)\!\otimes\!\sigma^{-}
       + e^{-i\delta_m t}\,
            \bigl(\textstyle\bigotimes_{k}A_k^{(m)}      \bigr)\!\otimes\!\sigma^{+}
       \Bigr],\qquad
\delta_m=\tilde{\omega}_m-\omega ,
\]
with \(N_F\ge1\) photon modes, \(n=2^{k}-1\) and \(M_0=0\). Then the Pauli–string decomposition of \(H_I(t)\) can be partitioned into
\[
\begin{cases}
2k & \text{if } N_F=1,\\[2pt]
4k & \text{if } N_F\ge2 ,
\end{cases}
\]
mutually commuting groups, for every \(t>0\). 
}

\begin{proof}
Lemma~\ref{lemma: commuting groups for a and a^dagger} partitions the Pauli expansion of each truncated bosonic operator  \((a)_n\) or \((a^\dagger)_n\) into \(2k\) commuting families labelled by the Hamming weight \(h\in\{1,\dots,k\}\) and the photon-\(Y\) parity (even/odd).

\vspace{4pt}
Now, every bosonic string \(Q_i\) in those families is now tensored with
\[
e^{+i\delta_m t}\,\sigma^{-}
      =e^{+i\delta_m t}\tfrac12(X-iY),\qquad
e^{-i\delta_m t}\,\sigma^{+}
      =e^{-i\delta_m t}\tfrac12(X+iY),
\]
so \emph{both} atomic Paulis \(X\) and \(Y\) appear for the same
\(Q_i\).  
If we were to leave them inside a single family the two variants would
anticommute, so we simply split each of the original \(2k\) families into an ``\(X\)'' half and a ``\(Y\)'' half.  That duplication produces \(4k\) commuting families altogether.

\vspace{4pt}
When \(N_F\ge2\) the \(k\)-qubit registers of different photon modes are
disjoint. Hence two Pauli strings drawn from different modes will commute so long as
their atomic parts commute, and this is guaranteed by the \(X/Y\) split we
just performed. Consequently the same \(4k\) family labels \((h,\text{parity},X\!/\!Y)\) can be used
globally, across all modes.

\vspace{4pt}
For \(N_F=1\) only one photon register exists. In every Hamming-weight sector either the even-\(Y\) strings \emph{or} the odd-\(Y\) strings are populated, never both, so the ``extra'' atomic split collapses and we return to the original \(2k\) commuting families.

\vspace{4pt}
Hence the Pauli decomposition of \(H_I(t)\) can always be partitioned into \(2k\) commuting groups when \(N_F=1\) and into \(4k\) groups when \(N_F\ge2\).

\end{proof}

Setting the detunings to zero (\(\delta_m = 0\)) removes the phase factors in \(H_I(t)\), transforming each complex coefficient in the bosonic decomposition into a purely real or purely imaginary number. Consequently, each bosonic string is accompanied by a single atomic Pauli (\(X\) or \(Y\), but never both), eliminating the need for the duplication step described in the preceding proof. This observation yields the following corollary.

\vspace{0.25cm}
\noindent\textbf{Corollary~\ref{corollary:number_of_commuting_groups_H_photon_atom_Schro} (restated).}
\textit{
Let
\[
H_{\text{Photon-Atom}}
=\sum_{m=1}^{N_F}\gamma_m\!
 \Bigl[
      \bigl(\textstyle\bigotimes_{k}A_k^{(m)\dagger}\bigr)\!\otimes\!\sigma^{-}\;+\;
      \bigl(\textstyle\bigotimes_{k}A_k^{(m)}      \bigr)\!\otimes\!\sigma^{+}
 \Bigr],
\]
then its Pauli decomposition can always be partitioned into exactly $2k$ mutually commuting families for all \(N_F\) values.
}

\vspace{0.25 cm}

To clearly illustrate our partitioning strategy, we summarize the procedure in Algorithm~\ref{alg:partition_pauli}, which is for Theorem \ref{thm:comm_family_HI}, however, Corollary \ref{corollary:number_of_commuting_groups_H_photon_atom_Schro} follows a similar procedure. Our algorithm exploits the structured decomposition of the bosonic operators to group Pauli strings into mutually commuting families, avoiding reliance on heuristic graph-coloring methods.
Empirically, our method matches the optimal number of commuting groups in all instances where exact solutions can be verified. For instance, when \( N_F = 1 \), \( k = 2 \), \( t > 0 \), and \( M_0 = 0 \), the Pauli decomposition of \( H_{\text{Photon-Atom}} \) yields the support
\begin{equation*}
    \{IXX, IXY, IYX, IYY, XXX, XXY, XYX, XYY, YXX, YXY, YYX, YYY, ZXX, ZXY, ZYX, ZYY\},
\end{equation*}
which can be partitioned by hand into exactly 4 commuting groups, consistent with our formula \( 2k = 4 \). See Figure~\ref{fig: comm_NF1_k2_using_our_alg} for the commuting groups identified by our method, and Figure~\ref{fig: comm_NF1_k2_using_heuristic} for the result of applying a greedy graph-coloring heuristic. In this instance, both methods yield the same result. However, as the number of Pauli terms increases, classical heuristics become unreliable. For instance, when \( N_F = 3 \), \( k = 2 \), \( t > 0 \), and \( M_0 = 0 \), the decomposition contains 48 Pauli strings (see Lemma~\ref{lemma:pauli_count_photon_atom_Schrodinger}). Our method yields exactly 8 commuting families, as predicted by Theorem~\ref{thm:comm_family_HI}, while greedy coloring may produce a suboptimal partition with more than 8 groups (see Figure~\ref{fig: comm_grouping_NF3_k2}). Note that our construction will always lead to to the same number of elements within each commuting family.

\vspace{0.25 cm}
\begin{algorithm}
\caption{Partition Pauli Strings of \( H_{\text{Photon-Atom}} \) into  Commuting Families}
\label{alg:partition_pauli}
\KwIn{Pauli strings from the decomposition of \( (a^\dagger)_n \), with \( n = 2^k - 1 \); number of photon modes \( N_F \)}
\KwOut{
    \(
    \begin{cases}
        2k & \text{if } N_F = 1 \\
        4k & \text{if } N_F \geq 2
    \end{cases}
    \)
    mutually commuting families \( \mathcal{C}_{h,p,q} \)
}

\vspace{0.1cm}
Initialise class sets \( \mathcal{G}_{h} \) for \( h = 1,\dots,k \)\;
\For{each Pauli string \( P \) in the decomposition of \( (a^\dagger)_n \)}{
    Determine the class \( h \) that \( P \) belongs to (based on hamming weight)\;
    Add \( P \) to \( \mathcal{G}_{h} \)\;
}

\vspace{0.1cm}
Initialise commuting groups \( \mathcal{G}_{h,p} \) for parity \( p \in \{0,1\} \)\;
\For{each class \( \mathcal{G}_{h} \)}{
    \For{each Pauli string \( P \in \mathcal{G}_h \)}{
        Let \( p \gets \#Y(P) \bmod 2 \)\;
        Add \( P \) to group \( \mathcal{G}_{h,p} \)\;
    }
}

\vspace{0.1cm}
Initialise final groups \( \mathcal{C}_{h,p,q} \) for atomic Pauli \( q \in \{X, Y\} \)\;
\If{\( N_F = 1 \)}{
    \For{each group \( \mathcal{G}_{h} \)}{
        \For{each Pauli string \( P \in \mathcal{G}_{h} \)}{
            \For{each atomic Pauli \( q \in \{X, Y\} \)}{
                Let \( P' \gets P \otimes q \)\;
                Compute \( p' \gets \#Y(P') \bmod 2 \)\;
                Add \( P' \) to \( \mathcal{C}_{h,p',q} \)\;
            }
        }
    }
}
\Else{
    \For{each photon mode \( m = 1 \) \KwTo \( N_F \)}{
        \For{each group \( \mathcal{G}_{h,p} \)}{
            \For{each Pauli string \( P \in \mathcal{G}_{h,p} \)}{
                Pad \( P \) with identities on all photon modes except mode \( m \)\;
                Add \( P \otimes X \) to \( \mathcal{C}_{h,p,X} \)\;
                Add \( P \otimes Y \) to \( \mathcal{C}_{h,p,Y} \)\;
            }
        }
    }
}
\vspace{0.1cm}
\Return{All commuting families \( \mathcal{C}_{h,p,q} \) for \( h = 1,\dots,k \), \( p \in \{0,1\} \), \( q \in \{X, Y\} \).}
\end{algorithm}

\begin{figure}[ht]
\centering
\begin{subfigure}[b]{0.48\textwidth}
    \centering
    \includegraphics[width=\textwidth]{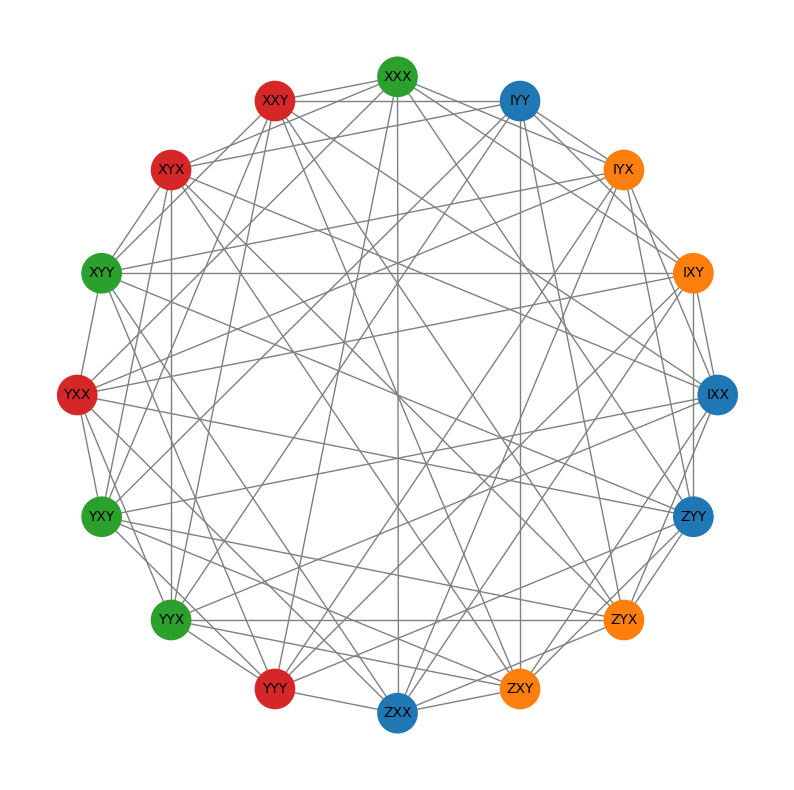}
    \caption{ }
    \label{fig: comm_NF1_k2_using_our_alg}
\end{subfigure}
\hfill
\begin{subfigure}[b]{0.48\textwidth}
    \centering
    \includegraphics[width=\textwidth]{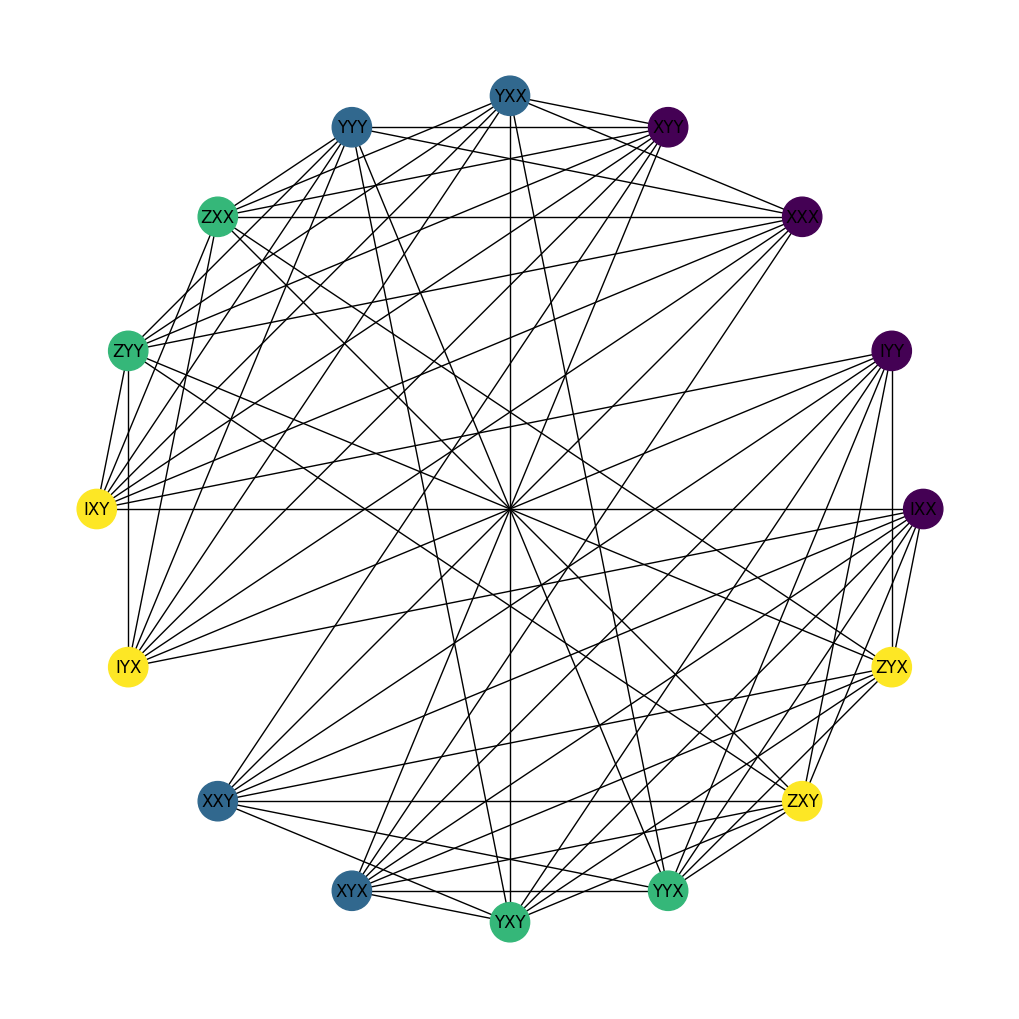}
    \caption{}
    \label{fig: comm_NF1_k2_using_heuristic}
\end{subfigure}
\caption{Complement graphs illustrating grouping structure of the Pauli strings from $H_{\textrm{Photon-Atom}}(t_i)$ with $N_F =1, k =2, t>0, M_0 =0$. (a) result from our partition technique, (b) result from running greedy coloring. }
\label{fig: comm_grouping_NF1_k2}
\end{figure}

\begin{figure}[ht]
\centering
\begin{subfigure}[b]{0.48\textwidth}
    \centering
    \includegraphics[width=\textwidth]{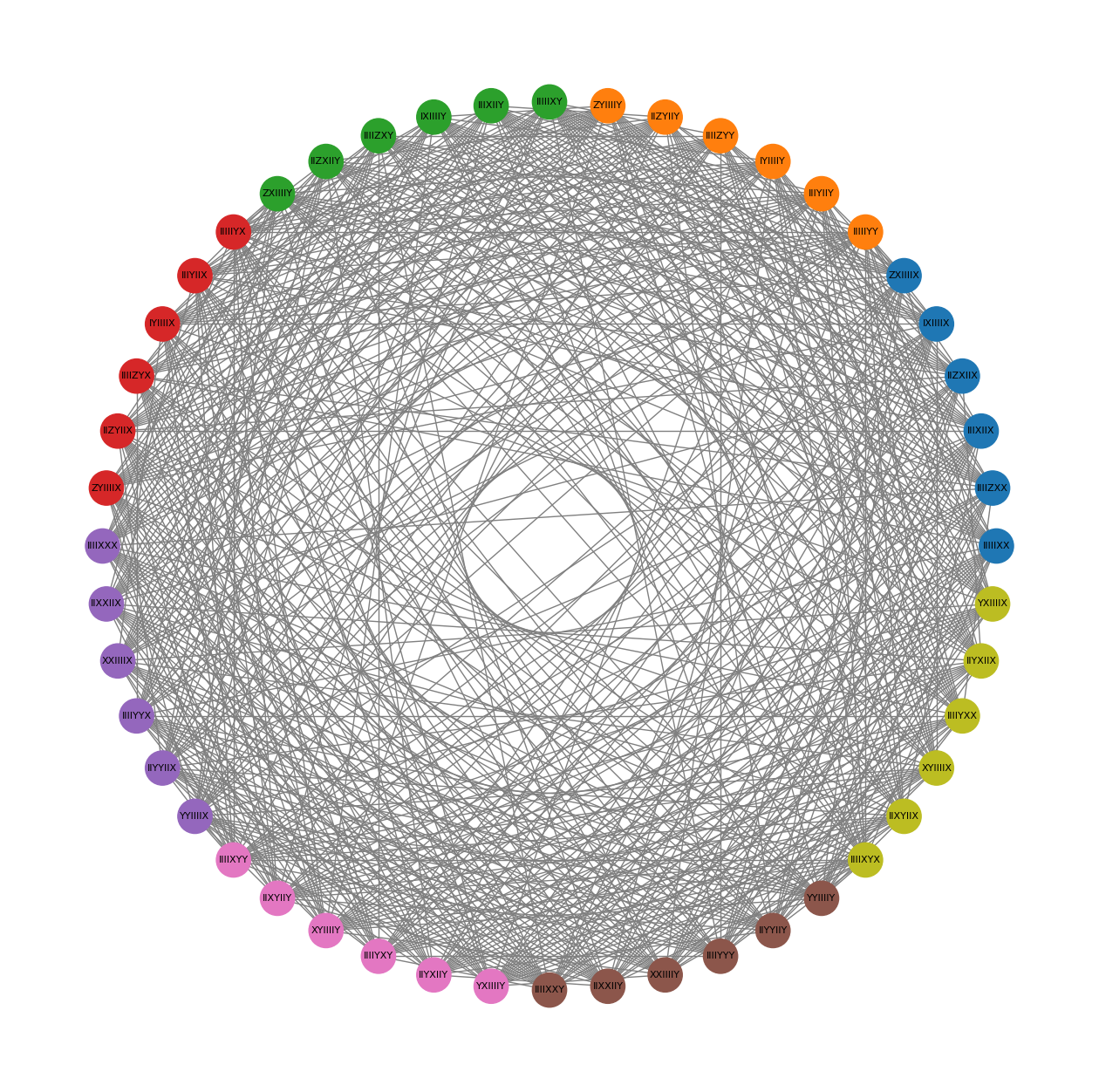}
    \caption{ }
    \label{fig: comm_NF3_k2_using_our_alg}
\end{subfigure}
\hfill
\begin{subfigure}[b]{0.48\textwidth}
    \centering
    \includegraphics[width=\textwidth]{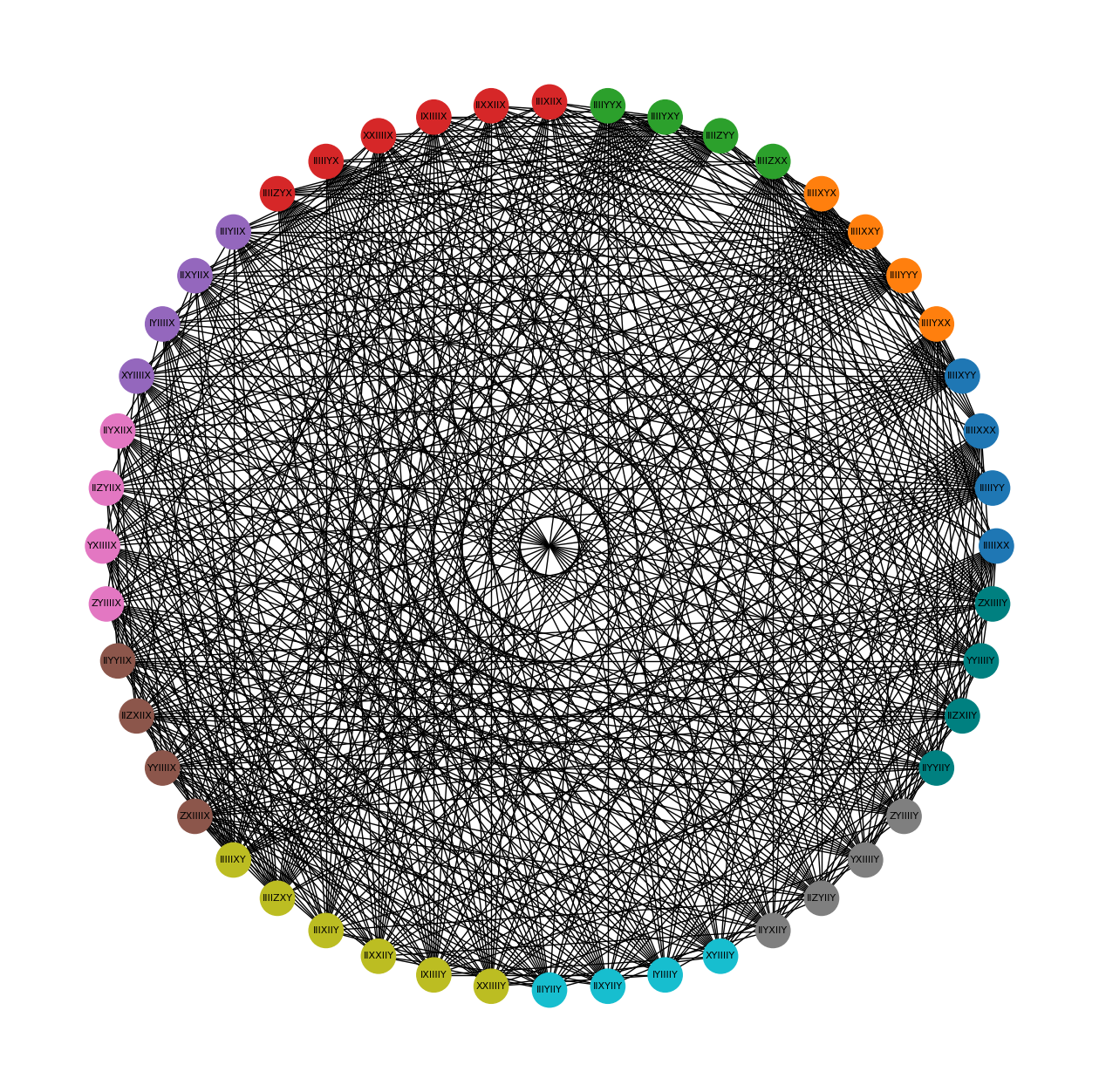}
    \caption{}
    \label{fig: comm_NF3_k2_using_heuristic}
\end{subfigure}
\caption{Complement graphs illustrating grouping structure of the Pauli strings from $H_{\textrm{Photon-Atom}}(t_i)$ with $N_F =1, k =2, t>0, M_0 =0$. (a) result from our partition technique, (b) result from running greedy coloring. }
\label{fig: comm_grouping_NF3_k2}
\end{figure}

\begin{figure}[ht]
    \centering
    \includegraphics[width= .85\textwidth]{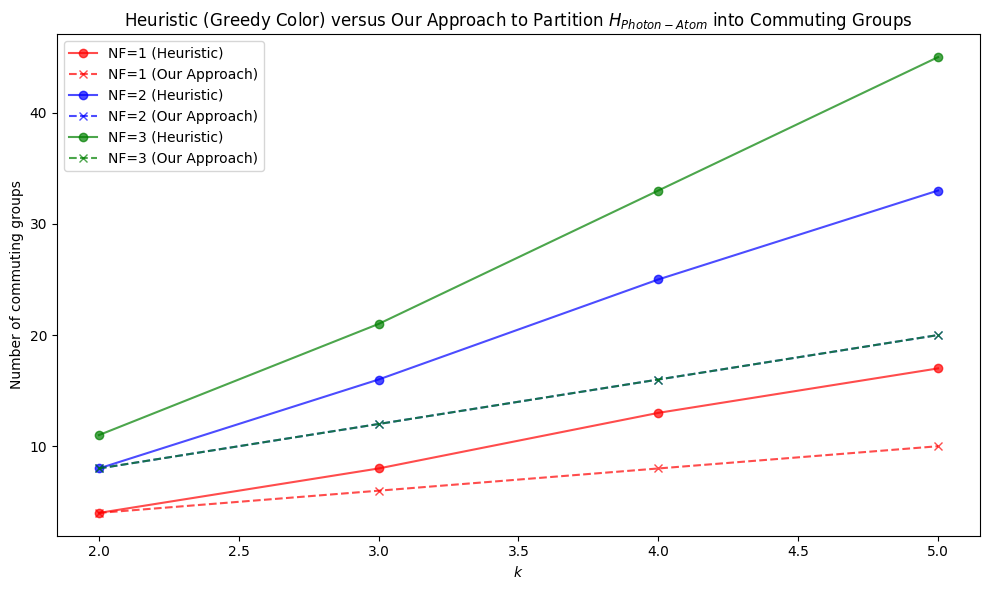}
    \caption{ Comparison between the number of commuting groups obtained via a heuristic graph-coloring algorithm and the theoretical bound given in Theorem~\ref{thm:comm_family_HI}, as a function of truncation level \(k\) and photon mode count \(N_F\). While the number of commuting groups grows with system size, our structured method consistently provides significantly better number of partitions than the heuristic approach. Notably, the heuristic does not appear to saturate for large \(N_F\), whereas our method guarantees a bound that is independent of \(N_F\) for \(N_F \geq 2\). This demonstrates that our partitioning strategy is not only computationally efficient but also optimal in practice, outperforming state of the art heuristics by a wide margin.}
  \label{fig: Heuristic_vs_OurApproach}
\end{figure}

\subsubsection{Pseudocode}
Here we provide a brief pseudocode to show how the \( H_I(t) \) of the Interaction picture can be generated. Again, note that \( H_{\text{Photon-Atom}} \) of the Schrodinger picture, \cref{eq:eJCM_2Level_Schrodinger_Picture}, is just a special case of this. 
We generate the Pauli decomposition  by decomposing the photon and atomic operators separately into Pauli strings, and tensoring them according to their qubit mapping. The time-dependent coefficients are incorporated during the combination step. See the pseudocode below for the full process. Note that the classical runtime to generate the Pauli string decomposition of $H_I(t)$ is
\begin{equation}
    \mathcal{O}(2^k \cdot k^2 + N_F \cdot 2^k \cdot k) = \mathcal{O}(2^k \cdot k \cdot (k + N_F)).
\end{equation}

\begin{algorithm}[ht]
\caption{Generate Pauli Decomposition of $H_{\mathrm{Photon-Atom}}(t)$}
\KwIn{Number of photon modes $N_F$, truncation level $k$, frequencies $\omega_m$, target frequencies $\tilde{\omega}_m$, couplings $g_m$, time $t$}
\KwOut{List of Pauli strings and corresponding coefficients}
\vspace{0.1cm}
Decompose matrices $a$, $a^\dagger$ into sets of (coefficient, Pauli label)\;
Define $\sigma = (X+iY)/2$ and $\sigma^\dagger = (X-iY)/2$\;
Initialise an empty dictionary to store Pauli strings and coefficients\;
\For{$m = 1$ \KwTo $N_F$}{
    Compute prefactors: 
    \[
    \text{prefactor}_+ = {\gamma_m} e^{i (\tilde{\omega}_m - \omega_m) t}, \quad
    \text{prefactor}_- = -{\gamma_m} e^{-i (\tilde{\omega}_m - \omega_m) t}
    \]
    \For{each Pauli term in $a^\dagger$ decomposition}{
        \For{each Pauli term in $\sigma$ decomposition}{
            Tensor Pauli labels onto corresponding photon qubits and atom qubit\;
            Multiply coefficients together with $\text{prefactor}_+$\;
            Add resulting Pauli string and coefficient to dictionary\;
        }
    }
    \For{each Pauli term in $a$ decomposition}{
        \For{each Pauli term in $\sigma^\dagger$ decomposition}{
            Tensor Pauli labels onto corresponding photon qubits and atom qubit\;
            Multiply coefficients together with $\text{prefactor}_-$\;
            Add resulting Pauli string and coefficient to dictionary\;
        }
    }
}
Filter out terms with coefficient magnitude below threshold\;
\Return{List of Pauli labels and coefficients}
\end{algorithm}


\section{Quantum Simulation of eJCM in Schrodinger Picture}

This section provides the detailed derivations for the analytical upper bounds on the Trotter errors presented in Theorem~\ref{thm: commutator norm first order} and Corrolary~\ref{cor: commutator norm second order}. These bounds are crucial for determining the number of Trotter steps required to approximate the time evolution \(e^{-i T H}\) by respective product formulas, \(S_1(T)\) and \(S_2(t)\), within a given precision for each time slice \(\Delta t\). The total error consists of two contributions: one from the non-commutativity between the interaction ( $H_{\text{Photon-Atom}}$ ) and the time-independent terms ($ H_{\text{Photon}} + H_{\text{Atom}} $), and another from non-commuting Pauli strings across different commuting families within the interaction Hamiltonian.

\vspace{0.25 cm}
\noindent\textbf{Theorem~\ref{thm: commutator norm first order} (restated).} \textit{
Consider the eJCM Hamiltonian \( H \) as in \cref{eq:eJCM_2Level_Schrodinger_Picture} with \( N_F \) photon modes truncated at level $k$ with characteristic frequencies bounded by \( \omega_{\max} = \max_m(\abs{\omega_m}, \abs{\tilde{\omega}}) \),  and an interaction term partitioned into \(G\) mutually commuting families. The first-order Trotter approximation \(S_1(t)\) over \(N_T\) steps for a total evolution time \(T\) has an error bound
\begin{equation*}
    \norm{ e^{-i T H} - S_1(T) } \leq \frac{T^2}{N_T} \left[ 
    \frac{\gamma_{\max}^2 \Lambda_k^2 N_F^2 k^2 (G - 1)}{72 G} 
    + \frac{1}{2}N_F \omega_{\max}\gamma_{\max}n^{1/2} 
    + N_F \omega_{\max}\gamma_{\max}n^{3/2} \right],
\end{equation*}
where \(\Lambda_k = (2^k + 1)^{3/2} - 1\), and \(n=2^k-1\). To achieve a target error $\epsilon_{\text{Trotter}}^{(1)} \leq \varepsilon$, the number of steps must satisfy
\begin{equation}
    N_T \geq \frac{T^2}{\varepsilon} \left[ 
    \frac{\gamma_{\max}^2 \Lambda_k^2 N_F^2 k^2 (G - 1)}{72 G} 
    + \frac{1}{2}N_F \omega_{\max}\gamma_{\max}n^{1/2} 
    + N_F \omega_{\max}\gamma_{\max}n^{3/2} \right].
\end{equation}
}

\begin{proof}
\label{proof: commutator norm first order}
We are interested in simulating the time evolution generated by the eJCM Hamiltonian
\begin{equation}
    H(t) = H_{\text{Photon}} + H_{\text{Atom}} + H_{\text{Photon-Atom}},
\end{equation}
using a first-order Trotter decomposition over \( N_T \) steps, each of duration \( T / N_T \). The approximation is given by
\begin{equation}
    S_1(t) = \left( e^{-i \frac{T}{N_T} H_{\text{Photon}}} \cdot e^{-i \frac{T}{N_T} H_{\text{Atom}}} \cdot e^{-i \frac{T}{N_T} H_{\text{Photon-Atom}} } \right)^{N_T}.
\end{equation}
Since \( H_{\text{Photon}} \) and \( H_{\text{Atom}} \) consist only of mutually commuting \( Z \)-type Pauli strings (see Section~\ref{sec:Quantum_Simulation_in_Schrodinger}), the total first-order Lie–Trotter error is bounded by~\cite{Childs2021TheoryOT}
\begin{equation}
\label{eq: first_order_trotter_bound}
    \left\| e^{-i T H} - S_1(T) \right\| 
    \leq \frac{T^2}{2 N_T} \left\| [H_{\text{Photon}} + H_{\text{Atom}}, H_{\text{Photon-Atom}}] \right\| 
    + \frac{T^2}{2 N_T} \sum_{i<j} \left\| [H^{(i)}, H^{(j)}] \right\|.
\end{equation}

Let us denote \( H_1 = H_{\text{Photon}} + H_{\text{Atom}} \), and let \( H_2 = H_{\text{Photon-Atom}} \) have Pauli decomposition
\begin{equation}
    H_2 = \sum_{\ell=1}^{N_P} \alpha_\ell Q_\ell,
\end{equation}
where \( Q_\ell \in \{I,X,Y,Z\}^{\otimes n} \). Suppose the \( Q_\ell \) are grouped into \( G = G(N_F,k) \) mutually commuting families indexed by disjoint subsets \( \mathcal{I}_1, \dots, \mathcal{I}_G \subseteq \{1, \dots, N_P\} \), so that each group defines a Hermitian operator
\begin{equation}
    H^{(i)} = \sum_{\ell \in \mathcal{I}_i} \alpha_\ell Q_\ell, \quad \text{with } [Q_\ell, Q_k] = 0 \text{ for all } \ell, k \in \mathcal{I}_i,
\end{equation}
and
\begin{equation}
    H_2 = \sum_{i=1}^G H^{(i)}.
\end{equation}
We first analyze the second term in \eqref{eq: first_order_trotter_bound}. By Lemma~\ref{lemma: pauli count for a and a^dagger}, the Pauli decomposition of each truncated ladder operator \( a \) or \( a^\dagger \) contains at most \( 2^k \cdot k \) terms, each with coefficient bounded by
\begin{equation}
    \frac{1}{3 \cdot 2^k} \left[(2^k + 1)^{3/2} - 1\right].
\end{equation}
To see this, note that  the maximum coefficient is coming from
\begin{equation}
    \frac{1}{2^k} \sum_{\substack{j = 1 \\ j \text{ odd}}}^{2^k - 1} \sqrt{j}.
\end{equation}
and 
\begin{equation}
   \frac{1}{2^k} \sum_{m = 0}^{2^{k-1} - 1} \sqrt{2m + 1} \leq \frac{1}{2^k} \int_0^{2^{k-1}} \sqrt{2m + 1} \, dm  = \frac{1}{3 \cdot 2^k} \left[(2^k + 1)^{3/2} - 1\right]
\end{equation}
Moreover, for large \(k\), we have 
\begin{equation}
    (2^k + 1)^{3/2} \approx 2^{3k/2} \left(1 + \frac{1}{2^k}\right)^{3/2} \approx 2^{3k/2} \left(1 + \frac{3}{2 \cdot 2^k}\right),
\end{equation}
so
\begin{equation}
   \frac{1}{2^k} \sum_{\substack{j = 1 \\ j \text{ odd}}}^{2^k - 1} \sqrt{j} \leq \frac{1}{3 \cdot 2^k} \cdot 2^{3k/2} \left(1 + \mathcal{O}(2^{-k})\right) \rightarrow \frac{1}{3} \sqrt{2^k}.
\end{equation}
which provides a nice asymptotic formula for the largest coefficient in the Pauli decomposition of $(a)$ or $(a^\dagger)$.
Because the interaction term involves tensoring each \( a \) or \( a^\dagger \) with \( \sigma \) or \( \sigma^\dagger \), the resulting coefficients are further scaled by \( 1/2 \), leading to
\begin{equation}
    |\alpha_\ell| \leq \gamma_{\max} \cdot \frac{1}{6 \cdot 2^k} \left[(2^k + 1)^{3/2} - 1\right] = \gamma_{\max} \cdot \frac{\Lambda_k}{6 \cdot 2^k},
\end{equation}
where \( \Lambda_k = (2^k + 1)^{3/2} - 1 \), and \( \gamma_{\max}  =  \max_{m} |\gamma_m | \). Let us assume the \( N_P \) Pauli strings are evenly distributed across the \( G \) commuting families. There are \( \binom{G}{2} \cdot (N_P/G)^2 = \frac{N_P^2(G-1)}{2G} \) inter-group commutator pairs. Each such pair contributes
\begin{equation}
    \left\| [\alpha_j Q_j, \alpha_k Q_k] \right\| \leq 2 |\alpha_j||\alpha_k| \leq 2 \gamma_{\max}^2 \left( \frac{\Lambda_k}{6 \cdot 2^k} \right)^2.
\end{equation}
Thus, the total intra-Hamiltonian Trotter error satisfies
\begin{equation}
    \epsilon_{\text{Trotter}}^{H_2} 
    \leq \frac{T^2}{2N_T} \cdot \frac{N_P^2(G-1)}{2G} \cdot 2 \gamma_{\max}^2 \left( \frac{\Lambda_k}{6 \cdot 2^k} \right)^2 
    = \frac{ \gamma_{\max}^2 \Lambda_k^2 T^2 \cdot N_P^2 (G-1) }{ 72 \cdot 4^k \cdot G N_T }.
\end{equation}
Using the bound \( N_P =  N_F \cdot 2^k \cdot k \), we obtain
\begin{equation}
\epsilon_{\text{Trotter}}^{H_2} 
\leq  \frac{ \gamma_{\max}^2 \Lambda_k^2 \cdot T^2 \cdot N_F^2 k^2 (G - 1)}{72 G N_T}
\end{equation}
Next, we estimate the cross-commutator \( [H_1, H_2] \). Let
\begin{equation}
    H_1 = H_{\text{Atom}} + \sum_{m=1}^{N_F} H_{\text{photon},m}, \quad H_2 = \sum_{m'=1}^{N_F} H_{\text{int},m'}.
\end{equation}
where \(H_{\text{photon},m} = \omega_m a_m^\dagger a_m\) and \(H_{\text{int},m'} \) is the interaction term involving only photon mode \(m'\). Note that \([H_{\text{photon},m}, H_{\text{int},m'} ] = 0\) if \(m \neq m'\). This leads to the following:
\begin{align}
    [H_1, H_2] &= \left[ H_{\text{Atom}} + \sum_{m=1}^{N_F} H_{\text{photon},m}, \quad \sum_{m'=1}^{N_F} H_{\text{int},m'} \right] \notag \\
    &= \sum_{m'=1}^{N_F} [H_{\text{Atom}}, H_{\text{int},m'}] + \sum_{m=1}^{N_F} \sum_{m'=1}^{N_F} [H_{\text{photon},m}, H_{\text{int},m'}] \notag \\
    &= \sum_{m=1}^{N_F} \left( [H_{\text{Atom}}, H_{\text{int},m}] + [H_{\text{photon},m}, H_{\text{int},m}] \right).
\end{align}
Applying the triangle inequality and let $\omega_{\max} = \max_m(\abs{\omega_m}, \abs{\tilde{\omega}})$), we have
\begin{align}
    \norm{[H_1, H_2]} &\leq \sum_{m=1}^{N_F} \norm{[H_{\text{Atom}} + H_{\text{photon},m}, H_{\text{int},m}]} \notag \\
    &\leq \sum_{m=1}^{N_F} 2 \left(\norm{H_{\text{Atom}}} + \norm{H_{\text{photon},m}}\right) \norm{H_{\text{int},m}} \notag \\ 
    &\leq \sum_{m=1}^{N_F} 2 \left(\frac{|\tilde{\omega}|}{2} + |\omega_m|n\right) \left(|\gamma_m|\sqrt{n}\right) \notag \\
    &\leq N_F \omega_{\max} (1 + 2n) (\gamma_{\max} \sqrt{n}) \notag \\
    &=  N_F \left( \omega_{\max}\gamma_{\max}n^{1/2} + 2\omega_{\max}\gamma_{\max}n^{3/2} \right)   \label{H1_H2_commutatorbound} \\
    &= \mathcal{O}(2 N_F \omega_{\max} \gamma_{\max} (2^k-1)^{3/2}) \notag.
\end{align}
Therefore, 
\begin{align}
    \epsilon_{\text{Trotter}}^{\text{cross}} & \le \frac{(T)^2}{2 N_T} 
    \left[ N_F \omega_{\max}\gamma_{\max}n^{1/2} + 2N_F \omega_{\max}\gamma_{\max}n^{3/2} \right] \notag \\ 
    &= \frac{(T)^2}{N_T} 
        \left[ \frac{1}{2}N_F \omega_{\max}\gamma_{\max}n^{1/2} + N_F \omega_{\max}\gamma_{\max}n^{3/2} \right]
\end{align}
Combining both contributions, we get
\begin{equation}
    \label{eq:final_first_order_trotter_error_bound}
    \epsilon_{\text{Trotter}}^{(1)} \leq \frac{(T)^2}{N_T} \left[ 
    \frac{\gamma_{\max}^2 \Lambda_k^2 N_F^2 k^2 (G - 1)}{72 G} 
    + \frac{1}{2}N_F \omega_{\max}\gamma_{\max}n^{1/2} 
    + N_F \omega_{\max}\gamma_{\max}n^{3/2} \right]
\end{equation}

\end{proof}

\vspace{0.5 cm}
Note that the bound in Theorem~\ref{thm: commutator norm first order} assumes that all commuting groups contain an equal number of Pauli strings, that is, \(M/G\) Pauli strings per group. However, this is not always the case, with each group have different size, with group sizes \(m_1, m_2, \dots, m_G\) satisfying \( \sum_{i=1}^G m_i = M \). 
In this scenario, the total number of potentially nonzero commutators is then bounded above by
\begin{equation}
    \sum_{i < j} m_i m_j.
\end{equation}
This expression can be simplified algebraically by first noting that
\begin{equation}
    \left( \sum_{i=1}^G m_i \right)^2 = \sum_{i=1}^G m_i^2 + 2 \sum_{i < j} m_i m_j = M^2,
\end{equation}
hence,
\begin{equation}
    \sum_{i < j} m_i m_j = \frac{1}{2} \left( M^2 - \sum_{i=1}^G m_i^2 \right).
\end{equation}
This expression is maximised when all group sizes are equal. To see this, note that \( f(x) = x^2 \) is convex, so by Jensen's inequality, we have
\begin{equation}
    \sum_{i=1}^G m_i^2 \geq G \left( \frac{1}{G} \sum_{i=1}^G m_i \right)^2 = \frac{M^2}{G},
\end{equation}
with equality if and only if all \( m_i = M/G \). Substituting this into the previous equation yields
\begin{equation}
    \sum_{i < j} m_i m_j \leq \frac{1}{2} \left( M^2 - \frac{M^2}{G} \right) = \frac{M^2(G - 1)}{2G}.
\end{equation}
Thus, the bound in Theorem~\ref{thm: commutator norm first order} represents a worst-case estimate and remains valid regardless of how the Pauli strings are partitioned into commuting groups.

\vspace{0.5 cm}

\noindent
\textbf{Jensen's inequality:} Given a  convex function \( \phi \), the inequality
\begin{equation}
    \phi \left( \frac{1}{n} \sum_{i=1}^n x_i \right) \leq \frac{1}{n} \sum_{i=1}^n \phi(x_i)
\end{equation}
holds for any finite sequence \(x_1, \dots, x_n\). Equality holds if and only if all \(x_i\) are equal.

\vspace{0.5 cm}
\noindent\textbf{Corollary~\ref{cor: commutator norm second order} (restated).} \textit{
Under the same assumptions as in Theorem~\ref{thm: commutator norm first order}, let \( S_2(T) \) denotes the second-order approximation. Let 
\begin{align*}
    C =  \gamma_{\max}^3 \cdot \frac{2 \Lambda_k^3 N_F^3 k^3 (G - 1)(G - 2)}{12 \cdot 324 G^2}  + \frac{ (1 + 2n) \sqrt{n} N_F \omega_{\max} \gamma_{\max}
    \left(
    2 \sqrt{n} N_F \gamma_{\max} + \frac{1}{2} (1 + 2n N_F) \omega_{\max}
    \right)}{12}.
\end{align*}
Then the second-order Trotter error satisfies 
\begin{align*}
    \epsilon_{\text{Trotter}}^{(2)} \leq \frac{T^3}{N_T^2} \cdot C.
\end{align*}
}

\begin{proof}
We consider the Hamiltonian \( H = H_1 + H_2 \) where \( H_1 = \sum_{j=1}^{N_Z} \beta_j P_j \) and \( H_2 = \sum_{\ell=1}^{N_P} \alpha_\ell Q_\ell \). Here  \( \{P_j\} \) are diagonal \( Z \)-type Pauli strings, and \( Q_\ell \in \{I, X, Y, Z\}^{\otimes n} \) with each coefficient satisfies \(|\alpha_\ell| \leq \frac{\gamma_{\max}}{6 \cdot 2^k} \Lambda_k\). Let \( S_2(T) \) denotes the second-order approximation
\[
S_2(T) = \left( e^{-i \frac{T}{2 N_T} H_1} \cdot e^{-i \frac{T}{N_T} H_2} \cdot e^{-i \frac{T}{2 N_T} H_1} \right)^{N_T}.
\]
Based on the error bounds as shown in Equation~\eqref{eq: second order trotter error},  the error of this second-order Trotter decomposition is bounded by 
 
\begin{equation}
    \epsilon_{\text{Trotter}}^{(2)} \leq \frac{T^3}{12 N_T^2} \left( 
    \left\lVert \sum_{i,j,k} [\alpha_i Q_i, [\alpha_j Q_j, \alpha_k Q_k]] \right\rVert 
    + \left\lVert [H_2, [H_2, H_1]] \right\rVert 
    + \frac{1}{2} \left\lVert [H_1, [H_1, H_2]] \right\rVert 
    \right).
\end{equation}

\noindent
\textbf{1. Bounding the first term:} As before, supposed the Pauli strings \( Q_\ell \) be partitioned into \( G \) mutually commuting groups. Assuming each group contains \( N_P / G \) Pauli strings, then the number of non-zero terms is 
\[ 
\binom{G}{3} \cdot \left( \frac{N_P}{G} \right)^3 = \frac{N_P^3 (G - 1)(G - 2)}{6 G^2}.
\]
Each term is bounded by $4 |\alpha|^3$. Now, note that each nonzero triple commutator satisfies
\begin{equation}
\|[\alpha Q, [\beta Q', \gamma Q'']]\| \leq 4 |\alpha||\beta||\gamma|.
\end{equation} 
Therefore,
\begin{align*}
      \sum_{i,j,k} \left\lVert [\alpha_i Q_i, [\alpha_j Q_j, \alpha_k Q_k]] \right\rVert  
      &\leq \frac{N_P^3 (G - 1)(G - 2)}{6 G^2} \cdot 4 \left( \frac{\gamma_{\max} \Lambda_k}{6 \cdot 2^k} \right)^3 \\
      &= \frac{(N_F \cdot 2^k \cdot k)^3 (G - 1)(G - 2)}{6 G^2} \cdot \frac{4 \gamma_{\max}^3 \Lambda_k^3}{216 (2^k)^3} \\
      &= \gamma_{\max}^3 \cdot \frac{ \Lambda_k^3 N_F^3 k^3 (G - 1)(G - 2)}{324 G^2}.
\end{align*}

\noindent
\textbf{2. Bounding the second and third terms:} 
This term involves nested commutators with two instances of \( H_2(t) \) and one of \(H_1\). Note that we can write these terms as:
\begin{equation}
    H_1 = \underbrace{H_{\text{Atom}} }_{H_A} + \underbrace{ \sum_{m=1}^{N_F} H_{\text{photon},m} }_{H_P}, \quad H_2 = \sum_{m'=1}^{N_F} H_{\text{int},m'}.
\end{equation}
We can use the Jacobi identity and standard commutator inequality to get the follow:
\begin{align*}
\norm{[H_1, [H_1, H_2]]} &\leq 2 \norm{H_1} \norm{[H_1, H_2]} \\
&\leq 2 \left(\norm{H_A} + \norm{H_P}\right) \left( \sum_{m=1}^{N_F} \norm{[H_{\text{Atom}} + H_{\text{photon},m}, H_{\text{int},m}(t)]} \right) \\
&\leq 2 \left(\frac{|\tilde{\omega}|}{2} + n\sum_{m=1}^{N_F}|\omega_m|\right) \left( \sum_{m=1}^{N_F} 2 \left(\norm{H_{\text{Atom}}} + \norm{H_{\text{photon},m}}\right) \norm{H_{\text{int},m}}  \right) \\
&\leq 2 \left(\frac{\omega_{\max}}{2} + n N_F \omega_{\max}\right) \left( N_F \omega_{\max}\gamma_{\max}\sqrt{n}(1+2n) \right)  \quad \textrm{using \cref{H1_H2_commutatorbound}} \notag \\
&=N_F \omega_{\max}^2 \gamma_{\max} \sqrt{n} (1 + 2n N_F)(1 + 2n)
\end{align*}
Similarly, we have 
\begin{align*}
\norm{[H_2, [H_2, H_1]]} &= \norm{[H_2, [H_1, H_2]]} \leq 2 \norm{H_2} \norm{[H_1, H_2]} \\
&\leq 2 \left( \sqrt{n} N_F \gamma_{\max} \right) \left( N_F \omega_{\max}\gamma_{\max}n^{1/2} + 2N_F\omega_{\max}\gamma_{\max}n^{3/2} \right) \\
&= 2n N_F^2 \omega_{\max}\gamma_{\max}^2 + 4n^2 N_F^2 \omega_{\max}\gamma_{\max}^2 \\
&= 2n N_F^2 \omega_{\max}\gamma_{\max}^2(1+2n)
\end{align*}

\noindent
\textbf{3. Total error:}
Substituting the three derived bounds above into our error formula, we have 
\begin{align*}
\epsilon_{\text{Trotter}}^{(2)} &\leq \frac{T^3}{12 N_T^2} \Biggl(
\gamma_{\max}^3 \frac{ 2 \Lambda_k^3 N_F^3 k^3 (G - 1)(G - 2) }{324 G^2}
+ 2n N_F^2 \omega{\max}\gamma_{\max}^2(1+2n) \\
& \qquad\qquad + \frac{1}{2} \left( N_F \omega{\max}^2 \gamma_{\max} \sqrt{n} (1 + 2n N_F)(1 + 2n) \right)
\Biggr) \\
&= \frac{T^3}{12 N_T^2} \Biggl( \gamma_{\max}^3 \cdot \frac{2 \Lambda_k^3 N_F^3 k^3 (G - 1)(G - 2)}{324 G^2} \\
& \qquad\qquad + (1 + 2n) \sqrt{n} N_F \omega_{\max} \gamma_{\max}
\left(
2 \sqrt{n} N_F \gamma_{\max} + \frac{1}{2} (1 + 2n N_F) \omega_{\max}
\right)
\Biggr)
\end{align*}

\end{proof}


\section{Quantum Simulation of eJCM in the Interaction Picture}
\label{sec_appendix:quantum_sim_interaction_picture}

\noindent

\begin{lemma}[Interaction Picture Hamiltonian]
\label{lemma:H_interaction}
Let the Schrodinger picture Hamiltonian be $H = H_0 + H_1$, where the free Hamiltonian is $H_0 = \sum_{m=1}^{N_F}  \omega _m \left( \bigotimes_{k=1}^{N_F} D_k^{(m)} \right) \otimes \mathbb{I}_2 + \frac{1}{2} \omega \left( \bigotimes_{m=1}^{N_F} \mathbb{I} \right) \otimes \sigma^z$ and the interaction term is $H_1 = \sum_{m=1}^{N_F} \gamma_m ( A_m^\dagger \otimes \sigma^- + A_m \otimes \sigma^+ )$, as in \cref{eq:eJCM_2Level_Schrodinger_Picture}. The corresponding interaction picture Hamiltonian, defined as $H_I(t) = e^{iH_0 t} H_1 e^{-iH_0 t}$, is given by
\begin{equation*}
    H_I(t) = \sum_{m=1}^{N_F} \gamma_{m} \left[
    e^{i(\tilde{\omega}_m - \omega)t} \left( \bigotimes_{k=1}^{N_F} A_k^{(m)\dagger} \right) \otimes \sigma^- +
    e^{-i(\tilde{\omega}_m - \omega)t} \left( \bigotimes_{k=1}^{N_F} A_k^{(m)} \right) \otimes \sigma^+
    \right].
\end{equation*}
\end{lemma}

\begin{proof}
Start with the definition $H_I(t) = e^{iH_0 t} H_1 e^{-iH_0 t}$. Now, since $H_{\text{Photon}}$ and $H_{\text{Atom}}$ commute, we can write $e^{iH_0 t} = e^{iH_{\text{Photon}}t}e^{iH_{\text{Atom}}t}$. This allows the transformation to be separated as:
\begin{align*}
    H_I(t) = \sum_{m=1}^{N_F} \gamma_m \Big[ &(e^{iH_{\text{Photon}}t} A_m^\dagger e^{-iH_{\text{Photon}}t}) \otimes (e^{iH_{\text{Atom}}t} \sigma^- e^{-iH_{\text{Atom}}t}) \\
    &+ (e^{iH_{\text{Photon}}t} A_m e^{-iH_{\text{Photon}}t}) \otimes (e^{iH_{\text{Atom}}t} \sigma^+ e^{-iH_{\text{Atom}}t}) \Big].
\end{align*}
We derive these operator transformations using the Baker-Campbell-Hausdorff (BCH) expansion, 
\[
e^X Y e^{-X} = Y + [X,Y] + \frac{1}{2!}[X,[X,Y]] + \cdots.
\]

First, we consider the photon operator transformations. For the creation operator $A_m^\dagger = \hat{a}_m^\dagger$, we set 
\[ 
X = i{\omega}_m t (\hat{a}_m^\dagger \hat{a}_m), \quad \textrm{and} \quad  Y = \hat{a}_m^\dagger
\]
The nested commutators follow a simple pattern starting with $[\hat{a}_m^\dagger \hat{a}_m, \hat{a}_m^\dagger] = \hat{a}_m^\dagger$:
\begin{align*}
    [X, Y] &= [i{\omega}_m t (\hat{a}_m^\dagger \hat{a}_m), \hat{a}_m^\dagger] = (i{\omega}_m t) [\hat{a}_m^\dagger \hat{a}_m, \hat{a}_m^\dagger] = (i{\omega}_m t) \hat{a}_m^\dagger = (i{\omega}_m t) Y, \\
    [X, [X, Y]] &= [i{\omega}_m t (\hat{a}_m^\dagger \hat{a}_m), (i{\omega}_m t) Y] = (i{\omega}_m t)^2 [\hat{a}_m^\dagger \hat{a}_m, \hat{a}_m^\dagger] = (i{\omega}_m t)^2 Y.
\end{align*}
Substituting this pattern into the BCH series yields the Taylor series for an exponential, 
\[
e^X Y e^{-X} = Y \sum_{k=0}^\infty \frac{(i{\omega}_m t)^k}{k!} = Y e^{i{\omega}_m t}
\]
For the annihilation operator $A_m = \hat{a}_m$, the initial commutator is $[\hat{a}_m^\dagger \hat{a}_m, \hat{a}_m] = -\hat{a}_m$, which introduces a sign flip, leading to the result $e^{-i{\omega}_m t} A_m$. Thus,
\begin{equation*}
    e^{iH_{\text{Photon}}t} A_m^\dagger e^{-iH_{\text{Photon}}t} = e^{i{\omega}_m t} A_m^\dagger \quad \text{and} \quad e^{iH_{\text{Photon}}t} A_m e^{-iH_{\text{Photon}}t} = e^{-i{\omega}_m t} A_m.
\end{equation*}

Next, we consider the atomic operator transformations. For the lowering operator $\sigma^-$, we set
\[
X = i\frac{\omega t}{2}\sigma^z, \quad Y = \sigma^-.
\]
We then use the commutator $[\sigma^z, \sigma^-] = -2\sigma^-$ to get:
\begin{align*}
    [X, Y] &= [i\frac{\omega t}{2}\sigma^z, \sigma^-] = (i\frac{\omega t}{2})[\sigma^z, \sigma^-] = (i\frac{\omega t}{2})(-2\sigma^-) = (-i\omega t) Y, \\
    [X, [X, Y]] &= [i\frac{\omega t}{2}\sigma^z, (-i\omega t) Y] = (-i\omega t) [i\frac{\omega t}{2}\sigma^z, \sigma^-] = (-i\omega t)(-i\omega t Y) = (-i\omega t)^2 Y.
\end{align*}
The BCH expansion becomes the Taylor series for $e^{-i\omega t}$:
\begin{align*}
    e^{i\frac{\omega t}{2}\sigma^z} \sigma^- e^{-i\frac{\omega t}{2}\sigma^z} &= Y + (-i\omega t)Y + \frac{1}{2!}(-i\omega t)^2 Y + \cdots = Y e^{-i\omega t}.
\end{align*}
Thus, we find $e^{iH_{\text{Atom}}t} \sigma^- e^{-iH_{\text{Atom}}t} = e^{-i\omega t}\sigma^-$. For the raising operator $\sigma^+$, the commutator $[\sigma^z, \sigma^+] = 2\sigma^+$ similarly leads to the result $e^{i\omega t}\sigma^+$. Therefore, 
\begin{equation*}
    e^{iH_{\text{Atom}}t} \sigma^- e^{-iH_{\text{Atom}}t} = e^{-i\omega t}\sigma^- \quad \text{and} \quad e^{iH_{\text{Atom}}t} \sigma^+ e^{-iH_{\text{Atom}}t} = e^{i\omega t}\sigma^+.
\end{equation*}

Now, we can put everything together. In particular, the first term in the sum becomes
\[
(e^{i{\omega}_m t} A_m^\dagger) \otimes (e^{-i\omega t}\sigma^-) = e^{i({\omega}_m - \omega)t} (A_m^\dagger \otimes \sigma^-).
\]
The second term becomes 
\[
(e^{-i{\omega}_m t} A_m) \otimes (e^{i\omega t}\sigma^+) = e^{-i({\omega}_m - \omega)t} (A_m \otimes \sigma^+).
\]
Put them together yields \cref{eq:HI_explicit}.

\end{proof}

\vspace{0.5 cm}

\noindent
\textbf{Proposition~\ref{prop: second order time discretisation error} (restated).} \textit{
For the simulation of a time-dependent Hamiltonian \(H(t)\) over a total time \(t\) using \(L\) steps of the second-order midpoint integrator \( \tilde{U}_j = e^{-i \Delta t H(t_j + \Delta t /2}) \), the global error bounded by 
\begin{equation}
    \left\| \mathcal{T} \exp\left(-i \int_0^t H(s) \, ds \right) - \prod_{j=0}^{L-1} U_j \right\|
\leq \frac{t^3}{L^2} \left( \frac{1}{24} \left\| H''(t) \right\|_{\infty, \infty} + \frac{1}{12} \left\| [H'(t), H(t)] \right\|_{\infty, \infty} \right)
\end{equation}
}

\begin{proof}
Let the total simulation time be \( t \), divided into \( L \) equal steps of duration \( \Delta t = t/L \). 
Consider a single time interval from \( t_j \) to \( t_{j+1} = t_j + \Delta t \). 
The exact time-evolution operator over this step is
\begin{equation}
    U_j = \mathcal{T} \exp\left(-i \int_{t_j}^{t_{j+1}} H(s)\, ds \right).
\end{equation}
Defining \( A(s) = -iH(s) \), we can express \( U_j \) using the Magnus expansion as \cite{Blanes2008TheME}
\begin{equation}
    U_j = \exp\left( \Omega_j(\Delta t) \right), \quad \text{where} \quad \Omega_j(\Delta t) = \sum_{k=1}^\infty \Omega_{j,k}(\Delta t).
\end{equation}
The first few terms relevant to our analysis are:
\begin{align}
    \Omega_{j,1}(\Delta t) &= \int_{t_j}^{t_{j+1}} A(s)\, ds, \label{eq:Omega1_def} \\
    \Omega_{j,2}(\Delta t) &= \frac{1}{2} \int_{t_j}^{t_{j+1}} ds_1 \int_{t_j}^{s_1} ds_2\, [A(s_1), A(s_2)], \label{eq:Omega2_def} \\
    \Omega_{j,3}(\Delta t) &= \frac{1}{6} \int_{t_j}^{t_{j+1}} ds_1 \int_{t_j}^{s_1} ds_2 \int_{t_j}^{s_2} ds_3\, \Big( [A(s_1), [A(s_2), A(s_3)]] + [[A(s_1), A(s_2)], A(s_3)] \Big) \label{eq:Omega3_def}
\end{align}
The series continues in this manner, with each \( \Omega_{j,k}(\Delta t) \) involving \( k \)-fold nested commutators of \( A(s) \), and \( \Omega_{j,k}(\Delta t) \) generally being \( \mathcal{O}((\Delta t)^k) \).

The second-order midpoint approximation for the evolution operator over this step is \( \tilde{U}_j = e^{-i \Delta t H(t_j + \Delta t/2)} \). This corresponds to choosing an approximate Magnus exponent \( \Omega_{\text{approx},j}(\Delta t) \) given by the midpoint rule applied to the integral in \( \Omega_{j,1}(\Delta t) \):
\begin{equation}
    \Omega_{\text{approx},j}(\Delta t) = \Delta t A(t_c),
\end{equation}
where \( t_c = t_j + \Delta t/2 \) is the midpoint of the interval. Thus, \( \tilde{U}_j = e^{\Omega_{\text{approx},j}(\Delta t)} \).

The error in the exponent for this single step is \( \Delta\Omega_j(\Delta t) = \Omega_j(\Delta t) - \Omega_{\text{approx},j}(\Delta t) \). This error arises from two main sources relevant to the \( \mathcal{O}((\Delta t)^3) \) local error. In particular, 1) the  error from approximating \( \Omega_{j,1}(\Delta t) \) by \( \Delta t A(t_c) \), and 2) the error from truncating the Magnus series, primarily the neglected \( \Omega_{j,2}(\Delta t) \) term, as \( \Omega_{j,3}(\Delta t) \) and higher terms will contribute to higher orders of error than we need to explicitly track for the leading error term of \( \Delta\Omega_j(\Delta t) \) for this second-order method. Thus, we can write

\[
\Delta\Omega_j(\Delta t) = \underbrace{\left(\Omega_{j,1}(\Delta t) - \Delta t A(t_c)\right)}_{\text{Error Source 1}} + \underbrace{\Omega_{j,2}(\Delta t)}_{\text{Error Source 2}} + \Omega_{j,3}(\Delta t) + \dots
\]
Let's us now analyze these primary error sources.

\vspace{0.25 cm}
\noindent
\textbf{Error source 1:}
We Taylor expand \( A(s) \) around the midpoint \( t_c = t_j + \Delta t/2 \), which gives us 
\[
A(s) = A(t_c) + (s-t_c)A'(t_c) + \frac{(s-t_c)^2}{2}A''(t_c) + \mathcal{O}((s-t_c)^3).
\]
Integrating this from \( t_j = t_c - \Delta t/2 \) to \( t_{j+1} = t_c + \Delta t/2 \), we have 
\begin{align*}
    \Omega_{j,1}(\Delta t) = \int_{t_c - \Delta t/2}^{t_c + \Delta t/2} A(s) ds &= \int_{-\Delta t/2}^{\Delta t/2} \left( A(t_c) + u A'(t_c) + \frac{u^2}{2}A''(t_c) + \mathcal{O}(u^3) \right) du \\
    &= \left[ uA(t_c) + \frac{u^2}{2}A'(t_c) + \frac{u^3}{6}A''(t_c) \right]_{-\Delta t/2}^{\Delta t/2} + \mathcal{O}((\Delta t)^5) \\
    &= \Delta t A(t_c) + 0 + \left(\frac{(\Delta t)^3}{48} - \frac{-(\Delta t)^3}{48}\right)A''(t_c) + \mathcal{O}((\Delta t)^5) \\
    &= \Delta t A(t_c) + \frac{(\Delta t)^3}{24}A''(t_c) + \mathcal{O}((\Delta t)^5).
\end{align*}
where \( u=s-t_c \). Thus, the error from approximating \( \Omega_{j,1}(\Delta t) \) by \( \Delta t A(t_c) \) is:
\begin{equation}
    \Omega_{j,1}(\Delta t) - \Delta t A(t_c) = \frac{(\Delta t)^3}{24}A''(t_c) + \mathcal{O}((\Delta t)^5). \label{eq:Omega1_error}
\end{equation}

\vspace{0.25 cm}
\noindent
\textbf{Error source 2:} Here we evaluate the leading term of the error arise from higher order truncation. Given that 
\( \Omega_{j,2}(\Delta t) = \frac{1}{2} \int_{t_j}^{t_{j+1}} ds_1 \int_{t_j}^{s_1} ds_2\, [A(s_1), A(s_2)] \). To find its leading contribution in powers of \( \Delta t \) (which is \( \mathcal{O}((\Delta t)^3) \)), we approximate \( A(s) \) using its first-order Taylor expansion around \( t_c \). Thus, 
\[
A(s) \approx A(t_c) + (s-t_c)A'(t_c).
\]
Substituting this into the commutator, we have 
\begin{align*}
    [A(s_1), A(s_2)] &\approx [A(t_c) + (s_1-t_c)A'(t_c), A(t_c) + (s_2-t_c)A'(t_c)] \\
    &= [A(t_c), A(t_c)] + [A(t_c), (s_2-t_c)A'(t_c)] \\
    &\quad + [(s_1-t_c)A'(t_c), A(t_c)] + [(s_1-t_c)A'(t_c), (s_2-t_c)A'(t_c)] \\
    &= 0 + (s_2-t_c)[A(t_c), A'(t_c)] + (s_1-t_c)[A'(t_c), A(t_c)] \\
    &\quad + (s_1-t_c)(s_2-t_c)[A'(t_c), A'(t_c)] \\
    &= (s_2-t_c)[A(t_c), A'(t_c)] - (s_1-t_c)[A(t_c), A'(t_c)] \\
    &= (s_2-s_1)[A(t_c), A'(t_c)].
\end{align*}
Note that the term \( (s_2-s_1)[A(t_c), A'(t_c)] \) is \( \mathcal{O}(\Delta t) \), and hence integrating this over the \( \mathcal{O}((\Delta t)^2) \) domain \( ds_1 ds_2 \) gives the \( \mathcal{O}((\Delta t)^3) \) contribution. More explicitly, let \( s_1' = s_1 - t_c \) and \( s_2' = s_2 - t_c \), we have
\begin{equation}
    \int_{t_j}^{t_{j+1}} ds_1 \int_{t_j}^{s_1} ds_2 (s_2-s_1) = \int_{-\Delta t/2}^{\Delta t/2} ds_1' \int_{-\Delta t/2}^{s_1'} ds_2' (s_2'-s_1') = -\frac{(\Delta t)^3}{6}
\end{equation}
Therefore, the leading term of \( \Omega_{j,2}(\Delta t) \) is:
\begin{equation}
    \Omega_{j,2}(\Delta t) = \frac{1}{2} \left(-\frac{(\Delta t)^3}{6}\right) [A(t_c), A'(t_c)] + \mathcal{O}((\Delta t)^4) = \frac{(\Delta t)^3}{12}[A'(t_c), A(t_c)] + \mathcal{O}((\Delta t)^4). \label{eq:Omega2_contrib}
\end{equation}

\vspace{0.25 cm}
\noindent
\textbf{Total error:}
Combining the dominant contributions from Equations \eqref{eq:Omega1_error} and \eqref{eq:Omega2_contrib}, we get 
\begin{align*}
\Delta\Omega_j(\Delta t) &=  \Omega_{\text{exact}, j}(\Delta t) - \Omega_{\text{approx}, j}(\Delta t) \\ 
&=   \left( \Delta t A(t_c) + \frac{(\Delta t)^3}{24}A''(t_c) + \mathcal{O}(\Delta t^5) \right)  + \left( \frac{\Delta t^3}{12}[A'(t_c), A(t_c)] + \mathcal{O}(\Delta t^4) \right) \\
&\quad + \left( \Omega_{3,j}(\Delta t) + \Omega_{4,j}(\Delta t) + \cdots \right) - \Delta t A(t_c)
\end{align*}
For symmetric integrators like the midpoint rule based method \( \exp(\Delta t A(t_c)) \), the local error in the exponent has an expansion in odd powers of \( \Delta t \) for its leading terms \cite{Blanes2008TheME, Scott2014Numerical}. This implies that the coefficient of the \( (\Delta t)^4 \) term in the full expansion of \( \Delta\Omega_j(\Delta t) \) (arising from the \( \mathcal{O}((\Delta t)^4) \) part of \( \Omega_{j,2} \), the parts of \( \Omega_{j,3} \) that might contribute to \( \mathcal{O}((\Delta t)^4) \), and \( \Omega_{j,4} \)) vanishes due to cancellations.
Thus, the error in the exponent for the step is precisely
\begin{equation}
\Delta\Omega_j(\Delta t) = \frac{(\Delta t)^3}{24}A''(t_c) + \frac{(\Delta t)^3}{12}[A'(t_c), A(t_c)] + \mathcal{O}((\Delta t)^5).
\end{equation}

The local error in the propagator for the \( j \)-th step is \( \left\| U_j - \tilde{U}_j \right\| \). Since \( H(t) \) is Hermitian, \( A(s)=-iH(s) \) is skew-Hermitian. For skew-Hermitian operators \( X, Y \), we have \( \|e^X - e^Y\| \leq \|X-Y\| \). Therefore, 
\begin{align*}
\left\| U_j - \tilde{U}_j \right\| &\leq \left\| \Delta\Omega_j(\Delta t) \right\| \\
&\leq \left\| \frac{(\Delta t)^3}{24}A''(t_c) + \frac{(\Delta t)^3}{12}[A'(t_c), A(t_c)] \right\| + \mathcal{O}((\Delta t)^5) \\
&\leq (\Delta t)^3 \left( \frac{1}{24}\left\|A''(t_c)\right\| + \frac{1}{12}\left\|[A'(t_c), A(t_c)]\right\| \right) + \mathcal{O}((\Delta t)^5).
\end{align*}
Since \( A(s) = -iH(s) \), we have \( A'(s)=-iH'(s) \) and \( A''(s)=-iH''(s) \). Hence, 
\[
\left\|A''(t_c)\right\| = \left\|H''(t_c)\right\|, \quad \left\|[A'(t_c), A(t_c)]\right\| = \left\|[H'(t_c), H(t_c)]\right\|.
\]
Therefore,
\begin{equation}
  \left\| U_j - \tilde{U}_j \right\| \leq (\Delta t)^3 \left( \frac{1}{24}\left\|H''(t_c)\right\| + \frac{1}{12}\left\|[H'(t_c), H(t_c)]\right\| \right) + \mathcal{O}((\Delta t)^5).  
\end{equation}
The global error is bounded by the sum of the local errors. Hence, 
\begin{align*}
\left\| \mathcal{T} \exp\left(-i \int_0^t H(s) \, ds \right) - \prod_{j=0}^{L-1} e^{-i \Delta t H(t_j + \Delta t /2)} \right\| &\leq \sum_{j=0}^{L-1} \left\| U_j - \tilde{U}_j \right\| \\
&\leq \sum_{j=0}^{L-1} (\Delta t)^3 \left( \frac{1}{24}\left\|H''(t_{c,j})\right\| + \frac{1}{12}\left\|[H'(t_{c,j}), H(t_{c,j})]\right\| \right) \\
&\leq L \cdot (\Delta t)^3 \left( \frac{1}{24}\left\|H''(t)\right\|_{\infty, \infty} + \frac{1}{12}\left\|[H'(t), H(t)]\right\|_{\infty, \infty} \right) \\
&= L \cdot \left(\frac{t}{L}\right)^3 \left( \frac{1}{24}\left\|H''(t)\right\|_{\infty, \infty} + \frac{1}{12}\left\|[H'(t), H(t)]\right\|_{\infty, \infty} \right) \\
&= \frac{t^3}{L^2} \left( \frac{1}{24}\left\|H''(t)\right\|_{\infty, \infty} + \frac{1}{12}\left\|[H'(t), H(t)]\right\|_{\infty, \infty} \right).
\end{align*}
\end{proof}

\noindent
\textbf{Lemma~\ref{lem:components_for_time_discretisation_error_bounds} (restated).} \textit{
Consider $H_I(t)$ as described in \cref{eq:HI_explicit}, and let $n = 2^k - 1$ be the photon truncation level per mode. Let $\| \cdot \|_{\infty, \infty}$ denote the supremum over $t \in [0, T]$. The norms required for the time discretization error bounds of $H_I(t)$ are:
\begin{align}
    \| H_I'(t)\|_{\infty, \infty} &\leq 2\sqrt{n} \sum_{m=1}^{N_F} |\gamma_m (\omega_m - \omega)| \\
    \| H_I''(t) \|_{\infty, \infty} &\leq 2\sqrt{n} \sum_{m=1}^{N_F} |\gamma_m| (\omega_m - \omega)^2 \\
    \|[H_I'(t), H_I(t)]\|_{\infty, \infty} &\leq 8n \left( \sum_{m=1}^{N_F} |\gamma_m (\omega_m - \omega)| \right) \left( \sum_{m=1}^{N_F} |\gamma_m| \right)
\end{align}
}

\begin{proof}
From \cref{eq:HI_explicit}, we have:
\begin{equation}
\label{supp eq: H_interaction}
H_I(t) = \sum_{m=1}^{N_F} \gamma_{m}
\left[
 e^{i(\omega_m - \omega)t} \left( \bigotimes_{l=1}^{N_F} A_l^{(m)\dagger} \right) \otimes \sigma^- + e^{-i(\omega_m - \omega)t} \left( \bigotimes_{l=1}^{N_F} A_l^{(m)} \right) \otimes \sigma^+
\right],
\end{equation}
where $A_l^{(m)\dagger}$ and $A_l^{(m)}$ are creation and annihilation operators for the $m$-th photon mode (acting as identity on other modes $l \neq m$), and $\sigma^\pm$ are atomic ladder operators.

Differentiating $H_I(t)$ with respect to $t$ gives:
\begin{equation}
\label{eq: derivative of H}
H_I'(t) = \sum_{m=1}^{N_F} i\gamma_{m}(\omega_m - \omega)
\left( e^{i(\omega_m - \omega)t} \left( \bigotimes_{l=1}^{N_F} A_l^{(m)\dagger} \right) \otimes \sigma^- - e^{-i(\omega_m - \omega)t} \left( \bigotimes_{l=1}^{N_F} A_l^{(m)} \right) \otimes \sigma^+ \right).
\end{equation}
To bound its norm, we apply the triangle inequality and note that $\|\hat{a}\|=\|\hat{a}^\dagger\|=\sqrt{n}$ and $\|\sigma^\pm\|=1$:
\begin{align}
\label{eq: bound on derivative of H}
\|H_I'(t)\|_{\infty, \infty}
&\leq \sup_t \left\| \sum_{m=1}^{N_F} i\gamma_{m}(\omega_m - \omega)
\left( e^{i(\omega_m - \omega)t} A_m^{\dagger} \otimes \sigma^- - e^{-i(\omega_m - \omega)t} A_m \otimes \sigma^+ \right) \right\| \notag \\
&\leq \sum_{m=1}^{N_F} |\gamma_m (\omega_m - \omega)|
\left(
\left\| A_m^{\dagger} \otimes \sigma^- \right\|
+
\left\| A_m \otimes \sigma^+ \right\|
\right) \notag \\
&= \sum_{m=1}^{N_F} |\gamma_m (\omega_m - \omega)| (\sqrt{n} + \sqrt{n}) \notag \\
&= 2\sqrt{n} \sum_{m=1}^{N_F} |\gamma_m (\omega_m - \omega)|.
\end{align}
This establishes the bound for $\|H_I'(t)\|_{\infty, \infty}$. A second differentiation and similar analysis establishes the bound for $\|H_I''(t)\|_{\infty, \infty}$:
\begin{equation}
    \|H_I''(t)\|_{\infty, \infty} \leq 2\sqrt{n} \sum_{m=1}^{N_F} |\gamma_m| (\omega_m - \omega)^2.
\end{equation}

To bound the norm of the commutator $\|[H_I'(t), H_I(t)]\|_{\infty, \infty}$, we use the property $\|[X,Y]\| \leq 2\|X\|\|Y\|$. This requires bounds on both $\|H_I'(t)\|$ and $\|H_I(t)\|$. The bound for $\|H_I'(t)\|$ is given in Equation~\eqref{eq: bound on derivative of H}.
The norm of $H_I(t)$ can be bounded similarly:
\begin{align}
\label{eq: sup bound on H_I(t)}
\|H_I(t)\|_{\infty, \infty} &\leq \sum_{m=1}^{N_F} |\gamma_m| \left( \|A_m^{\dagger} \otimes \sigma^-\| + \|A_m \otimes \sigma^+\| \right) \notag \\
&\leq \sum_{m=1}^{N_F} |\gamma_m| (\sqrt{n} + \sqrt{n}) \notag \\
&= 2\sqrt{n} \sum_{m=1}^{N_F} |\gamma_m|.
\end{align}
Combining these results, we find the bound for the commutator:
\begin{align}
    \label{eq: implicit bound on [H_I'(t), H_I(t)]}
    \|[H_I'(t), H_I(t)]\|_{\infty, \infty} &\leq 2 \|H_I'(t)\|_{\infty, \infty} \|H_I(t)\|_{\infty, \infty} \notag \\
    &\leq 2 \left( 2\sqrt{n} \sum_{m=1}^{N_F} |\gamma_m (\omega_m - \omega)| \right) \cdot \left( 2\sqrt{n} \sum_{m=1}^{N_F} |\gamma_m| \right) \notag \\
    &= 8n \left( \sum_{m=1}^{N_F} |\gamma_m (\omega_m - \omega)| \right) \left( \sum_{m=1}^{N_F} |\gamma_m| \right).
\end{align}

\end{proof}


\section{Quantum Circuit Construction and Structure}
\label{sec_appendix:explicit_circuit_construction}

As shown in Section~\ref{sec:Mapping_to_qubit}, the eJCM Hamiltonian can be decomposed into such 1-sparse Hermitian operators—specifically, into sums of Pauli strings $P \in \{I,X,Y,Z\}^{\otimes N}$. Thus, the circuit in Figure~\ref{fig: Time-Evolution First Order Trotter Half Circuit} can be constructed from primitive circuit elements that implement $e^{-i \theta P}$ for an arbitrary Pauli string $P$. 
This can be done by first noting that for a Pauli string of $Z$-type, 
$\Lambda = \{I, Z \}^{\otimes N}$, we can write it as 
$\Lambda = Z_{q_{i_1}} Z_{q_{i_2}} \cdots Z_{q_{i_k}}$, 
where $\{q_{i_1}, q_{i_2}, \ldots, q_{i_k}\}$ is the ordered subset of qubits on which the Pauli $Z$ operator acts, such that $i_1 < i_2 < \cdots < i_k$. Then we have
\begin{equation}
\label{eq: exponential of Pauli}
\begin{aligned}
e^{-i \theta \Lambda} &=   e^{-i \theta Z_{q_{i_1}} Z_{q_{i_2}} \cdots Z_{q_{i_k}}} \\
&= \bigg( C^X_{q_{i_1}, q_{i_2}} 
C^X_{q_{i_2}, q_{i_3}} 
\cdots 
C^X_{q_{i_{k-1}}, q_{i_k}} \bigg) \cdot 
e^{-i \theta  Z_{q_{i_k}}} \cdot 
\bigg( 
C^X_{q_{i_{k-1}}, q_{i_k}} 
C^X_{q_{i_{k-2}}, q_{i_{k-1}}} 
\cdots 
C^X_{q_{i_1}, q_{i_2}} 
\bigg) \\
&= \mathcal{P}_N\, e^{-i \theta Z_{q_{i_k}}} \, \mathcal{P}_N^\dagger,
\end{aligned}
\end{equation}
where $C^X_{i,j}$ denotes a CNOT (CX) gate with control qubit $i$ and target qubit $j$, and we have defined $\mathcal{P}_N = C^X_{q_{i_1}, q_{i_2}}  C^X_{q_{i_2}, q_{i_3}}  \cdots  C^X_{q_{i_{k-1}}, q_{i_k}}$ to simplify the expression.

\vspace{.25cm}
To generalize this to arbitrary Pauli strings $P = \bigotimes_{i=1}^N p_i$ where $p_i \in \{ I, X, Y, Z \}$, we first transform each $p_i$ into the $Z$ basis using
\[ 
HXH = Z, \hspace{1cm} (HS)Y(S^\dagger H) = R_X(\pi/2) Y R_X(-\pi/2) = Z
\]
where $H$ is the Hadamard gate, and $R_X(\theta)$ is a single-qubit rotation about the X-axis. Thus, each Pauli string can be expressed as
\[
P = \bigotimes_{i=1}^N P_i = \bigotimes_{i=1}^N R_i Z_i R_i^\dagger, \quad \text{where } R_i =
\begin{cases}
H & \text{if } p_i = X, \\
R_X(\pi/2) & \text{if }  p_i = Y, \\
I & \text{if } p_i = Z.
\end{cases}
\]
in a procedure also known as diagonalization. With the decomposition established, we can now explicitly construct the quantum circuit for implementing $e^{-i \theta P}$ for an arbitrary $N$-qubit Pauli string. This construction is illustrated in Figure~\ref{fig: Exp of Pauli Circuit}.

\begin{figure}[ht]
\begin{center}
\begin{adjustbox}{scale=1}
\begin{quantikz}
\lstick{$q_0$}      & \gate[1, style={fill=orange!20, rounded corners, minimum width=1.1cm, minimum height=0.6cm} ]{R_1}   & \gate[4, style={fill=green!20, rounded corners, minimum width= 1 cm, minimum height=0.6cm} ]{\mathcal{P}_N} & \qw   & \gate[4, style={fill=green!20, rounded corners, minimum width= 1 cm, minimum height=0.6cm} ]{\mathcal{P}_N^\dagger} &  \gate[1, style={fill=orange!20, rounded corners, minimum width=1.1cm, minimum height=0.6cm} ]{R_1^\dagger}  & \qw \\
\lstick{$q_1$}      & \gate[1, style={fill=orange!20, rounded corners, minimum width=1.1cm, minimum height=0.6cm} ]{R_2}    & \ghost[1 ]{P}     & \qw &  \ghost[1 ]{P}  & \gate[1, style={fill=orange!20, rounded corners, minimum width=1.1cm, minimum height=0.6cm} ]{R_2^\dagger}  & \qw \\
\lstick{$\vdots$}  \\
\lstick{$q_{N-1}$}  & \gate[1, style={fill=orange!20, rounded corners, minimum width=1.1cm, minimum height=0.6cm} ]{R_{N}} &   \ghost[1 ]{P}   & \gate[1, style={fill=purple!20, rounded corners, minimum width=1.1cm, minimum height=0.6cm} ]{R_z(2 \theta)}  & \qw  & \gate[1, style={fill=orange!20, rounded corners, minimum width=1.1cm, minimum height=0.6cm} ]{R_{N}^\dagger}  & \qw
\end{quantikz}
\end{adjustbox}
\caption{
\centering{Quantum circuit to prepare $e^{ -i P \theta} $, where $P = \bigotimes_{i=1}^N p_i$ is an arbitrary N-qubit Pauli string with $p_i \in \{ I, X, Y, Z\}$. Each gate $R_i \in \{I, H, R_X(\pi/2) \}$ and $\mathcal{P}_N = C^X_{q_{i_1}, q_{i_2}}  C^X_{q_{i_2}, q_{i_3}}  \cdots  C^X_{q_{i_{N-1}}, q_{i_N}}$ where $q_{i_j}$ is the qubit of the $j^{\text{th}}$ Pauli operator. }}
\label{fig: Exp of Pauli Circuit}
\end{center}
\end{figure}
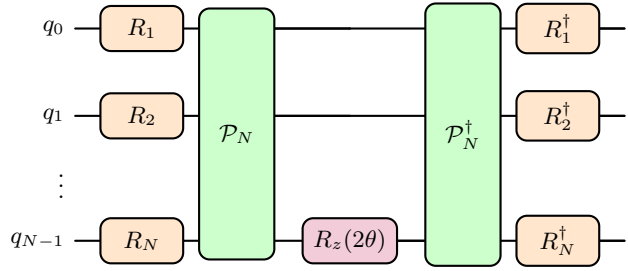

\vspace{.5 cm}

\begin{lemma}
\label{thm: exponential of sum of commuting terms}
Let \( G = \sum_{i=1}^k \alpha_i P_i \), where each \( \alpha_i \in \mathbb{R} \) and \( P_i \) are Pauli strings such that \( [P_i, P_j] = 0 \) for all \( i, j \). Then,
\[
e^{-iG} = e^{-i \sum_{i=1}^k \alpha_i P_i} = \prod_{i=1}^k e^{-i \alpha_i P_i}.
\]
\end{lemma}

\begin{proof}
We first consider the case of two commuting Hermitian operators, \( H_1 \) and \( H_2 \), such that \( [H_1, H_2] = 0 \). Using the Baker–Campbell–Hausdorff (BCH) formula, if \( [H_1, H_2] = 0 \), then
\begin{equation}
    e^{H_1 + H_2} = e^{H_1} e^{H_2}.
\end{equation}
More explicitly, we can verify this via the Taylor expansion:
\begin{align}
e^{H_1 + H_2} 
&= \sum_{n=0}^\infty \frac{1}{n!}(H_1 + H_2)^n \\
&= \sum_{n=0}^\infty \frac{1}{n!} \sum_{m=0}^n \binom{n}{m} H_1^{n-m} H_2^m 
\hspace{1cm}  \\
&= \sum_{n=0}^\infty \sum_{m=0}^n \frac{H_1^{n-m}}{(n-m)!} \frac{H_2^m}{m!}   \\
&= \sum_{p=0}^\infty \sum_{q=0}^\infty \frac{H_1^p}{p!} \frac{H_2^q}{q!} 
\hspace{1cm} \text{(let \( p = n - m \), \( q = m \), so \( n = p + q \))} \\
&= \left( \sum_{p=0}^\infty \frac{H_1^p}{p!} \right) 
   \left( \sum_{q=0}^\infty \frac{H_2^q}{q!} \right) 
\end{align}
We now proceed by induction. Suppose the result holds for \( k \) commuting operators, that is,
\begin{equation}
    e^{\sum_{i=1}^k H_i} = \prod_{i=1}^k e^{H_i}.
\end{equation}
Now consider \( k+1 \) commuting operators \( H_1, H_2, \dots, H_k, H_{k+1} \), all mutually commuting. Let \( H = \sum_{i=1}^k H_i \), then by the inductive hypothesis,
\begin{equation}
    e^{\sum_{i=1}^{k+1} H_i} = e^{H + H_{k+1}} = e^H e^{H_{k+1}} = \left( \prod_{i=1}^k e^{H_i} \right) e^{H_{k+1}} = \prod_{i=1}^{k+1} e^{H_i}.
\end{equation}
Since all Pauli strings \( P_i \) are Hermitian and satisfy \( P_i^2 = I \), and the coefficients \( \alpha_i \) are real, each \( \alpha_i P_i \) is Hermitian, and the same result holds for the exponential of their sum.
Thus, for commuting Pauli strings \( P_i \), we have
\begin{equation}
    e^{-i \sum_{i=1}^k \alpha_i P_i} = \prod_{i=1}^k e^{-i \alpha_i P_i}.
\end{equation}

\end{proof}

\begin{lemma}
\label{theorem: controlled exp Paulis}
Let \( U_\rho \) be a unitary operator of the form
\[
U_\rho = A B A^\dagger,
\]
where \(A\) and \(B\) are unitary operators, then the controlled version of \( U_\rho \) can be implemented as
\[
\text{controlled-}U_\rho = (I \otimes A) \cdot (\text{controlled-}B) \cdot (I \otimes A^\dagger),
\]
where only \( B \) is controlled.
\end{lemma}

\begin{proof}
By definition, the controlled-\( U_\rho \) gate acts as
\begin{equation}
    \text{controlled-}U_\rho = |0\rangle\langle 0| \otimes I + |1\rangle\langle 1| \otimes U_\rho.
\end{equation}
Substituting \( U_\rho = A B A^\dagger \), we have
\begin{equation}
    \text{controlled-}U_\rho = |0\rangle\langle 0| \otimes I + |1\rangle\langle 1| \otimes (A B A^\dagger).
\end{equation}
We can factor out \( A \) and \( A^\dagger \) as follows:
\begin{equation}
    (I \otimes A) \left( |0\rangle\langle 0| \otimes A^\dagger A + |1\rangle\langle 1| \otimes B \right) (I \otimes A^\dagger).
\end{equation}
Since \( A^\dagger A = I \), this simplifies to
\begin{equation}
    (I \otimes A) \left( |0\rangle\langle 0| \otimes I + |1\rangle\langle 1| \otimes B \right) (I \otimes A^\dagger).
\end{equation}
Recognising that
\begin{equation}
    |0\rangle\langle 0| \otimes I + |1\rangle\langle 1| \otimes B
\end{equation}
is precisely the controlled-\( B \) operation, we conclude
\begin{equation}
    \text{controlled-}U_\rho = (I \otimes A) \cdot (\text{controlled-}B) \cdot (I \otimes A^\dagger).
\end{equation}

\end{proof}

\end{document}